\documentclass[11pt]{article}
\usepackage{fullpage}
\usepackage{amsmath,amsfonts,amsthm,amssymb}

\usepackage{url}
\usepackage{color}
\usepackage[usenames,dvipsnames,svgnames,table]{xcolor}
\usepackage[colorlinks=true, linkcolor=red, urlcolor=blue, citecolor=gray]{hyperref}
\usepackage{epsfig}
\usepackage{algorithm}
\usepackage{algorithmic}

\usepackage{tikz}
\usetikzlibrary{arrows}

\newcommand{\plog}{\mathop\mathrm{polylog}}

\DeclareMathOperator{\tr}{tr}

\DeclareMathOperator{\fl}{fl}

\newcommand{\R}{\mathbb{R}}

\DeclareMathOperator*{\barA}{\bv{\bar A}}
\DeclareMathOperator*{\barT}{\bv{\bar T}}
\DeclareMathOperator*{\alga}{\mathcal{A}}
\DeclareMathOperator*{\argmin}{arg\,min}

\DeclareMathOperator*{\mach}{\epsilon_{mach}}
\DeclareMathOperator*{\gmach}{\gamma_{mach}}
\DeclareMathOperator*{\lmin}{\lambda_{min}}
\DeclareMathOperator*{\lmax}{\lambda_{max}}

\DeclareMathOperator*{\smax}{\sigma_{max}}
\DeclareMathOperator*{\rmin}{r_{min}}
\DeclareMathOperator*{\rmax}{r_{max}}
\newcommand{\poly}{\mathop\mathrm{poly}}
\newcommand{\diag}{\mathop\mathrm{diag}}
\newcommand{\wh}{\widehat}

\newcommand{\ol}{\overline}
\DeclareMathOperator{\nnz}{nnz}
\DeclareMathOperator{\mv}{mv}

\newcommand{\eqdef}{\mathbin{\stackrel{\rm def}{=}}}
\newcommand{\norm}[1]{\|#1\|}

\newcommand{\bs}[1]{\boldsymbol{#1}}
\newcommand{\bv}[1]{\mathbf{#1}}
\usepackage{etoolbox}

\newcommand\blfootnote[1]{%
	\begingroup
	\renewcommand\thefootnote{}\footnote{#1}%
	\addtocounter{footnote}{-1}%
	\endgroup
}

\makeatletter
\def\blfootnote{\gdef\@thefnmark{}\@footnotetext}
\makeatother

\newtheorem{theorem}{Theorem}
\newtheorem{lemma}[theorem]{Lemma}
\newtheorem{corollary}[theorem]{Corollary}

\newtheorem{claim}[theorem]{Claim}
\newtheorem{definition}{Definition}
\newtheorem{requirement}{Requirement}

\newtheorem*{rep@theorem}{\rep@title}
\newcommand{\newreptheorem}[2]{%
\newenvironment{rep#1}[1]{%
 \def\rep@title{#2 \ref{##1}}%
 \begin{rep@theorem}}%
 {\end{rep@theorem}}}
\makeatother
\newreptheorem{theorem}{Theorem}
\newreptheorem{claim}{Claim}

  \usepackage{nth}
  \usepackage{intcalc}

\title{Stability of the Lanczos Method \\ for Matrix Function Approximation}

\author{
Cameron Musco\\MIT\\ \texttt{cnmusco@mit.edu}
\and
Christopher Musco\\MIT\\ \texttt{cpmusco@mit.edu}
\and
Aaron Sidford\\Stanford\\ \texttt{sidford@stanford.edu}
}

\date{August 28, 2017}

\begin{document}
\maketitle

\begin{abstract}
Theoretically elegant and ubiquitous in practice, the Lanczos method can approximate $f(\bv{A})\bv{x}$ for any symmetric matrix $\bv{A} \in \R^{n\times n}$, vector $\bv{x}\in \R^n$, and function $f$. In \emph{exact arithmetic}, the method's error after $k$ iterations is bounded by the error of the best degree-$k$ polynomial uniformly approximating the scalar function $f(x)$ on the range $[\lmin(\bv{A}), \lmax(\bv{A})]$. However, despite decades of work, it has been unclear if this powerful guarantee holds in finite precision.

 We resolve this problem, proving that when $\max_{x \in [\lmin, \lmax]}|f(x)| \le C$, Lanczos essentially matches the exact arithmetic guarantee if computations use roughly $\log(nC\norm{\bv{A}})$ bits of precision. Our proof extends work of Druskin and Knizhnerman \cite{druskinRussian}, leveraging the stability of the classic Chebyshev recurrence to bound the stability of any polynomial approximating $f(x)$.

We also study the special case of $f(\bv{A}) = \bv{A}^{-1}$ for positive definite $\bv{A}$, where stronger guarantees hold for Lanczos. In exact arithmetic the algorithm performs as well as the best polynomial approximating $1/x$ \emph{at each of $\bv{A}$'s eigenvalues}, rather than on the full range $[\lmin(\bv{A}), \lmax(\bv{A})]$. In seminal work, Greenbaum gives a natural approach to extending this bound to finite precision: she proves that finite precision Lanczos and the related conjugate gradient method match any polynomial approximating $1/x$ \emph{in a tiny range around each eigenvalue} \cite{greenbaum1989behavior}.

For $\bv{A}^{-1}$, Greenbaum's bound appears stronger than our result. However, we exhibit matrices with condition number $\kappa$ where exact arithmetic Lanczos converges in $\plog(\kappa)$ iterations, but Greenbaum's bound predicts at best $\Omega(\kappa^{1/5})$ iterations in finite precision. It thus cannot offer more than a polynomial improvement over the $O(\kappa^{1/2})$ bound achievable via our result for general $f(\bv{A})$. Our analysis bounds the power of \emph{stable approximating polynomials}
and raises the question of if they fully characterize the behavior of finite precision Lanczos in solving linear systems. If they do, convergence in less than $\poly(\kappa)$ iterations cannot be expected, even for matrices with clustered, skewed, or otherwise favorable eigenvalue distributions.
\end{abstract}

\blfootnote{This document was edited on November 18th, 2024 to restate Lemmas 9,10, and 11 in a slightly tighter way.
 }

\thispagestyle{empty}
\clearpage
\setcounter{page}{1}

\section{Introduction}
The Lanczos method for iteratively tridiagonalizing a Hermitian matrix is one of the most important algorithms in numerical computation. Introduced for computing eigenvectors and eigenvalues \cite{lanczos1950iteration}, it remains the standard algorithm for doing so over half a century later \cite{saad2011numerical}. It also underlies state-of-the-art iterative solvers for linear systems
\cite{hestenes1952methods,saad2003iterative}. 

More generally, the Lanczos method can be used to iteratively approximate any function of a matrix's eigenvalues. Specifically, given $f:\R \mapsto \R$, symmetric $\bv{A} \in \mathbb{R}^{n \times n}$ with  eigendecomposition $\bv{V} \bs{\Lambda} \bv{V}^T$, and vector $\bv{x} \in \R^n$, it approximates $f(\bv{A})\bv{x}$, where: 
\begin{align*}
f(\bv{A}) \eqdef \bv{V} f(\bs{\Lambda}) \bv{V}^T.
\end{align*}
$f(\bs{\Lambda})$ is the result of applying $f$ to each diagonal entry of $\bs{\Lambda}$, i.e., to the eigenvalues of $\bv{A}$.
In the special case of linear systems, $f(x) = 1/x$ and $f(\bv{A}) = \bv{A}^{-1}$.
Other important matrix functions include the matrix log, the matrix exponential, the matrix sign function, and the matrix square root \cite{highamBook}. These functions are broadly applicable in scientific computing, and are increasingly used in theoretical computer science \cite{arora2007combinatorial,matrixExp,sachdeva2014faster} and machine learning \cite{han2015large,frostig2016principal,ubaru2016fast,allen2016faster,tremblay2016compressive}. In theses areas, there is interest in obtaining worst-case, end-to-end runtime bounds for approximating $f(\bv{A})\bv{x}$ up to a given precision.

The main idea behind the Lanczos method is to iteratively compute an orthonormal basis $\bv{Q}$ for the rank-$k$ \emph{Krylov subspace} $\mathcal{K}_k = [\bv{x}, \bv{A}\bv{x}, \bv{A}^2\bv{x}, \ldots, \bv{A}^{k-1}\bv{x}]$. The method then approximates $f(\bv{A})\bv{x}$ with a vector in $\mathcal{K}_k$ -- i.e. with $p(\bv{A})\bv{x}$ for some polynomial $p$ with degree $< k$.

Specifically, along with $\bv{Q}$, the algorithm computes $\bv{T} = \bv{Q}^T \bv{A} \bv{Q}$ and approximates $f(\bv{A}) \bv{x}$ with $\bv{y} = \norm{\bv{x}} \cdot \bv{Q} f(\bv{T}) \bv{e}_1$.\footnote{Here $\bv{e}_1$ is the first standard basis vector. There are a number of variations on the Lanczos method, especially for the case of solving linear systems, however we consider just this simple, general version. 
} Importantly, $\bv{y}$ can be computed efficiently: iteratively constructing $\bv{Q}$ and $\bv{T}$ requires just $k-1$ matrix-vector multiplications with $\bv{A}$. Furthermore, due to a special iterative construction, $\bv{T}$ is tridiagonal. It is thus possible to accurately compute its eigendecomposition, and hence apply arbitrary functions $f(\bv{T})$, including $\bv{T}^{-1}$, in $\tilde O(k^2)$ time. 

Note that $\bv{y} \in \mathcal{K}_k$ and so can be written as $p(\bv{A})\bv{x}$ for some polynomial $p$. While this is not necessarily the polynomial minimizing $\norm{p(\bv{A})\bv{x}- f(\bv{A})\bv{x}}$, for the Euclidean norm $\|\cdot \|$, $\bv{y}$ satisfies:

\begin{align}\label{introBound}
\norm{f(\bv{A})\bv{x} - \bv{y}} \le 2\norm{\bv{x}} \cdot
\min_{\substack{\text{polynomial $p$}\\ \text{with degree $<k$} }} \left(\max_{x \in [\lmin(\bv{A}),\lmax(\bv{A})]} \left | f(x) - p(x) \right| \right).
\end{align}
where $\lmax(\bv{A})$ and $\lmin(\bv{A})$ are the largest and smallest eigenvalues of $\bv{A}$ respectively. That is, up to a factor of 2, the error of Lanczos in approximating $f(\bv{A})\bv{x}$ is bounded by the uniform error of the best polynomial approximation to $f$ with degree $< k$.
Thus, to bound the performance of Lanczos after $k$ iterations, it suffices to prove the existence of any degree-$k$ polynomial approximating the scalar function $f$, even if the explicit polynomial is not known

\section{Our contributions}

Unfortunately, as has been understood since its introduction, the performance of the Lanczos algorithm in exact arithmetic does not  predict its behavior when implemented in finite precision. Specifically, it is well known that the basis $\bv{Q}$ loses orthogonality. This leads to slower convergence when computing eigenvectors and values, and a wide range of reorthogonalization techniques have been developed to remedy the issue (see e.g. \cite{parlett1979lanczos,simon1984lanczos} or \cite{parlett1998symmetric,meurant2006lanczos} for surveys).

However, in the case of matrix function approximation, these remedies appear unnecessary. Vanilla Lanczos continues to perform well in practice, despite loss of orthogonality. In fact, it even converges when $\bv{Q}$ has numerical rank $\ll k$ and thus cannot span $\mathcal{K}_k$. Understanding when and why the Lanczos algorithm runs efficiently in the face of numerical breakdown has been the subject of intensive research for decades -- we refer the reader to \cite{meurant2006lanczos} for a survey. Nevertheless, despite experimental and theoretical evidence, no iteration bounds comparable to the exact arithmetic guarantees were known for general matrix function approximation in finite precision.


\subsection{General function approximation in finite precision}
 Our main positive result closes this gap for general functions by showing that a bound nearly matching \eqref{introBound} holds even when Lanczos is implemented in finite precision.
In Section \ref{lanczos_in_finite_precision} we show:
\begin{theorem}[Function Approximation via Lanczos in Finite Arithmetic]\label{mainLemmaFullRuntime}
Given real symmetric $\bv{A}\in \R^{n \times n}$, $\bv{x} \in \R^{n}$, $\eta \le \norm{\bv{A}}$, $\epsilon \le 1$, and any function $f$ with $|f(x)| <C$ for $x \in [\lmin(\bv{A})-\eta,\lmax(\bv{A}) + \eta]$, let $B = \log \left (\frac{ nk \norm{\bv{A}}}{\epsilon\eta} \right )$. The Lanczos algorithm run on a floating point computer with $\Omega (B)$ bits of precision for $k$ iterations returns $\bv{y}$ satisfying:
\begin{align}\label{eq:finitePrecisionGuarantee1}
\norm{f(\bv{A}) \bv{x} - \bv{y}} \le (7k \cdot \delta_{k} + \epsilon C) \norm{\bv{x}}
\end{align}
where
\begin{align*}
\delta_{k} \eqdef \min_{\substack{\textnormal{polynomial $p$} \\ \textnormal{ with degree $<k$} }} \left ( \max_{x \in [\lmin(\bv{A})-\eta,\lmax(\bv{A})+\eta]} |p(x)-f(x)| \right ).
\end{align*}
If basic arithmetic operations on floating point numbers with $\Omega (B)$ bits of precision have runtime cost $O(1)$, the algorithm's runtime is $O(\mv(\bv{A})k + k^2B + kB^2)$, where $\mv(\bv{A})$ is the time required to multiply the matrix $\bv{A}$ with a vector.
\end{theorem}
The bound of  \eqref{eq:finitePrecisionGuarantee1} matches \eqref{introBound} up to an $O(k)$ factor along with a small $\epsilon C$ additive error term, which decreases exponentially in the bits of precision available.
 For typical functions, the degree of the best uniform approximating polynomial depends logarithmically on the desired accuracy. So the $O(k)$ factor equates to just a logarithmic increase in the degree of the approximating polynomial, and hence the number of iterations required for a given accuracy. The theorem requires a uniform approximation bound on the slightly extended range $[\lmin(\bv{A})-\eta,\lmax(\bv{A})+\eta]$, however in typical cases this has essentially no effect on the bounds obtainable.

In Section~\ref{sec:applications} we give several example applications of Theorem~\ref{mainLemmaFullRuntime} that illustrate these principles. We show how to stably approximate the matrix sign function, the matrix exponential, and the top singular value of a matrix. Our runtimes all either improve upon or match state-of-the-art runtimes, while holding rigorously under finite precision computation. They demonstrate the broad usefulness of the Lanczos method and our approximation guarantees for matrix functions.
%

\subsubsection{Techniques and comparison to prior work}
We begin with
the groundbreaking work of Paige \cite{paigeThesis,paige1976error,paige1980error}, which gives a number of results on the behavior of the Lanczos tridiagonalization process in finite arithmetic. Using Paige's bounds, we demonstrate that if $f(x)$ is a degree $< k$ Chebyshev polynomial of the first kind, Lanczos can apply it very accurately. This proof, which is the technical core of our error bound, leverages the well-understood stability of the recursive formula for computing Chebyshev polynomials \cite{clenshaw}, \emph{even though this formula is not explicitly used when applying Chebyshev polynomials via Lanczos}.

To extend this result to general functions, we first show that Lanczos will effectively apply the `low degree polynomial part' of $f(\bv{A})$, incurring error depending on the residual $\delta_k$ (see Lemma \ref{mainLemmaDetailed}). So we just need to show that this polynomial component can be applied stably. To do so, we appeal to our proof for the special case of Chebyshev polynomials via the following argument, which appears formally in the proof of Lemma \ref{poly_with_error}:
If $|f(x)| \le C$ on $[\lmin(\bv{A}),\lmax(\bv{A})]$, then the optimal degree $k$ polynomial approximating $f(x)$ on this range is bounded by $2C$ in absolute value since it must have uniform error $<C$, the error given by setting $p(x) = 0$. Since its magnitude is bounded, this polynomial has coefficients bounded by $O(C)$ when written in the Chebyshev basis. Accordingly, by linearity, Lanczos only incurs error $O(C)$ times greater than what is obtained when applying Chebyshev polynomials. This yields the additive error bound $\epsilon C$ in Theorem \ref{mainLemmaFullRuntime}, proving that, \emph{for any bounded function, Lanczos can apply the optimal approximating polynomial accurately}. 

Ultimately, our proof can be seen as a more careful application of the techniques of
Druskin and Knizhnerman \cite{druskinRussian,druskinreview}.
They also use the stability of Chebyshev polynomials to understand stability for more general functions, but give an error bound which depends on a coarse upper bound for $\delta_k$. Additionally, their work ignores stability issues that can arise when computing the final output $\bv{y} = \norm{\bv{x}} \cdot \bv{Q}f(\bv{T})\bv{e}_1$. We provide a complete analysis by showing that $\bv{y}$ can be computed stably whenever $f(x)$ is well approximated by a low degree polynomial, and hence give the first end-to-end runtime bound for Lanczos in finite arithmetic.

Our work is also similar to that of Orecchia, Sachdeva, and Vishnoi, who give accuracy bounds for a slower variant of Lanczos with re-orthogonalization that requires $\sim O(\mv(\bv{A})k + k^3)$ time, in contrast to the $\sim O(\mv(\bv{A})k + k^2)$ time required for our Theorem \ref{mainLemmaFullRuntime}  \cite{matrixExp}.
Furthermore, their results require a bound on the coefficients of the polynomial $p(x)$. Many optimal approximating polynomials, like the Chebyshev polynomials, have coefficients which are exponential in their degree. Accordingly, \cite{matrixExp} requires that the number of bits used to match such polynomials with Lanczos grows polynomially (rather than logarithmically) with the approximating degree. In fact, as shown in \cite{frostig2016principal}, any degree $k$ polynomial with coefficients bounded by $C$ can be well approximated by a polynomial with degree $O(\sqrt{k \log (kC)})$. So \cite{matrixExp} only gives good bounds for polynomials that are inherently suboptimal. Additionally, like Druskin and Knizhnerman, \cite{matrixExp} only addresses roundoff errors that arise during matrix vector multiplication with $\bv{A}$, assuming stability for other components of their algorithm.

%
%
\subsection{Linear systems in finite precision}
\label{intro_linear_systems}
Theorem \ref{mainLemmaFullRuntime} shows that for general functions, the Lanczos method performs nearly as accurately in finite precision as in exact arithmetic: after $k$ iterations, it still nearly matches the accuracy of the best degree $< k$ uniform polynomial approximation to $f(x)$ over $\bv{A}$'s eigenvalue range. 

However, in
%
the important special case of solving positive definite linear systems, i.e., when $\bv{A}$ has all positive eigenvalues and $f(\bv{A}) = \bv{A}^{-1}$, it is well known that \eqref{introBound} can be strengthened in exact arithmetic. Lanczos performs as well as the best polynomial approximating $f(x) = 1/x$ \emph{at each of $\bv{A}$'s eigenvalues} rather than over the full range $[\lmin(\bv{A}),\lmax(\bv{A})]$. Specifically,\footnote{Note that slightly stronger bounds where $p$ depends on $\bv{x}$ are available. We work with \eqref{introBound2} for simplicity since it only depends on $\bv{A}$'s eigenvalues.}
\begin{align}\label{introBound2}
\norm{\bv{A}^{-1}\bv{x} - \bv{y}_k} \le \sqrt{\kappa(\bv{A})}\cdot \norm{\bv{x}} \cdot\min_{\substack{\text{polynomial $p$}\\ \text{with degree $<k$}}} \max_{x \in \{\lambda_{1}(\bv{A}),\lambda_2(\bv{A}),...,\lambda_{n}(\bv{A})\}} \left | p(x) -1/x\right|
\end{align}
where $\kappa(\bv{A}) = \norm{\bv{A}} \norm{\bv{A}^{-1}}$ is $\bv{A}$'s condition number.
\eqref{introBound2} is proven in Appendix \ref{sec:linsystems_app}. It can be much stronger than \eqref{introBound}, and correspondingly Theorem \ref{mainLemmaFullRuntime}. Specifically, the best bound obtainable from \eqref{introBound} is that after $\tilde O(\sqrt{\kappa(\bv{A})})$ iterations, $\bv{y} \approx \norm{\bv{A}^{-1}\bv{x} }$. In contrast, \eqref{introBound2} shows that even when $\kappa(\bv{A})$ is very large, 
$n$ iterations are enough to compute $\bv{A}^{-1} \bv{x}$ exactly: $p(x)$ can be set to the polynomial which exactly interpolates $1/x$ at each of $\bv{A}$'s eigenvalues. \eqref{introBound2} also gives improved bounds for matrices with clustered, skewed, or otherwise favorable eigenvalue distributions \cite{axelsson1986rate,harveyPaper}. For example, assuming exact arithmetic, it can be used to analyze  preconditioners for graph Laplacians, which induce heavily skewed eigenvalue distributions \cite{spielman2009note,pengExperiment}. It can also be applied to algorithms for solving asymmetric Laplacian systems corresponding to \emph{directed graphs} \cite{directed_lap}.


Understanding whether \eqref{introBound2} carries over to finite precision is an important open question, which has actually received more attention than the general matrix function problem.
In seminal work, Greenbaum \cite{greenbaum1989behavior} gives a natural finite precision extension of \eqref{introBound2}: performance can be bounded by the error in approximating $1/x$ in a \emph{tiny range around each eigenvalue}. 
Here ``tiny'' means essentially on the order of machine precision -- the approximation need only be over ranges of width $\eta$ as long as the bits of precision used is $\gtrsim \log(1/\eta)$.

Greenbaum's bound applies to the conjugate gradient (CG) method, a somewhat optimized way of applying Lanczos to linear systems. A precise version of Theorem 3 in  \cite{greenbaum1989behavior} can be summarized as follows (see Appendix \ref{sec:linsystems_app} for a detailed discussion):

%
%
\begin{theorem}[Conjugate Gradient in Finite Arithmetic \cite{greenbaum1989behavior}]\label{thm:greenbaum}
Given positive definite $\bv{A} \in \R^{n \times n}$ and $\bv{x} \in \R^n$, after $k$ iterations, the conjugate gradient algorithm run on a computer with $\Omega \left (\log \frac{nk\|\bv{A}\|}{\min(\eta,\lmin(\bv{A}) )} \right)$ bits of precision returns $\bv{y}$ satisfying: 
$$\norm{\bv{A}^{-1} \bv{x} - \bv{y}} \le 2\kappa(\bv{A}) \cdot \bar\delta_{k}\norm{\bv{x}}$$
where 
\begin{align*}
\bar\delta_{k} \eqdef \min_{\substack{\textnormal{polynomial $p$}\\ \textnormal{ with degree $<k$} }} \left ( \max_{x \in \bigcup_{i=1}^n [\lambda_i(\bv{A})-\eta, \lambda_i(\bv{A})+\eta]} |p(x)-1/x| \right ).
\end{align*}
The CG algorithm run for $k$ iterations requires $O(\mv(\bv{A}) k + nk)$ time, where $\mv(\bv{A})$ is the time required to multiply $\bv{A}$ by a vector. 
\end{theorem}
Theorem \ref{thm:greenbaum} does not apply to general matrix functions but,
at least for the special case of $f(\bv{A}) = \bv{A}^{-1}$, it is stronger than our Theorem \ref{mainLemmaFullRuntime}. It is natural to ask by how much.

\subsubsection{Lower bound}

Surprisingly, we show that Greenbaum's bound is much weaker than the exact arithmetic guarantee \eqref{introBound2}, and in fact is
not significantly more powerful than Theorem \ref{mainLemmaFullRuntime}. Specifically, in Section \ref{sec:lower} we prove  that for any $\kappa$ and interval width $\eta$, there is a natural class of matrices with condition number $\kappa$ and just $O(\log \kappa \cdot \log 1/\eta)$ eigenvalues for which any `stable approximating polynomial' of the form required by Theorem \ref{thm:greenbaum} achieving $\bar{\delta}_k \le 1/6$ must have degree $\Omega(\kappa^c)$ for a fixed constant $c \ge 1/5$.

\begin{theorem}[Stable Approximating Polynomial Lower Bound] \label{thm:lb}There exists a fixed constant $1/5 \le c \le 1/2$ such that for any $\kappa \ge 2$, $0 < \eta \le \frac{1}{20 \kappa^2}$, and $n \ge \lfloor \log_2 \kappa\rfloor \cdot  \lceil \ln 1/\eta \rceil$, there is a positive definite $\bv{A} \in \R^{n \times n}$ with condition number $ \le \kappa$, such that for any $k < \lfloor \kappa^{c}/377 \rfloor$:
\begin{align*}
\bar{\delta}_k \eqdef \min_{\substack{\text{polynomial $p$}\\ \text{ with degree $<k$}}} \left ( \max_{x \in \bigcup_{i=1}^n [\lambda_i(\bv{A})-\eta, \lambda_i(\bv{A})+\eta]} |p(x)-1/x| \right ) \ge 1/6.
\end{align*}
\end{theorem}
%
Theorem \ref{thm:lb} immediately gives a strong lower bound against Greenbaum's result, even if we only require constant factor error. Setting $\log(1/\eta) = n/\log(\kappa)$ we have:
\begin{corollary}\label{cor:lb}
There exists a fixed $c \ge 1/5$ such that for any $\kappa \ge 2$, there is a positive definite $\bv{A} \in \R^{n \times n}$ with condition number $\le \kappa$ such that Theorem \ref{thm:greenbaum} predicts that CG must run for $\Omega(\kappa^c)$ iterations to guarantee $\norm{\bv{A}^{-1}\bv{x}-\bv{y}} \le \frac{\kappa \cdot \norm{\bv{x}}}{3}$ if $o(n/\log \kappa)$ bits of precision are used.
\end{corollary}
As a consequence, if we set $\kappa = n^d$ for arbitrarily large constant $d$, Theorem \ref{thm:greenbaum} only guarantees a $\Omega(n^{cd})$ iteration bound, even when the precision used is \emph{nearly exponential} in $n$. Since  $O(\kappa^{1/2}) = O(n^{d/2})$ is already achievable via Theorem \ref{mainLemmaFullRuntime} with $O(\log n)$ bits of precision, Greenbaum's bound is not a significant improvement, except in very high precision regimes. While our constant $c$ is  $< 1/2$, we believe the proof can be tightened to show that $\sim \kappa^{1/2}$ degree is necessary.

Corollary \ref{cor:lb} can also be interpreted as showing the existence of matrices with $O(\log^2 \kappa)$ eigenvalues for which Theorem \ref{thm:greenbaum} requires $\Omega(\kappa^c)$ iterations for convergence if $O(\log \kappa)$ bits of precision are used. This is nearly exponentially worse than the exact arithmetic case, where \eqref{introBound2} gives convergence to perfect accuracy in $O(\log^{2}\kappa)$ iterations. 

Theorem \ref{thm:lb} seems damning for establishing iteration bounds on the Lanczos and CG methods in finite precision that go significantly beyond uniform approximation of $1/x$. Informally, all known bounds improving on $O(\sqrt{\kappa})$ iterations, including those for clustered or skewed eigenvalue distributions, require a polynomial that stably approximates $1/x$ on some small subset of poorly conditioned eigenvalues. We rule out the existence of such polynomials.

However, Theorem \ref{thm:lb} is not a general lower bound on the performance of finite precision Lanczos methods for solving linear systems. It is possible that these methods do something ``smarter'' than applying a fixed stable polynomial. Thus, we see our result as pointing to two possibilities:

\noindent
\hspace{.04\textwidth}\begin{minipage}{.92\textwidth}
\vspace{1em}
\textbf{Optimistic:} Bounds comparable \eqref{introBound2} can be proven for finite precision Lanczos or conjugate gradient, but are out of the reach of current techniques. Proving such bounds may require looking beyond a ``polynomial'' view of these methods.

\vspace{.5em}
\noindent\textbf{Pessimistic:} For finite precision Lanczos methods to converge in $k$ iterations, there must essentially exist a stable degree $k$ polynomial approximating $1/x$ in small ranges around $\bv{A}$'s eigenvalues. If this is the case, our lower bound could be extended to an unconditional lower bound on the number of iterations required for solving $\bv{A}^{-1}\bv{x}$ with such methods.
\end{minipage}

\section{Notation and linear algebra preliminaries}
\textbf{Notation} \hspace{.5em} We use bold uppercase letters for matrices and lowercase letters for vectors (i.e. matrices with multiple rows, 1 column). A lowercase letter with a subscript is used to denote a particular column vector in the corresponding matrix. E.g. $\bv{q}_5$ denotes the $5^\text{th}$ column in the matrix $\bv{Q}$.
Non-bold letters denote scalar quantities. A superscript $^T$ denotes the transpose of a matrix or vector. $\bv{e}_i$ denotes the $i^\text{th}$ standard basis vector, i.e. a vector with a 1 at position $i$ and 0's elsewhere. Its length will be clear from context. We use $\bv{I}_{k\times k}$ to denote the $k\times k$ identity matrix, removing the subscript when it is clear from context. When discussing runtimes, we occasionally use $\tilde{O}(x)$ as shorthand for $O(x\log^c x)$, where $c$ is a fixed positive constant.

\medskip
\noindent\textbf{Matrix Functions} \hspace{.5em} The main subject of this work is matrix functions and their approximation by matrix polynomials. We define matrix functions in the standard way, via the eigendecomposition:
\begin{definition}[Matrix Function]
\label{def:matrix_func}
For any function $f:\R\to\R$, for any real symmetric matrix $\bv{M}$, which can be diagonalized as $\bv{M} = \bv{V}\bv{\Lambda}\bv{V}^T$, we define the matrix function $f(\bv{M})$ as:
	\begin{align*}
		f(\bv{M}) \eqdef \bv{V}f(\bv{\Lambda})\bv{V}^T,
	\end{align*}
where $f(\bv{\Lambda})$ is a diagonal matrix obtained by applying $f$ independently to each eigenvalue on the diagonal of $\bv{\Lambda}$ (including any $0$ eigenvalues).
\end{definition}

\medskip
\noindent\textbf{Other} \hspace{.5em}
For a vector $\bv{x}$, $\norm{\bv{x}}$ denotes the Euclidean norm. For a matrix $\bv{M}$, $\norm{\bv{M}}$ denotes the spectral norm and $\kappa(\bv{M}) = \norm{\bv{M}} \norm{\bv{M}^{-1}}$ the condition number. We denote the eigenvalues of a symmetric matrix $\bv{M}\in \R^{n\times n}$ by $\lambda_1(\bv{M}) \ge \lambda_2(\bv{M}) \ge \ldots \ge \lambda_n(\bv{M})$, often writing $\lmax(\bv{M}) \eqdef \lambda_1(\bv{M})$ and $\lmin(\bv{M}) \eqdef \lambda_n(\bv{M})$.
 $\nnz(\bv{M})$ denotes the number of non-zero entries in $\bv{M}$.

\section{The Lanczos method in exact arithmetic}
\label{sec:exact_arthmetic}
We begin by presenting the classic Lanczos method and demonstrate how it can be used to approximate $f(\bv{A}) \bv{x}$ for any function $f$ and vector $\bv{x}$ when computations are performed in \emph{exact arithmetic}. While the results in this section are well known, we include an analysis that will mirror and inform our eventual finite precision analysis.
%
%
%

\begin{algorithm} 
\caption{Lanczos Method for Computing Matrix Functions}
{\bf input}: symmetric $\bv{A} \in \mathbb{R}^{n \times n}$, $\#$ of iterations $k$, vector $\bv{x} \in \R^n$, function $f:\R\rightarrow\R$\\
{\bf output}: vector $\bv{y}\in \R^n$ which approximates $f(\bv{A})\bv{x}$
\begin{algorithmic}[1]
\STATE{$\bv{q}_{0} = \bv{0}$, $\bv{q}_{1} = \bv{x}/\|\bv{x}\|$, $\beta_1 = 0$}
\FOR{$i \in {1,\ldots,k}$} 
\STATE {$\bv{q}_{i+1} \gets \bv{A}\bv{q}_i  - \beta_{i}\bv{q}_{i-1}$}
\STATE {$\alpha_i \gets \langle \bv{q}_{i+1}, \bv{q}_i\rangle$}
\STATE {$\bv{q}_{i+1} \gets \bv{q}_{i+1} - \alpha_i \bv{q}_i$} 
\STATE {$\beta_{i+1} \gets \norm{\bv{q}_{i+1}}$}
 \IF{$\beta_{i+1} == 0$} \STATE {break loop}\ENDIF
\STATE {$\bv{q}_{i+1} \gets \bv{q}_{i+1}/\beta_{i+1}$}
 \ENDFOR
 \vspace{.5em}
\STATE{$\bv{T} \gets \left[ \begin{matrix}
				\alpha_1 & \beta_2 & & 0\\
				\beta_2 & \alpha_2 & \ddots & \\
				& \ddots & \ddots & \beta_{k} \\
				0 & & \beta_{k} & \alpha_k 
				\end{matrix} \right]$, \hspace{1em} $\bv{Q} \gets \left[\begin{matrix} \bv{q}_1 & \ldots & \bv{q}_k \end{matrix} \right]$, 
				}
 \vspace{.5em}
\RETURN{$\bv{y} = \|\bv{x}\|\cdot\bv{Q}f(\bv{T})\bv{e}_1$} \label{final_step}
\end{algorithmic}
\label{alg:lanczos}
\end{algorithm}


We study the standard implementation of Lanczos described in Algorithm \ref{alg:lanczos}.
In exact arithmetic, the algorithm computes an orthonormal matrix $\bv{Q}$ with $\bv{q}_1 = \bv{x}/\|\bv{x}\|$ as its first column such that for all $j \leq k$, $[\bv{q}_1, \bv{q}_2, \ldots, \bv{q}_j]$ spans the rank-$j$ Krylov subspace:
\begin{align}\label{krylovDef}
	\mathcal{K}_j = [\bv{x}, \bv{A}\bv{x}, \bv{A}^2\bv{x}, \ldots, \bv{A}^{j-1}\bv{x}].
\end{align}
The algorithm also computes symmetric tridiagonal $\bv{T} \in \R^{k\times k}$ such that $\bv{T} =  \bv{Q}^T\bv{A}\bv{Q}$.
\footnote{
For conciseness, we ignore the case when the algorithm terminates early because $\beta_{i+1} = 0$. In this case, either $\bv{A}$ has rank $i$ or $\bv{x}$ only has a non-zero projection onto $i$ eigenvectors of $\bv{A}$. Accordingly, for any $j \ge 1$ $\mathcal{K}_j$ is spanned by $\mathcal{K}_i$ so there is no need to compute additional vectors beyond $\bv{q}_i$: any polynomial $p(\bv{A})\bv{x}$ can be formed by recombining vectors in $[\bv{q}_1, \bv{q}_2, \ldots, \bv{q}_i]$. It is tedious but not hard to check that our proofs go through in this case.
}

While the Krylov subspace interpretation of the Lanczos method is useful in understanding the function approximation guarantees that we will eventually prove, there is a more succinct way of characterizing the algorithm's output that doesn't use the notion of Krylov subspaces. It has been quite useful in analyzing the algorithm since the work of Paige \cite{paigeThesis}, and will be especially useful when we study the algorithm's behavior in finite arithmetic. 

\begingroup
\makeatletter
\apptocmd{\thetheorem}{\unless\ifx\protect\@unexpandable@protect\protect \,\,-- Exact Arithmetic\fi}{}{}
\makeatother
\begin{claim}[Lanczos Output Guarantee]
\label{claim:exact_lanc_guarantee}
Run for $k\leq n$ iterations using exact arithmetic operations, the Lanczos algorithm (Algorithm \ref{alg:lanczos}) computes $\bv{Q} \in \R^{n\times k}$, an additional column vector $\bv{q}_{k+1} \in \R^n$, 
a scalar $\beta_{k+1}$, and a symmetric tridiagonal matrix $\bv{T} \in \R^{k\times k}$ such that:
	\begin{align}\label{eq:exact_lanc_guarantee}
		\bv{A}\bv{Q} = \bv{Q}\bv{T} + \beta_{k+1}\bv{q}_{k+1}\bv{e}_k^T,
	\end{align}
and
	\begin{align}\label{spectral_norm_bound_on_q}
		\left[\bv{Q},\bv{q}_{k+1}\right]^T\left[\bv{Q},\bv{q}_{k+1}\right] = \bv{I}.
	\end{align}
Together \eqref{eq:exact_lanc_guarantee} and \eqref{spectral_norm_bound_on_q} also imply that:
\begin{align}
\label{exact_arith_eig_bound}
\lmin(\bv{T}) &\geq \lmin(\bv{A}) &&\text{and} &\lmax(\bv{T}) &\leq  \lmax(\bv{A}).
\end{align}
When run for $k \ge n$ iterations, the algorithm terminates at the $n^{th}$ iteration with $\beta_{n+1} = 0$.
\end{claim}
\endgroup

We include a brief proof in Appendix \ref{omitted_proofs} for completeness. The formulation of Claim \ref{claim:exact_lanc_guarantee} is valuable because it allows use to analyze how Lanczos applies polynomials via the following identity:
\begin{align}
\label{telescope}
	\bv{A}^q\bv{Q} -  \bv{Q}\bv{T}^q &= \sum_{i=1}^q\bv{A}^{q-i}\left(\bv{A}\bv{Q} -  \bv{Q}\bv{T}\right)\bv{T}^{i-1}.
\end{align}
In particular, \eqref{eq:exact_lanc_guarantee} gives an explicit expression for $\left(\bv{A}\bv{Q} -  \bv{Q}\bv{T}\right)$. Ultimately, our finite precision analysis is based on a similar expression for this central quantity.

\subsection{Function approximation in exact arithmetic}

We first show that Claim \ref{claim:exact_lanc_guarantee} can be used to prove \eqref{introBound}: Lanczos approximates matrix functions essentially as well as the best degree $k$ polynomial approximates the corresponding scalar function on the range of $\bv{A}$'s eigenvalues. We begin with a statement that applies for any function $f(x)$:

\begingroup
\makeatletter
\apptocmd{\thetheorem}{\unless\ifx\protect\@unexpandable@protect\protect \,\,-- Exact Arithmetic\fi}{}{}
\makeatother
\begin{theorem}[Approximate Application of Matrix Functions]\label{exact_lanczos_final_theorem}
Suppose $\bv{Q} \in \R^{n\times k}$, $\bv{T} \in \R^{k\times k}$, $\beta_{k+1}$, and $\bv{q}_{k+1}$ are computed by the Lanczos algorithm (Algorithm \ref{alg:lanczos}), run with exact arithmetic on inputs $\bv{A}$ and $\bv{x}$. Let
\begin{align*}
\delta_k =  \min_{\substack{\text{polynomial $p$} \\ \text{w/ degree $<k$}}}\left(\max_{x \in [\lmin(\bv{A}),\lmax(\bv{A})]}\left|f(x) - p(x)\right|	\right).
\end{align*}
Then the output $\bv{y} = \|\bv{x}\|\cdot\bv{Q}f(\bv{T})\bv{e}_1$ satisfies:
	\begin{align} \label{eq:approx_function_app}
		\|f(\bv{A})\bv{x} - \bv{y}\| \leq 2\delta_k \|\bv{x}\|.
		\end{align}
\end{theorem}
\endgroup

Theorem \ref{exact_lanczos_final_theorem} is proven from the following lemma, which says that the Lanczos algorithm run for $k$ iterations can \emph{exactly} apply any matrix polynomial with degree $< k$.

\begingroup
\makeatletter
\apptocmd{\thetheorem}{\unless\ifx\protect\@unexpandable@protect\protect \,\,-- Exact Arithmetic\fi}{}{}
\makeatother
\begin{lemma}[Exact Application of Polynomials]\label{exact_lanczos}
If $\bv{A}$, $\bv{Q}$, $\bv{T}$, $\beta_{k+1}$, and $\bv{q}_{k+1}$ satisfy  \eqref{eq:exact_lanc_guarantee} of Claim \ref{claim:exact_lanc_guarantee} (e.g. because they are computed with the Lanczos method), then for any polynomial $p$ with degree $<k$:
	\begin{align*}
		p(\bv{A})\bv{q}_1 = \bv{Q}p(\bv{T})\bv{e}_1.
	\end{align*}
\end{lemma}
\endgroup
Recall that in Algorithm \ref{alg:lanczos}, we set $\bv{q}_1 = \bv{x}/\|\bv{x}\|$, so the above trivially gives $p(\bv{A}) \bv{x} = \norm{\bv{x}} \bv{Q} p(\bv{T}) \bv{e}_1$.
\begin{proof}
We show that for any integer $1 \le q < k$:
\begin{align}\label{exact_power_app}
		\bv{A}^q\bv{q}_1 = \bv{Q}\bv{T}^q\bv{e}_1.
\end{align}
The lemma then follows by linearity as any polynomial $p$ with degree $< k$ can be written as the sum of these monomial terms.
To prove \eqref{exact_power_app}, we appeal to the telescoping sum in \eqref{telescope}. Specifically, since $\bv{q}_1 = \bv{Q}\bv{e}_1$, \eqref{exact_power_app}  is equivalent to:
\begin{align}\label{exact_power_app_rearranged}
	\left(\bv{A}^q\bv{Q} -  \bv{Q}\bv{T}^q\right) \bv{e}_1 = \bv{0}.
\end{align}
For $q \geq 1$, \eqref{telescope} let's us write:
\begin{align*}
	\left(\bv{A}^q\bv{Q} -  \bv{Q}\bv{T}^q\right) \bv{e}_1 &= \left(\sum_{i=1}^q\bv{A}^{q-i}\left(\bv{A}\bv{Q} -  \bv{Q}\bv{T}\right)\bv{T}^{i-1}\right)\bv{e}_1.
\end{align*}
Substituting in \eqref{eq:exact_lanc_guarantee}: 
\begin{align} \label{simplified_telescope_sum}
	\left(\bv{A}^q\bv{Q} - \bv{Q}\bv{T}^q\right) \bv{e}_1 =
	\beta_{k+1}\sum_{i=1}^q\bv{A}^{q-i}\bv{q}_{k+1}\bv{e}_k^T\bv{T}^{i-1}\bv{e}_1.
\end{align}
Since $\bv{T}$ is tridiagonal, $\bv{T}^{i-1}\bv{e}_1$ is zero everywhere besides its first $i$ entries. So, as long as $q < k$, $\bv{e}_k^T\bv{T}^{i-1}\bv{e}_1 = 0$ for all $i \leq q$. Accordingly, \eqref{simplified_telescope_sum} evaluates to $\bv{0}$, proving \eqref{exact_power_app_rearranged} and Lemma \ref{exact_lanczos}.
\end{proof}

With Lemma \ref{exact_lanczos} in place, Theorem \ref{exact_lanczos_final_theorem} intuitively follows because Lanczos always applies the ``low degree polynomial part'' of $f(\bv{A})$. The proof is a simple application of triangle inequality.

\begin{proof}[Proof of Theorem \ref{exact_lanczos_final_theorem}]

\begin{align}
\label{taking_norm_out}
\|f(\bv{A})\bv{x} - \bv{y}\| = \|f(\bv{A})\bv{q}_1 - \bv{Q}f(\bv{T})\bv{e}_1\| \cdot  \|\bv{x}\|
\end{align}
For any polynomial $p$, we can write:
\begin{align}
\|f(\bv{A})\bv{q}_1 - \bv{Q}f(\bv{T})\bv{e}_1\|  &\leq \|p(\bv{A})\bv{q}_1 - \bv{Q}p(\bv{T})\bv{e}_1\| + \|\left[f(\bv{A}) - p(\bv{A})\right]\bv{q}_1 - \bv{Q}\left[f(\bv{T}) - p(\bv{T})\right]\bv{e}_1\| \nonumber \\
&\leq 0 + \|\left[f(\bv{A}) - p(\bv{A})\right]\bv{q}_1\| + \|\bv{Q}\left[f(\bv{T}) - p(\bv{T})\right]\bv{e}_1\| \nonumber\\
\label{main_tri_split}
&\leq \|f(\bv{A}) - p(\bv{A})\| + \|\bv{Q}\|\|f(\bv{T}) - p(\bv{T})\|.
\end{align} 
In the second step we use triangle inequality, in the third we use Lemma \ref{exact_lanczos} and triangle inequality, and in the fourth we use submultiplicativity of the spectral norm and the fact that $\norm{\bv{q}_1} = \norm{\bv{e}_1} = 1$. 

$f(\bv{A}) - p(\bv{A})$ is symmetric and has an eigenvalue equal to $f(\lambda) - p(\lambda)$ for each eigenvalue $\lambda$ of $\bv{A}$. Accordingly:
\begin{align*}
\|f(\bv{A}) - p(\bv{A})\| \leq \max_{x \in [\lmin(\bv{A}),\lmax(\bv{A})]} |f(x) - p(x)|.
\end{align*}
Additionally, by \eqref{exact_arith_eig_bound} of Claim \ref{claim:exact_lanc_guarantee}, for any eigenvalue $\lambda(\bv{T})$ of $\bv{T}$, $\lmin(\bv{A}) \leq \lambda(\bv{T})  \leq \lmax(\bv{A})$ so:
\begin{align*}
\|f(\bv{T}) - p(\bv{T})\| \leq \max_{x \in [\lmin(\bv{A}),\lmax(\bv{A})]} |f(x) - p(x)|.
\end{align*}
Plugging both bounds into \eqref{main_tri_split}, along with the fact that $\|\bv{Q}\| = 1$ and that these statements hold for \emph{any} polynomial with degree $<k$ gives $\norm{f(\bv{A})\bv{x} - \bv{y}} \le 2 \delta_k \norm{\bv{x}}$ after rescaling via \eqref{taking_norm_out}.
\end{proof}

As discussed in the introduction, Theorem \ref{exact_lanczos_final_theorem} can be tightened in certain special cases, including when $\bv{A}$ is positive definite and $f(\bv{A}) = \bv{A}^{-1}$. We defer consideration of this point to Section \ref{sec:lower}.

\section{Finite precision preliminaries}
\label{sec:fpprelim}

Our goal is to understand how Theorem \ref{exact_lanczos_final_theorem} and related bounds translate from exact arithmetic to finite precision. 
In particular, our results apply to machines that employ floating-point arithmetic. We use $\mach$ to denote the relative precision of the floating-point system. 
An algorithm is generally considered ``stable" if it runs accurately when $1/\mach$ is bounded by some polynomial in the input parameters, i.e., when the number of bits required is logarithmic in these parameters.

We say a machine has precision $\mach$ if it can perform computations to relative error $\mach$, which necessarily requires that it can represent numbers to relative precision $\mach$ -- i.e., it has  $\geq \log_2 (1/\mach)$ bits in its floating point significand. To be precise, we require:

\begin{requirement}[Accuracy of floating-point arithmetic]\label{req1}
Let $\circ$ denote any of the four basic arithmetic operations ($+$, $-$, $\times$, $\div$) and let  $\fl(x \circ y)$ denote the result of computing $x \circ y$. Then a machine with precision  $\mach$ must be able to compute:
\begin{align*}
\fl(x \circ y) &= (1+\delta) (x \circ y) & &\text{where} & |\delta| &\le \mach
\end{align*}
and
\begin{align*}
\fl(\sqrt{x}) &= (1+\delta) \sqrt{x} & &\text{where} & |\delta| &\le \mach.
\end{align*}
\end{requirement}
Requirement \ref{req1} is satisfied by any computer implementing the IEEE 754 standard for floating-point arithmetic \cite{IEEE} with $\geq \log_2 (1/\mach)$ bits of precision, as long as  operations do not overflow or underflow\footnote{Underflow is only a concern for $\times$ and $\div$ operations. On any computer implementing gradual underflow and a guard bit, Requirement \ref{req1} always holds for $+$ and $-$, even when underflow occurs. $\sqrt{x}$ cannot underflow or overflow.}. Underflow or overflow occur when $(1+\delta) (x \circ y)$ cannot be represented in finite precision for any $\delta$ with $|\delta| \le \mach$, either because $x \circ y$ is so large that it exceeds the maximum expressible number on the computer or because it is so small that expressing the number to relative precision would require a negative exponent that is larger in magnitude than that supported by the computer.
As is typical in stability analysis, we will ignore the possibility of overflow and underflow because doing so significantly simplifies the presentation of our results \cite{higham2002accuracy}. 

However, because the version of Lanczos studied normalizes vectors at each iteration, it is not hard to check that our proofs, and the results of Paige, and Gu and Eisenstat that we rely on, go through with overflow and underflow accounted for. To be more precise, overflow does not occur as long as all numbers in the input (and their squares) are at least a $\poly(k,n,C)$ factor smaller than the maximum expressible number (recall that in Theorem \ref{mainLemmaFullRuntime}, $C$ is an upper bound on $|f(x)|$ over our eigenvalue range). That is, overflow is avoided if we assume the exponent in our floating-point system has $\Omega(\log\log(kn\cdot\max(C,1)))$ bits overall and $\Omega(1)$ bits more than what is needed to express the input. This ensures, for example, that the computation of $\norm{\bv{x}}$ does not overflow and that the multiplication $\bv{A}\bv{w}$ does not overflow for any unit norm $\bv{w}$. 

To account of underflow, Requirement \ref{req1} can be modified by including additive error $\gmach$ for $\times$ and $\div$ operations, where $\gmach$ denotes the smallest expressible positive number on our floating-point machine. The additive error carries through all calculations, but will be swamped by multiplicative error as long as we assume that $\norm{\bv{A}}$, $\norm{\bv{x}}$, $\mach$, and our function upper bound $C$ are larger than $\gmach$ by a $\poly(k,n, 1/\mach)$ factor. This ensures, e.g., that $\bv{x}$ can be normalized stably and, as we will discuss, allows for accurate multiplication of the input matrix $\bv{A}$ any vector.

In addition to Requirement \ref{req1}, we also require the following of matrix-vector multiplications involving our input matrix $\bv{A}$:
\begin{requirement}[Accuracy of matrix multiplication]\label{req2}
Let $\fl(\bv{A}\bv{w})$ denote the result of computing $\bv{A}\bv{w}$ on our floating-point computer. Then a computer with precision  $\mach$ must be able to compute, for any $\bv{w}\in \R^{n}$,
\begin{align*}
\|\fl(\bv{A}\bv{w}) - \bv{A}\bv{w}\| \leq 2n^{3/2}\|\bv{A}\|\|\bv{w}\| \mach.
\end{align*}
\end{requirement}
If $\bv{A}\bv{w}$ is computed explicitly, as long as $n\mach \leq \frac{1}{2}$ (which holds for all of our results), any computer satisfying Requirement \ref{req1} also satisfies Requirement \ref{req2} \cite{Wilkinson,higham2002accuracy}. We list Requirement \ref{req2} separately to allow our analysis to apply in situations where $\bv{A}\bv{w}$ is computed approximately for reasons other than rounding error. For example, in many applications where $\bv{A}$ cannot be accessed explicitly, $\bv{A}\bv{w}$ is approximated with an iterative method \cite{frostig2016principal,matrixExp}. As long as this computation is performed to the precision specified in Requirement \ref{req2}, then our analysis holds.

As mentioned, when $\bv{A}\bv{w}$ is computed explicitly, underflow could occur during intermediate steps on a finite precision computer. This will add an error term of $2n^{3/2}\gmach$ to $\|\fl(\bv{A}\bv{w}) - \bv{A}\bv{w}\|$. However, under our assumption that $\mach \|\bv{A}\| \gg \gmach$, this term is subsumed  by the $2n^{3/2}\|\bv{A}\|\|\bv{w}\| \mach$ term whenever $\|\bv{w}\|$ is not tiny (in Algorithm \ref{alg:lanczos}, $\|\bv{w}\|$ is always very close to 1).

Finally, we mention that, in our proofs, we typically show that operations incur error $\mach \cdot F$ for some value $F$ that depends on problem parameters. Ultimately, to obtain error $0 < \epsilon \leq 1$ we then require that $\mach \leq \epsilon/F$. Accordingly, during the course of a proof we will often assume that $\epsilon\cdot F \leq 1$. Additionally, all runtime bounds are for the unit-cost RAM model: we assume that computing $\fl(x \circ y)$ and $\fl(\sqrt{x})$ require $O(1)$ time. For simplicity, we also assume that the scalar function $f$ we are interested in applying to $\bv{A}$ 
 can be computed to relative error $\mach$ in $O(1)$ time.
%
\section{Lanczos in finite precision}
\label{lanczos_in_finite_precision}
The most notable issue with the Lanczos algorithm in finite precision is 
that $\bv{Q}$'s column vectors lose the mutual orthogonality property of \eqref{spectral_norm_bound_on_q}. In practice, this loss of orthogonality is quite severe: $\bv{Q}$ will often have numerical rank $\ll k$.  Naturally, $\bv{Q}$'s column vectors will thus also fail to span the Krylov subspace $\mathcal{K}_k = [\bv{q}_1, \bv{A}\bv{q}_1, \ldots, \bv{A}^{k-1}\bv{q}_1]$, and so we do not expect to be able to accurately apply all degree $< k$ polynomials.  Surprisingly, this does not turn out to be much of a problem!



\subsection{Starting point: Paige's results}
In particular, a seminal result of Paige shows that while \eqref{spectral_norm_bound_on_q}  falls apart under finite precision calculations, \eqref{eq:exact_lanc_guarantee} of Claim \ref{claim:exact_lanc_guarantee} still holds, up to small error. In particular, in \cite{paige1976error} he proves 
that:
\begin{theorem}[Lanczos Output in Finite Precision,
 \cite{paige1976error}]
\label{thm:paige_main}
Run for $k$ iterations
 on a computer satisfying Requirements \ref{req1} and \ref{req2} with relative precision $\mach$, the Lanczos algorithm (Algorithm \ref{alg:lanczos})  computes $\bv{Q} \in \R^{n\times k}$, an additional column vector $\bv{q}_{k+1} \in \R^n$, 
a scalar $\beta_{k}$, and a symmetric tridiagonal matrix $\bv{T} \in \R^{k\times k}$ such that:
\begin{align}
\label{eq:finite_lanc_guarantee}
\bv{A}\bv{Q} = \bv{Q}\bv{T} + \beta_{k+1}\bv{q}_{k+1}\bv{e}_k^T + \bv{E},
\end{align}
and
\begin{align}
\label{finite_spectral_norm_bound_on_q}
\|\bv{E}\| &\leq k(2n^{3/2}+7)\norm{\bv{A}}\epsilon_{mach}, \\
\left |\|\bv{q}_i\| -1 \right |&\leq (n + 4)\mach \text{\hspace{.5em}  for all $i$}.
\end{align}
In \cite{paige1980error} (see equation 3.28), it is shown that together, the above bounds also imply:
\begin{align}
\label{old_corr}
\lmin(\bv{A}) - \epsilon_1 \leq \lambda(\bv{T}) \leq \lmax(\bv{A}) + \epsilon_1
\end{align}
where $\epsilon_1  = k^{5/2}\|\bv{A}\|\left(68 +17n^{3/2}\right)\epsilon_{mach}$.
\end{theorem}

Paige was interested in using Theorem \ref{thm:paige_main} to understand how $\bv{T}$ and $\bv{Q}$ can be used to compute approximate eigenvectors and values  for $\bv{A}$. His bounds are quite strong: for example, \eqref{old_corr} shows that 
\eqref{exact_arith_eig_bound} still holds up to tiny additive error, even though establishing that result for exact arithmetic relied heavily on the orthogonality of $\bv{Q}$'s columns. 


\subsection{Finite precision lanczos for applying polynomials}

Theorem \ref{thm:paige_main} allows us to give a finite precision analog of Lemma \ref{exact_lanczos} for polynomials with magnitude $|p(x)|$ bounded on a small extension of the 
eigenvalue range $[\lmin(\bv{A}),\lmax(\bv{A})]$. 
\begin{lemma}[Lanczos Applies Bounded Polynomials]\label{poly_with_error}  Suppose $\bv{Q} \in \R^{n\times k}$ and $\bv{T} \in \R^{k\times k}$ are computed by the Lanczos algorithm
 on a computer satisfying Requirements \ref{req1} and \ref{req2} with relative precision $\mach$, and thus these matrices satisfy the bounds of Theorem \ref{thm:paige_main}. For any $\eta \ge  85 n^{3/2}k^{5/2} \norm{\bv{A}}\mach$, if $p$ is a polynomial with degree $<k$ and $|p(x)| \leq C$ for all $x\in \left[\lmin(\bv{A})-\eta, \lmax(\bv{A}) + \eta \right]$ then:
\begin{align}
\label{smooth_poly_error_bound}
\|\bv{Q}p(\bv{T})\bv{e}_1 - p(\bv{A})\bv{q}_1\| 
 \leq  \frac{4Ck^3 \norm{\bv{E}}}{\lmax(\bv{A}) - \lmin(\bv{A})  + 2\eta}\end{align}
where $\bv{E}$ is the error matrix defined in Theorem \ref{thm:paige_main}.
\end{lemma}

\subsubsection*{Finite precision Lanczos applies Chebyshev polynomials}

It is not immediately clear how to modify the proof of Lemma \ref{exact_lanczos} to handle the error $\bv{E}$ in \eqref{eq:finite_lanc_guarantee}. Intuitively, any bounded polynomial cannot have too large a derivative by the Markov brothers' inequality \cite{markovBrothers}, and so we expect $\bv{E}$  to have a limited effect. However, we are not aware of a way to make this reasoning formal for matrix polynomials and arbitrary error matrix $\bv{E}$. 

As illustrated in \cite{matrixExp}, there is a natural way to prove \eqref{smooth_poly_error_bound} for the monomials $\bv{A}, \bv{A}^2, \ldots, \bv{A}^{k-1}$. The bound can then be extended to all polynomials via triangle inequality, but error is amplified by the coefficients of each monomial component in $p(\bv{A})$. Unfortunately, there are polynomials that are uniformly bounded by $C$ (and thus have bounded derivative) even though their monomial components can have coefficients much larger than $C$. The ultimate effect is that the approach taken in \cite{matrixExp} would incur an exponential dependence on $k$ on the right hand side of \eqref{smooth_poly_error_bound}.
 
To obtain our stronger polynomial dependence, we proceed with a different two-part analysis. We first show that \eqref{smooth_poly_error_bound} holds for any \emph{Chebyshev polynomial} with degree $< k$ that is appropriately stretched and shifted to the range $\left[\lmin(\bv{A})-\eta, \lmax(\bv{A}) + \eta \right]$. Chebyshev polynomials have magnitude much smaller than that of their monomial components, but because they can be formed via a well-behaved recurrence, we can show that they are stable to the perturbation $\bv{E}$. We can then obtain the general result of Lemma \ref{poly_with_error} because \emph{any} bounded polynomial can be written as a weighted sum of such Chebyshev polynomials, with bounded weights.

Let $T_0,T_1, \ldots, T_{k-1}$ be the first $k$ Chebyshev polynomials of the first kind, defined recursively:
\begin{align}\label{recurrence}
	T_0(x) &= 1, \nonumber\\
	T_1(x) &= x, \nonumber\\
	T_i(x) &= 2xT_{i-1}(x) - T_{i-2}(x).
\end{align}
The roots of the Chebyshev polynomials lie in $[-1,1]$ and this is precisely the range where they remain ``well behaved'': for $|x| > 1$, $T_i(x)$ begins to grow quite quickly. 
Define 
\begin{align*}
\rmax &\eqdef \lmax + \eta & &\text{and} & \rmin &\eqdef \lmin -\eta
\end{align*}
and
%
\begin{align}\label{tprime}
	\delta &= \frac{2}{\rmax - \rmin} &&\text{and} 	& \ol{T}_i(x) &= T_i\left(\delta(x-\rmin) - 1\right).
\end{align}
$\ol{T}_i(x)$ is the $i^{th}$ Chebyshev polynomial stretched and shifted so that $\ol{T}_i(\rmin) = T_i(-1) $ and $\ol{T}_i(\rmax) = T_i(1)$. We prove the following:


\begin{lemma}[Lanczos Applies Chebyshev Polynomials Stably]\label{cheby_poly_with_error} Suppose $\bv{Q} \in \R^{n\times k}$ and $\bv{T} \in \R^{k\times k}$ are computed by the Lanczos algorithm
 on a computer satisfying Requirements \ref{req1} and \ref{req2} with relative precision $\mach$ and thus these matrices satisfy the bounds of Theorem \ref{thm:paige_main}. For any $\eta \ge 85 n^{3/2}k^{5/2}\norm{\bv{A}} \mach$, for all $i < k$,
\begin{align}
\label{cheby_poly_error_bound}
\|\bv{Q}\ol{T}_i(\bv{T})\bv{e}_1 - \ol{T}_i(\bv{A})\bv{q}_1\| 
\leq \frac{4i^2\cdot\|\bv{E}\|}{\lmax(\bv{A}) - \lmin(\bv{A})  + 2\eta}
\end{align}
where $\bv{E}$ is the error matrix in Theorem \ref{thm:paige_main} and $\ol{T}_i$ is the $i^{th}$ shifted Chebyshev polynomial of \eqref{tprime}.
\end{lemma}
\begin{proof}
Let $\rmin = \lmin(\bv{A})  -\eta$ and $\rmax = \lmax(\bv{A})  + \eta$.  Define $
\barA \eqdef \delta(\bv{A} - \rmin\bv{I}) - \bv{I}$ and $\barT \eqdef \delta(\bv{T} - \rmin\bv{I}) - \bv{I}$.
so \eqref{cheby_poly_error_bound} is equivalent to:
\begin{align}
\label{shifted_cheby_poly_error_bound}
\|\bv{Q} T_i (\barT)\bv{e}_1 - T_i(\barA)\bv{q}_1\| \leq \frac{4i^2\cdot \|\bv{E}\|}{\rmax-\rmin} .
\end{align}
So now we just focus on showing \eqref{shifted_cheby_poly_error_bound}.
We use the following notation:
\begin{align*}
	\bv{t}_i 		&\eqdef T_i(\bv{\barA})\bv{q}_1, \hspace{3em}
	\bv{\tilde t}_i	\eqdef T_i(\bv{\barT})\bv{e}_1, \\
	\bv{d}_i 		&\eqdef \bv{t}_i - \bv{Q}\bv{\tilde t}_i,\hspace{3em}
	\bs{\xi}_i 		\eqdef  \delta\bv{E}\bv{\tilde t}_{i-1}.
\end{align*}
Proving \eqref{shifted_cheby_poly_error_bound} is equivalent to showing $\|\bv{d}_i\| \le \frac{4i^2\|\bv{E}\|}{\rmax-\rmin}$. From the Chebyshev recurrence \eqref{recurrence} for all $i \ge 2$:
\begin{align*}
	\bv{d}_i 	&= \left( 2\barA \bv{t}_{i-1} - \bv{t}_{i-2}\right) - \bv{Q} \left( 2\barT \bv{\tilde t}_{i-1} - \bv{\tilde t}_{i-2} \right)\\
				&=2(\barA \bv{t}_{i-1} - \bv{Q} \barT \bv{\tilde t}_{i-1}) - \bv{d}_{i-2}.
\end{align*}
Applying the perturbed Lanczos relation \eqref{eq:finite_lanc_guarantee}, we can write $\bv{Q} \barT = \barA \bv{Q} - \delta \beta_{k+1} \bv{q}_{k+1} \bv{e}_k^T - \delta\bv{E}$. Plugging this in above we then have:
\begin{align*}
	\bv{d}_i &=2\barA ( \bv{t}_{i-1} - \bv{Q} \bv{\tilde t}_{i-1}) - \bv{d}_{i-2} + 2\delta \beta_{k+1} \bv{q}_{k+1} \bv{e}_k^T\bv{\tilde t}_{i-1} + 2\delta\bv{E}\bv{\tilde t}_{t-1}\\
	&= (2\barA \bv{d}_{i-1} - \bv{d}_{i-2}) + 2\delta \beta_{k+1} \bv{q}_{k+1} \bv{e}_k^T\bv{\tilde t}_{i-1} + 2\bs{\xi}_i.
\end{align*}
Finally, we use as in Lemma \ref{exact_lanczos}, that $\bv{e}_k^T\bv{\tilde t}_{i-1} = \bv{e}_k^T T_{i-1}(\barT) \bv{e}_1 = 0$ since $\barT$ (like $\bv{T}$) is tridiagonal. Thus, $\bv{T}^{q-1}\bv{e}_1$ is zero outside its first $q$ entries and so for $i < k$, $T_{i-1}(\barT) \bv{e}_1$ is zero outside of its first $k-1$ entries. This gives the error recurrence:
\begin{align}\label{eq:error_recur}
\bv{d}_i = (2 \barA \bv{d}_{i-1} - \bv{d}_{i-2}) + 2\bs{\xi}_i.
\end{align}

As in standard stability arguments for the \emph{scalar} Chebyshev recurrence, we can analyze \eqref{eq:error_recur} using  Chebyshev polynomials of the second kind \cite{clenshaw}. The $i^\text{th}$ Chebyshev polynomial of the second kind is denoted $U_i(x)$ and defined by the recurrence
\begin{align}\label{secondRecurrence}
	U_0(x) &= 1, \nonumber\\
	U_1(x) &= 2x, \nonumber\\
	U_i(x) &= 2xT_{i-1}(x) - T_{i-2}(x).
\end{align} 
We claim that for any $i \ge 0$, defining $U_{k}(x) = 0$ for any $k < 0$ for convinience:
\begin{align} \label{eq:sum_of_type_2s}
	\bv{d}_i = U_{i-1}(\barA)\bs{\xi}_1 + 2\sum_{j=2}^{i} U_{i-j}(\barA)\bs{\xi}_j.
\end{align}
This follows by induction starting with the base cases:
\begin{align*}
	\bv{d}_0 &= \bv{0}, \text{ and } \bv{d}_1 = \bs{\xi}_1.
\end{align*}
Using \eqref{eq:error_recur} and assuming by induction that \eqref{eq:sum_of_type_2s} holds for all $j < i$,
\begin{align*}
	\bv{d}_i &= 2\bs{\xi}_i + \left(2\barA \bv{d}_{i-1} - \bv{d}_{i-2}\right) \\
			&= 2\bs{\xi}_i + [2\bar A U_{i-2}(\barA)\bs{\xi}_1 - U_{i-3}(\barA)\bs{\xi}_1] + 4\barA \sum_{j=2}^{i-1} U_{i-1-j}(\barA)\bs{\xi}_j - 2\sum_{j=2}^{i-2} U_{i-2-j}(\barA)\bs{\xi}_j \\
			&= 2\bs{\xi}_i + U_{i-1}(\barA) \bs{\xi}_1 + \left(2\sum_{j=2}^{i-2} 2\barA U_{i-1-j}(\barA)\bs{\xi}_j - U_{i-2-j}(\barA)\bs{\xi}_j\right) + 4\barA U_{0}(\barA)\bs{\xi}_{i-1} \\
			&= U_{i-1}(\barA) \bs{\xi}_1 + 2\sum_{j=2}^{i} U_{i-j}(\barA)\bs{\xi}_j,
\end{align*}
establishing \eqref{eq:sum_of_type_2s}.
It follows from triangle inequality and submultiplicativity that
\begin{align*}
	\|\bv{d}_i \| \leq 2\sum_{j=1}^{i} \|U_{i-j}(\barA)\|\|\bs{\xi}_j\|.
\end{align*}
Since $\barA$ is symmetric (it is just a shifted and scaled $\bv{A}$), $U_{k}(\barA)$ is equivalent to the matrix obtained by applying $U_{k}(x)$ to each of $\barA$'s eigenvalues, which lie in the range $[-1,1]$. It is well known that, for values in this range $U_k(x) \leq k+1$ \cite{chebyBook}. Accordingly, $\|U_{i-j}(\barA)\| \leq i-j+1$, so
\begin{align}\label{dSumBound}
	\| \bv{d}_i \| \leq 2\sum_{j=1}^{i} (i-j+1)\|\bs{\xi}_j\| \leq 2\sum_{j=1}^{i} i\|\bs{\xi}_j\|.
\end{align}
We finally bound $\|\bs{\xi}_j\| = \delta\bv{E}\bv{\tilde t}_{j-1}$. Recall that $\bv{\tilde t}_{j-1} \eqdef T_{j-1}(\barT)\bv{e}_1$ so:
\begin{align}
	\|\bs{\xi}_j\| &= \|\delta\bv{E}T_{j-1}(\barT)\bv{e}_1\| \nonumber\\
	&\leq \delta\left\|\bv{E}\right\| \left| T_{j-1}(\barT)\right\| =  \frac{2}{\rmax - \rmin} \cdot \left\|\bv{E}\right\| \left\|T_{j-1}\left ( \delta\left(\bv{T} - \rmin\bv{I}\right) - \bv{I}  \right)\right\| \label{eq:final_norm_bound},
\end{align}
where we used that $\|\bv{e}_1\| = 1$, and $\delta = \frac{2}{\rmax - \rmin}$.
By \eqref{old_corr} of Theorem \ref{thm:paige_main} and our requirement that $\eta \ge 85 n^{3/2} k^{5/2} \norm{\bv{A}}\mach$, $\bv{T}$ has all eigenvalues in $[\lmin - \eta, \lmax + \eta]$. Thus $\barT = \delta\left(\bv{T} - \rmin\bv{I}\right) - \bv{I}$ has all eigenvalues in $[-1,1]$.
We have $T_{j-1}(x) \le 1$ for $x \in [-1,1]$, giving $\left\|T_{j-1}\left ( \delta\left(\bv{T} - \rmin\bv{I}\right) - \bv{I}  \right)\right\| \le 1$.

Plugging this back into \eqref{eq:final_norm_bound},
$
\norm{\bs{\xi}_j} \le \frac{2 \norm{\bv{E}}}{\rmax-\rmin}
$
and plugging into \eqref{dSumBound}, $\norm{\bv{d}_i} \le \frac{4 i^2 \norm{\bv{E}}}{\rmax-\rmin}$. 
This gives the result, recalling that $\rmin = \lmin(\bv{A})  -\eta$ and $\rmax = \lmax(\bv{A})  + \eta$.
\end{proof}

\subsubsection*{From Chebyshev polynomials to general polynomials}

As discussed, with Lemma \ref{cheby_poly_with_error} in place, we can prove Lemma \ref{poly_with_error} by writing any bounded polynomial in the Chebyshev basis.

\begin{proof}[Proof of Lemma \ref{poly_with_error}]
Recall that we define $\rmin = \lmin - \eta$, $\rmax = \lmax+\eta$, and $\delta = \frac{2}{\rmax-\rmin}$. Let
\begin{align*}
\ol{p}(x) = p\left(\frac{x+1}{\delta} + \rmin\right).
\end{align*}
For any $x \in [-1,1]$, $\ol{p}(x) = p(y)$ for some $y \in [\rmin,\rmax]$. This immediately gives
$|\ol{p}(x)| \leq C$ on $[-1,1]$ by the assumption that $|p(x)| \le C$ on $[\lmin-\eta,\lmax+\eta] = [\rmin,\rmax]$.

Any polynomial with degree $< k$ can be written as a weighted sum of the first $k$ Chebyshev polynomials (see e.g. \cite{chebyBook}). Specifically we have:
\begin{align*}
\ol{p}(x) = c_0T_0(x) + c_1T_1(x) + \ldots + c_{k-1}T_{k-1}(x),
\end{align*}
where the $i^\text{th}$ coefficient is given by:
\begin{align*}
c_i = \frac{2}{\pi}\int_{-1}^1 \frac{\bar p(x)T_i(x)}{\sqrt{1-x^2}}.
\end{align*}  
$|T_i(x)| \leq 1$ on $[-1,1]$ and $\int_{-1}^1 \frac{1}{\sqrt{1-x^2}} = \pi$, and since $|\ol{p}(x)| \le C \text { for } x \in [-1,1] $ we have for all $i$:
\begin{align}\label{coeffBound}
|c_i| \leq 2C.
\end{align}

By definition, $p(x) = \ol{p}\left(\delta(x-\rmin) - 1\right)$. Letting $\barA = \delta(\bv{A}-\rmin \bv{I})-\bv{I}$ and $\barT = \delta(\bv{T}-\rmin \bv{I})-\bv{I}$ as in the proof of Lemma \ref{cheby_poly_with_error}, we have $$\|\bv{Q} p(\bv{T})\bv{e}_1 - p(\bv{A})\bv{q}_1\| = \left\| \bv{Q} \ol{p}\left(\barT \right)\bv{e}_1 - \ol{p}\left(\barA \right)\bv{q}_1\right\|$$ so need to upper bound the right hand side to prove the lemma.
Applying triangle inequality:
\begin{align*}
	\left\|\bv{Q} \ol{p}\left(\barT \right)\bv{e}_1 - \ol{p}(\barA)\bv{q}_1\right\| 
	\leq  \sum_{i=0}^{k-1} c_i\left\|\bv{Q} T_i\left(\barT \right)\bv{e}_1 - T_i(\barA)\bv{q}_1\right\|,
\end{align*}
where for each i, $|c_i| \le 2C$ by \eqref{coeffBound}. Combining with Lemma \ref{cheby_poly_with_error},  we thus have:
\begin{align*}
\left\|\bv{Q}\ol{p}(\barT)\bv{e}_1 - \ol{p}(\barA)\bv{q}_1\right\|  \le \frac{8C\cdot\norm{\bv{E}}}{\lmax(\bv{A}) - \lmin(\bv{A})  + 2\eta} \sum_{i=0}^{k-1} i^2 \le \frac{4C k^3 \norm{\bv{E}}}{\lmax(\bv{A}) - \lmin(\bv{A})  + 2\eta},
\end{align*}
which gives the lemma.
\end{proof}

\subsection{Completing the analysis}\label{sec:completing}
With Lemma \ref{poly_with_error}, we have nearly proven our main result, Theorem \ref{mainLemmaFullRuntime}. We first show, using a proof mirroring our analysis in the exact arithmetic case, that Lemma \ref{poly_with_error} implies that $\bv{Q}f(\bv{T})\bv{e}_1$ well approximates $f(\bv{A})\bv{q}_1$. Thus the output $\bv{y} = \norm{\bv{x}}\bv{Q}f(\bv{T})\bv{e}_1$ well approximates $f(\bv{A})\bv{x}$. With this bound, all that remains in proving Theorem \ref{mainLemmaFullRuntime} is to show that we can compute $\bv{y}$ accurately using known techniques (although with a tedious error analysis).

\begin{lemma}[Stable function approximation via Lanczos]\label{mainLemmaDetailed} Suppose $\bv{Q} \in \R^{n\times k}$ and $\bv{T} \in \R^{k\times k}$ are computed by the Lanczos algorithm
 on a computer satisfying Requirements \ref{req1} and \ref{req2} with relative precision $\mach$ and thus these matrices satisfy the bounds of Theorem \ref{thm:paige_main}. For degree $k$ and any $\eta$ with $ 85 n^{3/2} k^{5/2} \norm{\bv{A}} \epsilon_{mach} \le \eta \le \norm{\bv{A}}$ define:
\begin{align}\label{deltaKCDef}
\delta_{k} =  \min_{\substack{\textnormal{polynomial $p$} \\\textnormal{ with degree $<k$}}} \left(\max_{x \in [\lmin(\bv{A}) - \eta,\lmax(\bv{A}) + \eta]}\left|f(x) - p(x)\right|	\right)
\end{align}
and $C = \max_{x \in [\lmin(\bv{A}) - \eta,\lmax(\bv{A}) + \eta]} |f(x)|$.
Then we have:
\begin{align}\label{finalBoundExactF}
\|f(\bv{A})\bv{q}_1 - \bv{Q}f(\bv{T})\bv{e}_1\| \le (k+2)\delta_{k} + \mach \cdot \frac{92C  k^4 n^{3/2} \norm{\bv{A}}}{\lmax(\bv{A}) - \lmin(\bv{A}) + 2\eta} .
\end{align}
\end{lemma}
\begin{proof}
	Let $\gamma = \lmax(\bv{A}) - \lmin(\bv{A}) + 2\eta$.
Applying Lemma \ref{poly_with_error}, letting $p$ be the optimal degree $< k$ polynomial achieving $\delta_k$, by \eqref{deltaKCDef} and our bound on $f(x)$ on this range:
\begin{align*}
\|\bv{Q}p(\bv{T})\bv{e}_1 - p(\bv{A})\bv{q}_1\| \le \frac{4 k^3 (C+\delta_k)\norm{\bv{E}}}{\gamma}.
\end{align*}
By triangle inequality, spectral norm submultiplicativity, and the fact that $\norm{\bv{q}_1} \approx 1$ (certainly $\norm{\bv{q}_1} \leq 2$ even if $\bv{x}$ is normalized in finite-precision) we have:
\begin{align}
\|\bv{Q}f(\bv{T})\bv{e}_1 - f(\bv{A})\bv{q}_1\| &\le \|\bv{Q}p(\bv{T})\bv{e}_1 - p(\bv{A})\bv{q}_1\| + \|\bv{Q}f(\bv{T})\bv{e}_1 - \bv{Q}p(\bv{T})\bv{e}_1\| + \| f(\bv{A})\bv{q}_1-p(\bv{A})\bv{q}_1 \|\nonumber\\
&\le 4 k^3 (C + \delta_k)\norm{\bv{E}}/\gamma+ \|\bv{Q} \| \|f(\bv{T})\bv{e}_1 - p(\bv{T})\bv{e}_1\| + \| f(\bv{A})\bv{q}_1-p(\bv{A})\bv{q}_1 \|\nonumber\\
&\le 4 k^3 (C+\delta_k) \norm{\bv{E}}/\gamma + (\norm{\bv{Q}}+2)\delta_{k}\label{uninstantiatedBound}.
\end{align}
The last inequality follows from the definition of $\delta_{k}$ in \eqref{deltaKCDef} and the fact that all eigenvalues of $\bv{T}$ lie in $[\lmin(\bv{A}) - \eta,\lmax(\bv{A}) + \eta]$ by  \eqref{old_corr} of Theorem \ref{thm:paige_main} since $\eta > 85n^{3/2}k^{5/2} \norm{\bv{A}} \mach$.
By Theorem \ref{thm:paige_main}, we also have $\norm{\bv{q}_i} \le 1 + (n+4)\mach$ for all $i$. This gives $\norm{\bv{Q}} \le \norm{\bv{Q}}_F \le k + k(n+4)\mach$. Further, $\norm{\bv{E}} \le k(2n^{3/2} + 7)\norm{\bv{A}} \mach$. Plugging into \eqref{uninstantiatedBound}, loosely bounding $\delta_k < C$ (since we could always set $p(x) = 0$), and using that $\eta \le \norm{\bv{A}}$ so $\gamma \leq 4 \|\bv{A}\|$, gives \eqref{finalBoundExactF} and thus completes the lemma.
\end{proof}

After scaling by a $\norm{\bv{x}}$ factor, 
Lemma \ref{mainLemmaDetailed} shows that the output $\bv{y} = \norm{\bv{x}} \bv{Q}f(\bv{T})\bv{e}_1$ of Lanczos approximates $f(\bv{A})\bv{x}$ to within a $(k+2)\delta_{k} \norm{\bv{x}}$ factor (plus a lower order term depending on $\mach$), where $\delta_{k}$ is the best approximation given by a degree $< k$ polynomial on the eigenvalue range.
Of course, in finite precision, we cannot exactly compute $\bv{y}$. However, it is known that it is possible to stably compute an eigendecomposition of a symmetric tridiagonal $\bv{T}$ in $\tilde O(n^2)$ time (\cite{gu1995divide}, see Appendix \ref{sec:post}). This allows us to explicitly approximate $f(\bv{T})$ and thus $\bv{y}$. The upshot is our main theorem:
\begin{reptheorem}{mainLemmaFullRuntime}[Function Approximation via Lanczos in Finite Arithmetic] Given real symmetric $\bv{A}\in \R^{n \times n}$, $\bv{x} \in \R^{n}$, $\eta \le \norm{\bv{A}}$, $\epsilon \le 1$, and any function $f$ with $|f(x)| <C$ for $x \in [\lmin(\bv{A})-\eta,\lmax(\bv{A}) + \eta]$, let $B = \log \left (\frac{ nk \norm{\bv{A}}}{\epsilon \eta} \right )$. Suppose Algorithm \ref{alg:lanczos} is run for $k$ iterations on a computer satisfying Requirements \ref{req1} and \ref{req2} with relative precision $\mach = 2^{-\Omega(B)}$ (e.g. on  computer using $\Omega (B)$ bits of precision). If in Step \ref{final_step}, $\bv{y}$ is computed using the eigendecomposition algorithm of \cite{gu1995divide}, it satisfies:
\begin{align}\label{eq:finitePrecisionGuarantee2}
\norm{f(\bv{A}) \bv{x} - \bv{y}} \le (7k \cdot \delta_{k}+\epsilon C) \norm{\bv{x}}
\end{align}
where
\begin{align*}
\delta_{k} \eqdef \min_{\substack{\textnormal{polynomial $p$} \\ \textnormal{ with degree $<k$} }} \left ( \max_{x \in [\lmin(\bv{A})-\eta,\lmax(\bv{A})+\eta]} |p(x)-f(x)| \right ).
\end{align*}
The algorithm's runtime is $O(\mv(\bv{A})k + k^2B + kB^2)$, where $\mv(\bv{A})$ is the time required to multiply $\bv{A}$ by a vector to the precision required by Requirement \ref{req2} (e.g. $O(\nnz(\bv{A})$ time if $\bv{A}$ is given explicitly).
\end{reptheorem}
We note that the dependence on $\eta$ in our bound is typically mild. For example, it $\bv{A}$ is positive semi-definite, if it is possible to find a good polynomial approximation on $[\lmin(\bv{A}),\lmax(\bv{A})]$, it is possible to find an approximation with similar degree on, e.g., $[\frac{1}{2}\lmin(\bv{A}), 2\lmax(\bv{A})]$, in which case $\eta = \Theta(\lmin(\bv{A})|)$. For some functions, we can select an even larger $\eta$ (and thus require fewer bits). For example, in Section \ref{sec:applications} our applications to the matrix step function and matrix exponential both set $\eta = \lmax(\bv{A})$.

\begin{proof}
%
We can apply Lemma \ref{mainLemmaDetailed} to show that:
\begin{align}\label{EndtoEnd1}
\norm{f(\bv{A})\bv{q}_1 - \bv{Q} f(\bv{T}) \bv{e}_1} \le (k+2)\delta_{k} + \mach \cdot \frac{92 C  k^4 n^{3/2} \norm{\bv{A}}}{\lmax(\bv{A}) - \lmin(\bv{A}) + 2\eta}.
\end{align}
The lemma requires $\mach \le \frac{\eta}{85 n^{3/2} k^{5/2} \norm{\bv{A}}}$, which holds since we require $\epsilon \le 1$ and set $\mach = 2^{-\Omega(B)}$ with $B = \log \left (\frac{nk \norm{\bv{A}}}{\epsilon \eta}\right )$. This also ensures that the second term of \eqref{EndtoEnd1} becomes very small, and so we can bound:
\begin{align}\label{EndtoEnd2}
\norm{f(\bv{A})\bv{q}_1 - \bv{Q} f(\bv{T}) \bv{e}_1} \le (k+2) \delta_{k} + \epsilon C/4.
\end{align}

We now show that a similar bound still holds when we compute $\bv{Q}f(\bv{T})\bv{e}_1$ approximately. Via an error analysis of the symmetric tridiagonal eigendecomposition algorithm of Gu and Eisenstat \cite{gu1995divide}, contained in Lemma \ref{lem:postProcess} of Appendix \ref{sec:post}, for any $\epsilon_1 \le 1/2$ with 
\begin{align}\label{treq}
ck^3 \log k \cdot \mach \le \epsilon_1 \le \frac{\eta}{4\norm{\bv{T}}}
\end{align}
 for large enough $c$, in 
$O(k^2 \log \frac{k}{\epsilon_1} + k \log^2\frac{k}{\epsilon_1})$ time we can compute $\bv{z}$ satisfying:
\begin{align}\label{zBound}
\norm{f(\bv{T})\bv{e}_1 - \bv{z}} \le 2 \delta_{k} + \epsilon_1 \cdot \left ( \frac{16 k^3 C  \norm{\bv{T}}}{\lmax(\bv{T}) - \lmin(\bv{T}) + 2\eta} + 16 C \right ).
\end{align}
By our restriction that $\epsilon \le1$ and $\eta \le \norm{\bv{A}}$, since $B = \log \left (\frac{nk \norm{\bv{A}}}{\epsilon \eta}\right )$, we have $\mach = 2^{-\Omega(B)} \le \frac{1}{(nk)^c}$ for some large constant $c$.
This gives $\norm{\bv{T}} \le \norm{\bv{A}} + \mach k^{5/2} \norm{\bv{A}} (68 + 17n^{3/2}) \le 2\norm{\bv{A}}$ by  \eqref{old_corr} of Theorem \ref{thm:paige_main}. 
 Thus, if we set $\epsilon_1 = \left ( \frac{\epsilon\eta }{3nk \norm{\bv{A}}}\right )^c$ for large enough $c$, by \eqref{zBound} we will have:
\begin{align}\label{EndtoEnd3}
\norm{f(\bv{T})\bv{e}_1 - \bv{z}} \le 2 \delta_{k} + \frac{\epsilon C}{4(k+1)}.
\end{align}

Furthermore $\norm{\bv{Q}} \le \norm{\bv{Q}}_F \le k + k(n+4)\mach \le k+1$ by Paige's bounds (Theorem \ref{thm:paige_main}) and the fact that $\mach \le \left (\frac{1}{nk}\right )^c$ for some large $c$. Using \eqref{EndtoEnd3}, this gives:
\begin{align}\label{EndtoEnd4}
\norm{\bv{Q}f(\bv{T})\bv{e}_1 - \bv{Q}\bv{z}} \le (2k+2)\delta_{k} + \epsilon C/4.
\end{align}

As discussed in Section \ref{sec:fpprelim}, if $\bv{Q}\bv{z}$ is computed on a computer satisfying Requirement \ref{req1} then the output $\bv{\bar y}$ satisfies:
$$
\norm{\bv{Qz}- \bv{\bar y}} \le 2\max(n,k)^{3/2} \mach \norm{\bv{Q}}\norm{\bv{z}}.
$$
By \eqref{EndtoEnd3}, $\norm{\bv{z}} \le \norm{f(\bv{T})} + 2 \delta_{k} + \frac{\epsilon C}{4(k+1)} = O(C + \delta_{k}) = O(C)$ since $\delta_k \le C$. Accordingly, by our choice of $\mach$ we can bound $\norm{\bv{Qz}-\bv{\bar y}} \le \epsilon C/4$.
Combining with \eqref{EndtoEnd2} and \eqref{EndtoEnd4} we have:
\begin{align}\label{unscaledBound}
\norm{f(\bv{A})\bv{q}_1 - \bv{\bar y}} \le (3k+ 4) \delta_{k} + 3\epsilon C/4.
\end{align}
This gives the final error bound of \eqref{eq:finitePrecisionGuarantee2} after rescaling by a $\norm{\bv{x}}$ factor.
$\norm{\bv{\bar y}} \le \norm{f(\bv{A})\bv{q}_1} + (3k + 4) \delta_k + 3\epsilon C/4 =  O(kC)$ and so, by our setting of $\mach$, we can compute
$\bv{y} = \norm{\bv{x}} \cdot \bv{\bar y}$ up to additive error $\frac{\epsilon C\cdot \norm{\bv{x}}}{8}$.
 Similarly, we have $\norm{f(\bv{A}) \bv{q}_1 \norm{\bv{x}}-f(\bv{A})\bv{x}} = \frac{\epsilon C\cdot \norm{\bv{x}}}{8}$ even when $\bv{q}_1 = \bv{x}/\norm{\bv{x}}$ is computed approximately. Overall this lets us claim using \eqref{unscaledBound}:
\begin{align*}
\norm{f(\bv{A})\bv{x} - \bv{y}} \le \left [(3k+ 4) \delta_{k} + \epsilon C \right ]\cdot \norm{\bv{x}}
\end{align*}
which gives our final error bound.
The runtime follows from noting that each iteration of Lanczos requires $\mv(\bv{A}) + O(n) = O(\mv(\bv{A}) )$ time. The stable eigendecomposition of $\bv{T}$ up to error $\epsilon_1$ requires $O(k^2 \log \frac{k}{\epsilon_1} + k \log^2 \frac{k}{\epsilon_1}) = O(k^2 B + kB^2)$ time by our setting of $\epsilon_1$.  With this eigendecomposition in hand, computing $\bv{Q}f(\bv{T})\bv{e}_1$ takes an additional $O(nk) = O(\mv(\bv{A})k)$ time.
\end{proof}

\section{Lower bound}\label{sec:lower}
In the previous section, we proved that finite precision Lanczos essentially matches the best known exact arithmetic iteration bounds for \emph{general matrix functions}. These bounds depend on the degree needed to uniformly approximate of $f(x)$ over $[\lmin(\bv{A}), \lmax(\bv{A})]$. We now turn to the special case of positive definite linear systems, where tighter bounds can be shown.
%
%
 
Specifically, equation \eqref{introBound2}, proven in Theorem \ref{exact_lanczos_linear_systems}, shows that the error of Lanczos after $k$ iterations matches the error of the best polynomial approximating $1/x$ at each of $\bv{A}$'s eigenvalues, rather than on the full range $[\lmin(\bv{A}), \lmax(\bv{A})]$. 
Greenbaum proved a natural extension of this bound to the finite precision CG method, showing that its performance matches the best polynomial approximating $1/x$ on tiny ranges around each of $\bv{A}$'s eigenvalues \cite{greenbaum1989behavior}. Recall that ``tiny'' means essentially on the order of machine precision -- the approximation need only be over ranges of width $\eta$ as long as the bits of precision used is $\gtrsim \log(1/\eta)$.
We state a simplified version of this result as Theorem \ref{thm:greenbaum} and provide a full discussion in Appendix \ref{sec:linsystems_app}.

At first glance, Theorem \ref{thm:greenbaum} appears to be a very strong result -- intuitively, approximating $1/x$ on small intervals around each eigenvalue seems much easier than uniform approximation. 

\subsection{Main theorem}

Surprisingly, we show that this is not the case: Greenbaum's result can be much weaker than the exact arithmetic bounds of Theorem \ref{exact_lanczos_linear_systems}. We prove that for any $\kappa$ and interval width $\eta$, there are matrices with condition number $\kappa$ and just $O(\log \kappa \cdot \log 1/\eta)$ eigenvalues for which any `stable approximating polynomial' of the form required by Theorem \ref{thm:greenbaum} achieving error $\bar{\delta}_k \eqdef \le 1/6$ must have degree $\Omega(\kappa^c)$ for a fixed constant $c \ge 1/5$.

This result immediately implies a number of iteration lower bounds on Greenbaum's result, even when we just ask for constant factor approximation to $\bv{A}^{-1} \bv{x}$. See Corollary \ref{cor:lb} and surrounding discussion for a full exposition. As a simple example, setting $\eta = 1/\poly(\kappa)$, our result shows the existence of matrices with $\log^2(\kappa)$ eigenvalues for which Theorem \ref{thm:greenbaum} requires $\Omega(\kappa^c)$ iterations for convergence if $O(\log \kappa)$ bits of precision are used. This is nearly exponentially worse than the $O(\log^2 \kappa)$ iterations required for exact computation of $\bv{A}^{-1}\bv{x}$ in exact arithmetic by \eqref{introBound2}.

\begin{reptheorem}{thm:lb} There exists a fixed constant $1/5 \le c \le 1/2$ such that for any $\kappa \ge 2$, $0 < \eta < \frac{1}{20 \kappa^2}$, and $n \ge \lfloor\log_2 \kappa \rfloor \cdot \lceil \ln 1/\eta \rceil$, there is a positive definite $\bv{A} \in \R^{n \times n}$ with condition number $\le \kappa$, such that for any $k < \lfloor \kappa^{c}/377\rfloor$
\begin{align*}
\bar{\delta}_k \eqdef\min_{\substack{\text{polynomial $p$}\\ \text{ with degree $<k$}}} \left ( \max_{x \in \bigcup_{i=1}^n [\lambda_i(\bv{A})-\eta, \lambda_i(\bv{A})+\eta]} |p(x)-1/x| \right ) \ge 1/6.
\end{align*}
\end{reptheorem}


We prove Theorem \ref{thm:lb} by arguing that there is no polynomial $p(x)$ with degree $\le \kappa^c/377$ which has $p(0) = 1$ and $|p(x)| < 1/3$ for every $x \in \bigcup_{i=1}^n [\lambda_i(\bv{A})-\eta, \lambda_i(\bv{A})+\eta]$.
Specifically, we show:
\begin{lemma}\label{lbIntermediate}
 There exists a fixed constant $1/5 \le c \le 1/2$ and such that for any $\kappa \ge 2$, $0 < \eta \le \frac{1}{20\kappa^2} $ and $n \ge \lfloor \log_2 \kappa \rfloor \cdot \lceil \ln 1/\eta \rceil$, there are $\lambda_1,...,\lambda_n \in [1/\kappa,1]$, such that for any polynomial $p$ with degree $k \le \kappa^{c}/377$ and $p(0) = 1$:
\begin{align*}
\max_{x \in \bigcup_{i=1}^n [\lambda_i-\eta, \lambda_i+\eta]} |p(x)| \ge 1/3.
\end{align*}
\end{lemma}
Lemma \ref{lbIntermediate} can be viewed as an extension of the classic Markov brother's inequality \cite{markovBrothers}, which implies that any polynomial with $p(0) = 1$ and $|p(x)| \le 1/3$ for all $x \in [1/\kappa,1]$ must have degree $\Omega(\sqrt{\kappa})$. Lemma \ref{lbIntermediate}  shows that even if we just restrict $|p(x)| \le 1/3$ on a few small subintervals of $[1/\kappa,1]$, $\Omega(\kappa^c)$ degree is still required. We do not carefully optimize the constant $c$, although we believe it should be possible to improve to nearly $1/2$ (see discussion in Appendix \ref{tighterLowerBound}). This would match the upper bound achieved by the Chebyshev polynomials of the first kind, appropriately shifted and scaled.
Given Lemma \ref{lbIntermediate} it is easy to show Theorem \ref{thm:lb}:
\begin{proof}[Proof of Theorem \ref{thm:lb}]
Let $\bv{A} \in \R^{n\times n}$ be any matrix with eigenvalues equal to $\lambda_1,...,\lambda_n$ -- e.g. a diagonal matrix with these values as its entries.
Assume by way of contradiction that there is a polynomial $p(x)$ with degree $k < \lfloor \kappa^c/377 \rfloor$ which satisfies:
\begin{align*}
\max_{x \in \bigcup_{i=1}^n [\lambda_i(\bv{A})-\eta, \lambda_i(\bv{A})+\eta]} |p(x)-1/x| < 1/6.
\end{align*}
Then if we set $\bar{p}(x) = 1-x p(x)$, $\bar{p}(0) = 1$ and for any $x \in \bigcup_{i=1}^n [\lambda_i-\eta, \lambda_i+\eta]$, 
\begin{align*}
|\bar{p}(x)| \le | x p(x) - 1| \le |x| \cdot |p(x)-1/x| < \frac{|x|}{6} \le \frac{1}{3}
\end{align*}
since $|x| \le 2$ when $\eta \le 1$. Since $\bar{p}(x)$ has degree $k+1 \le \kappa^c/377$, it thus contradicts Lemma \ref{lbIntermediate}.
\end{proof}

\subsection{Hard instance construction}
\label{spectrum_construction}
We begin by describing the ``hard'' eigenvalue distribution that is used to prove Lemma \ref{lbIntermediate} for any given condition number $\kappa \geq 2$ and range radius $\eta$. Define $\lfloor \log_2(\kappa) \rfloor$ intervals: 
$$I_i \eqdef \left [\frac{1}{2^{i}},\frac{1}{2^{i-1}} \right] \text{ for } i = 1,\ldots, \lfloor \log_2(\kappa) \rfloor. $$
In each interval $I_i$ we place $z$ evenly spaced eigenvalues, where:
\begin{align*}
z = \lceil \ln 1/\eta \rceil.
\end{align*}
That is, the eigenvalues in interval $I_i$ are set to:
\begin{align}\label{eigSpacing}
\lambda_{i,j} = \frac{1}{2^{i}} + \frac{j}{z2^{i}} \text{ for } j = 1, \ldots, z.
\end{align}
Thus, our construction uses $\lfloor \log_2 \kappa \rfloor \cdot \lceil \ln 1/\eta \rceil$ eigenvalues total. The smallest is $> \frac{1}{\kappa}$ and  the largest is $\leq 1$, as required in the statement of Lemma \ref{lbIntermediate}. For convenience, we also define:
\begin{align*}
\mathcal{R}_{i,j} &\eqdef [\lambda_{i,j} - \eta, \lambda_{i,j}+\eta] &\mathcal{R}_{i} &\eqdef \bigcup_{j}\mathcal{R}_{i,j} &&\text{ and } &\mathcal{R} &\eqdef \bigcup_{i}\mathcal{R}_{i}.
\end{align*}
By the assumption of Lemma \ref{lbIntermediate} that $\eta \le \frac{1}{20 k^2}$, we have $\eta z = \eta\lceil\ln \frac{1}{\eta}\rceil \leq \sqrt{\eta} + \eta \leq \frac{1}{4\kappa}$.  So none of the $\mathcal{R}_{i,j}$ overlap and in fact are distance at least $\frac{1}{2z\kappa}$ apart (since the eigenvalues themselves have spacing at least $\frac{1}{z\kappa}$ by \eqref{eigSpacing}). An illustration is included in Figure \ref{fig:lower_bound_spect}.


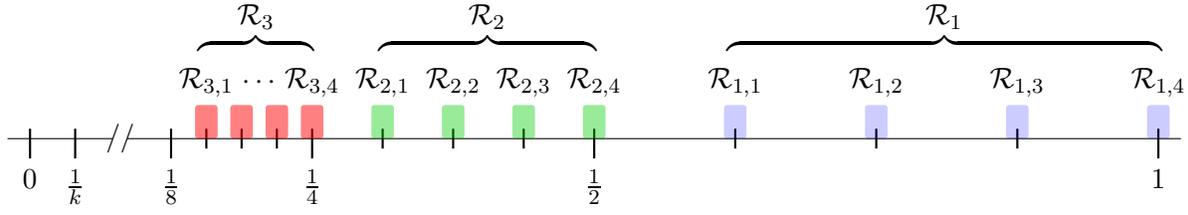
\begin{figure}[H]
\centering
\begin{tikzpicture}[scale=15,par/.style={sloped,fill=white,inner sep=-.6ex}]
\draw (-0.02,0) -- node[par]{//} (.18,0) -- (1.02,0);
\foreach \x/\xtext in {0/0,0.04/$\frac{1}{k}$,0.125/$\frac{1}{8}$,0.25/$\frac{1}{4}$,.5/$\frac{1}{2}$,1/1}
    \draw[thick] (\x,0.25pt) -- (\x,-0.5pt) node[below] {\xtext};   
\foreach \x/\xtext in {0.15625/,0.1875/,0.21875/,0.3125/,0.375/,0.4375/,.625/,.75/,.875/}
    \draw[thick] (\x,0.25pt) -- (\x,-0.25pt) node[below] {\xtext};
\foreach \x/\xtext in {0.15625/$\mathcal{R}_{3,1}$,0.1875/,0.21875/,0.25/$\mathcal{R}_{3,4}$}
{
    \fill[opacity = 0.5, red,rounded corners=.25ex] ({\x-.01},0ex) -- ({\x+.01}, 0ex) -- ({\x+.01}, .18ex) -- ({\x-.01},.18ex) -- cycle;
    \draw[thick] (\x,0.8pt)  node[above] {\xtext};
} 
\draw[thick] (0.20,1.9pt) node[above] {$\overbrace{\hspace{4em}}$};
\draw[thick] (0.20,2.5pt) node[above] {$\mathcal{R}_{3}$};
\draw[thick] (0.205,1pt)  node[above] {$\cdots$};
\foreach \x/\xtext in {0.3125/$\mathcal{R}_{2,1}$,0.375/$\mathcal{R}_{2,2}$,0.4375/$\mathcal{R}_{2,3}$,.5/$\mathcal{R}_{2,4}$}
{
    \fill[opacity = 0.4, black!20!green,rounded corners=.25ex] ({\x-.01},0ex) -- ({\x+.01}, 0ex) -- ({\x+.01}, .18ex) -- ({\x-.01},.18ex) -- cycle;
    \draw[thick] (\x,0.8pt)  node[above] {\xtext};
}  
\draw[thick] (0.405,1.9pt) node[above] {$\overbrace{\hspace{7.5em}}$};
\draw[thick] (0.405,2.5pt) node[above] {$\mathcal{R}_{2}$};
\foreach \x/\xtext in {.625/$\mathcal{R}_{1,1}$,.75/$\mathcal{R}_{1,2}$,.875/$\mathcal{R}_{1,3}$,1/$\mathcal{R}_{1,4}$}
{
    \fill[opacity = 0.2, blue,rounded corners=.25ex] ({\x-.01},0ex) -- ({\x+.01}, 0ex) -- ({\x+.01}, .18ex) -- ({\x-.01},.18ex) -- cycle;
    \draw[thick] (\x,0.8pt)  node[above] {\xtext};
}   
\draw[thick] (0.81,1.9pt) node[above] {$\overbrace{\hspace{15em}}$};
\draw[thick] (0.81,2.5pt) node[above] {$\mathcal{R}_{1}$};
\end{tikzpicture}
\caption{A sample ``hard'' distribution of eigenvalues with $z = 4$. The width of each range $\mathcal{R}_{i,j}$ is over-exaggerated for illustration -- in reality each interval has width $2\eta$, where $\eta \leq \frac{1}{4z\kappa}$.} 
\label{fig:lower_bound_spect}
\end{figure}

\subsection{Outline of the argument}
\label{sec:lb_outline}
Let $p$ be any polynomial with degree $k$ that satisfies $p(0) = 1$. 
To prove Lemma \ref{lbIntermediate} we need to show that we cannot have $|p(x)| \leq 1/3$ for all $x \in \mathcal{R}$ unless $k$ is relatively high (i.e. $\geq \kappa^c$). Let $r_1,\ldots, r_k$
denote $p$'s $k$ roots. So $|p(x)|  = \prod_{i=1}^k | 1 - \frac{x}{r_i}|$. Then define
\begin{align}
\label{g_sum_form}
g(x) \eqdef \ln|p(x)| = \sum_{i=1}^k \ln \left| 1 - \frac{x}{r_i} \right|.
\end{align}
To prove that $|p(x)| \geq 1/3$ for some $x \in \mathcal{R}$, it suffices to show that,
\begin{align}
\label{log_form_goal}
\max_{x \in \mathcal{R}} g(x) \geq -1.
\end{align}
We establish \eqref{log_form_goal} via a potential function argument. For any positive weight function $w(x)$,
\begin{align*}
\max_{x\in\mathcal{R}} g(x) \geq \frac{\int_{\mathcal{R}} w(x) g(x) dx}{\int_{\mathcal{R}} w(x) dx}.
\end{align*}
I.e., any weighted average lower bounds the maximum of a function. From \eqref{g_sum_form}, we have:
\begin{align}
\label{cont_lower_bound_main_disc}
\frac{1}{k}\max_{x\in\mathcal{R}} g(x) \geq \frac{1}{k}\cdot\frac{\int_{\mathcal{R}} w(x) g(x) dx}{\int_{\mathcal{R}} w(x) dx} \geq  \min_{r} \frac{\int_{\mathcal{R}} w(x) \ln|1 - x/r| dx}{\int_{\mathcal{R}} w(x) dx}.
\end{align}
We focus on bounding this last quantity. More specifically, we set $w(x)$ to be:
\begin{align*}
w(x) \eqdef 2^{ic} \text{\hspace{.5em} for \hspace{.5em}} x \in \mathcal{R}_i.
\end{align*}
The weight function increases from a minimum of $\sim 2^c$ to a maximum of $\sim \kappa^c$ as $x \in \mathcal{R}$ decreases from 1 towards $1/\kappa$.
With this weight function, we will be able prove that \eqref{cont_lower_bound_main_disc} is lower bounded by $-O(\frac{1}{\kappa^c})$. 
It will then follow that \eqref{log_form_goal} holds for any polynomial with degree $k = O(\kappa^c)$.


\subsection{Initial Observations}
Before giving the core argument, we make an initial observation that simplifies our analysis:
\begin{claim}
\label{initial_lb_assume}
If Lemma \ref{lbIntermediate} holds for the hard instance described in Section \ref{spectrum_construction} and all real rooted polynomials with roots on the range $[1/\kappa,1 +\eta]$, then it holds for all polynomials.
\end{claim}
\begin{proof} We first show that we can consider just real rooted polynomials, before arguing that we can also assume their roots are within the range $[1/\kappa,1+\eta]$.

\medskip
\noindent\textbf{Real rooted:}
If there is any polynomial equal to $1$ at $x=0$ with magnitude $\leq 1/3$ for $x\in \mathcal{R}$, then there must be a real polynomial (i.e. with real coefficients) of the same degree that only has smaller magnitude on $\mathcal{R}$. So we focus on $p(x)$ with real coefficients.
Letting the roots of $p(x)$ be $r_1,\ldots,r_k$ and using that $p(0) = 1$, we can write:
\begin{align}
\label{mult_form_p}
p(x) = \prod_{i=1}^k (1-x/r_i).
\end{align}
By the complex conjugate root theorem, any polynomial with real coefficients and a complex root must also have its conjugate as a root. Thus, if $p(x)$ has root $\frac{1}{a + bi}$ for some $a,b$, the above product contains a term of the form: 
\begin{align*}
 (1- x (a-bi))(1- x(a+bi)) = 1 - 2 a x + a^2 x^2 + b^2 x^2.
\end{align*}
If we just set $b = 0$ (i.e. take the real part of the root), $1 - 2 a x + a^2 x^2 + b^2 x^2$ decreases for all $x >0$. In fact, since $(1-2ax + a^2 x^2) = (1-ax)^2 > 0$, the absolute value $|1 - 2 a x + a^2 x^2 + b^2 x^2|$ decreases if we set $b = 0$. Accordingly, by removing the complex part of $p$'s complex root, we obtain a polynomial of the same degree that remains $1$ at $0$, but has smaller magnitude everywhere else.

\medskip
\noindent\textbf{Roots in eigenvalue range:}
First note that we can assume $p$ doesn't have any negative roots: removing a term in \eqref{mult_form_p} of the form $(1 - x/r_i)$ for $r_i <0$ produces a polynomial with lower degree that is $1$ at $0$ but smaller in magnitude for all $x > 0$. It is not hard to see that by construction $\mathcal{R} \subseteq [1/\kappa,1+\eta]$ and thus $x > 0$ for all $x \in \mathcal{R}$. Thus removing a negative root can only lead to smaller maximum magnitude over $\mathcal{R}$.

Now, suppose $p$ has some root $0 < r < 1/\kappa$. For all $x \geq 1/\kappa$, 
\begin{align*}
\left|1-\frac{x}{1/\kappa}\right| < \left|1-\frac{x}{r}\right|.
\end{align*}
Accordingly, by replacing $p$'s root at $r$ with one at $(1/\kappa)$ we obtain a polynomial of the same degree that is smaller in magnitude for all $x \geq 1/\kappa$ and thus for all $x \in \mathcal{R} \subseteq [1/\kappa,1+\eta]$.

Similarly, suppose $p$ has some root $r > 1 + \eta$. For all $x \leq 1+\eta$,
\begin{align*}
\left|1-\frac{x}{1 + \eta}\right| < \left|1-\frac{x}{r}\right|.
\end{align*}
So by replacing $p$'s root at $r$ with a root at $(1+\eta)$, we obtain a polynomial that has smaller magnitude everywhere in $\mathcal{R}$.
\end{proof}

\subsection{Main argument}
With Claim \ref{initial_lb_assume}, we are now ready to prove Lemma \ref{lbIntermediate}, which implies Theorem \ref{thm:lb}.
\begin{proof}[Proof of Lemma \ref{lbIntermediate}]
Since we can restrict our attention to real rooted polynomials with each root $r_i \in [\frac{1}{\kappa}, 1+\eta]$, to prove \eqref{log_form_goal} via \eqref{cont_lower_bound_main_disc} we just need to establish that:
\begin{align}
\label{cont_lower_bound_main_spec}
\min_{r \in [\frac{1}{\kappa}, 1+\eta]} \frac{\int_{\mathcal{R}} w(x) \ln|1 - x/r| dx}{\int_{\mathcal{R}} w(x) dx} \geq -\frac{377}{\kappa^c}.
\end{align}
Consider the denominator of the left hand side:
\begin{align*}
\int_{\mathcal{R}} w(x) dx = \sum_{i = 1}^{\lfloor \log_2(\kappa) \rfloor} \int_{\mathcal{R}_i} 2^{ic} dx =  \sum_{i = 1}^{\lfloor \log_2(\kappa) \rfloor} 2\eta z2^{ic} \leq \eta z \kappa^c.
\end{align*}
With this bound in place, to prove \eqref{cont_lower_bound_main_spec} we need to show:
\begin{align*}
\min_{r \in [\frac{1}{\kappa}, 1+\eta]} \int_{\mathcal{R}} w(x) \ln|1 - x/r| dx \geq - 377\eta z.
\end{align*}
Recalling our definition of $\mathcal{R}$, this is equivalent to showing that:
\begin{align}
\label{cont_lower_bound_main_num_only}
\text{For all }&r \in \left[\frac{1}{\kappa}, 1+\eta\right], & \sum_{i = 1}^{\lfloor \log_2(\kappa) \rfloor} \int_{\mathcal{R}_i} w(x) \ln|1 - x/r| dx &\geq - 377\eta z.
\end{align}
To prove \eqref{cont_lower_bound_main_num_only} we divide the sum into three parts. Letting $\lambda_{\ell,h}$ be the eigenvalue closest to $r$:
\begin{align}
\sum_{i = 1}^{\lfloor \log_2(\kappa) \rfloor} \int_{\mathcal{R}_i} w(x) &\ln|1 - x/r| dx = \nonumber\\
\sum_{i = 1}^{\ell-2} &\int_{\mathcal{R}_i} w(x) \ln|1 - x/r| dx \label{int_high}\\
+ \sum_{i = \ell-1}^{\ell+1} &\int_{\mathcal{R}_i} w(x) \ln |1 - x/r| dx \label{int_mid}\\
+ \sum_{i = \ell+2}^{\lfloor \log_2(\kappa) \rfloor} &\int_{\mathcal{R}_i} w(x) \ln|1 - x/r| dx \label{int_low}.
\end{align} 
Note that when $\ell$ lies towards the limits of $\{1, \ldots, {\lfloor \log_2(\kappa) \rfloor}\}$, the sums in  \eqref{int_low} and \eqref{int_high} may contain no terms and \eqref{int_mid} may contain less than 3 terms. 

To gain a better understanding of each of these terms, consider Figure \ref{region_plot}, which plots $\ln|1 - x/r|$ for an example value of $r$. \eqref{int_high} is a weighted integral over regions $\mathcal{R}_i$ that lie well above $r$. Specifically, for all $x \in \bigcup_{i = 1}^{ \ell-2} \mathcal{R}_i$, $x \geq 2r$ and thus $\ln|1 - x/r|$ is \emph{strictly positive}. Accordingly, \eqref{int_high} is a positive term and will help in our effort to lower bound \eqref{cont_lower_bound_main_num_only}.

On the other hand, \eqref{int_mid} and \eqref{int_low} involve values of $x$ which are close to $r$ or lie below the root. For these values, $\ln|1 - x/r|$ is \emph{negative} and thus \eqref{int_mid} and \eqref{int_low} will hurt our effort to lower bound \eqref{cont_lower_bound_main_num_only}. We need to show that the negative contribution cannot be too large.

\begin{figure}[H]
\centering
\includegraphics[width=.65\textwidth]{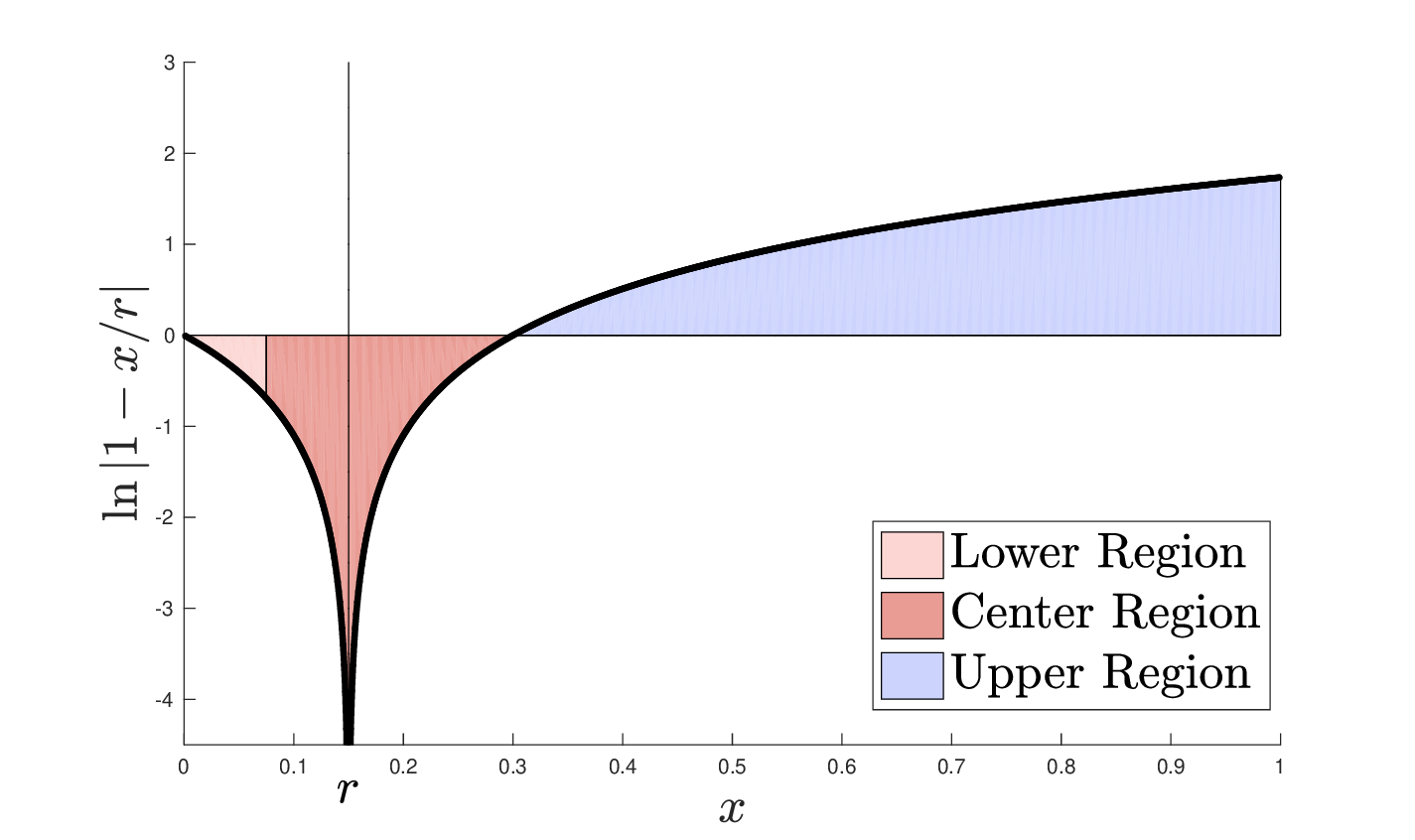} 
\caption{Plot of $\ln|1 - x/r|$ for $r = 1/10$. Proving that \eqref{cont_lower_bound_main_num_only} cannot be too small for any root $r$ requires lower bounding a weighted integral of this function over $\mathcal{R} \subset [1/\kappa,1+\eta]$.}
\label{region_plot}
\end{figure}

\subsubsection{Center region}
We first evaluate \eqref{int_mid}, which is the range containing eigenvalues close to $r$. In particular, we start by just considering $\mathcal{R}_{\ell,h}$, the interval around the eigenvalue nearest to $r$.
\begin{align*}
\int_{\mathcal{R}_{\ell,h}} w(x) \ln|1 - x/r| dx  =  2^{\ell c}\int_{\lambda_{\ell,h} - \eta}^{\lambda_{\ell,h} + \eta} \ln|1 - x/r| \geq 2^{\ell c} \int_{\lambda_{\ell,h} - \eta}^{\lambda_{\ell,h} +\eta} \ln|1 - x/\lambda_{\ell,h}|.
\end{align*} 
The inequality follows because $\ln|1-x/r|$ strictly increases as $x$ moves away from $r$. Accordingly, the integral takes on its minimum value when $r$ is centered in the interval $[\lambda_{\ell,h} -\eta, \lambda_{\ell,h} +\eta]$.  

\begin{align*}
2^{\ell c} \int_{\lambda_{\ell,h} - \eta}^{\lambda_{\ell,h} +\eta} \ln|1 - x/\lambda_{\ell,h}| = 2^{\ell c + 1} \int_{0}^{\eta} \ln \frac{x}{\lambda_{\ell,h}} &= 2^{\ell c + 1} \eta \left(\ln\eta - \ln\lambda_{\ell,h} - 1\right). 
\end{align*}
Since $\ln(\eta) \leq -1$ by the assumption that $\eta \le \frac{1}{20 \kappa^2}$ and since $- \ln\lambda_{\ell,h} \geq 0$ since $\lambda_{\ell,h} \le 1$, we  obtain:
\begin{align}
\label{bound_for_closest_bucket}
\int_{\mathcal{R}_{\ell,h}} w(x) \ln|1 - x/r| dx \geq 4\cdot 2^{\ell c}\eta \ln \eta \geq -4\cdot 2^{\ell c}\eta z.
\end{align}

Now we consider the integral over $\mathcal{R}_{\ell,i}$ for all $i \neq h$ and also over the entirety of $\mathcal{R}_{\ell+1}$ and $\mathcal{R}_{\ell-1}$.  
For all $x \in \left[\mathcal{R}_{\ell+1}\cup(\mathcal{R}_{\ell} \setminus\mathcal{R}_{\ell,h})\cup\mathcal{R}_{\ell-1}\right]$, $w(x) \leq 2^{(\ell+1)c} \leq 2^{1/5} \cdot 2^{lc}$ since $c \ge 1/5$. So we have:
\begin{align}\label{weightBound}
\int_{\mathcal{R}_{\ell+1}\cup(\mathcal{R}_{\ell} \setminus \mathcal{R}_{\ell,h})\cup\mathcal{R}_{\ell-1}} w(x) \ln|1 - x/r| dx &\geq \int_{\mathcal{R}_{\ell+1}\cup(\mathcal{R}_{\ell} \setminus \mathcal{R}_{\ell,h})\cup\mathcal{R}_{\ell-1}} w(x) \min(\ln|1 - x/r|,0) dx\nonumber\\
\geq
&2^{1/5} \cdot 2^{lc} \int_{\mathcal{R}_{\ell+1}\cup(\mathcal{R}_{\ell} \setminus \mathcal{R}_{\ell,h})\cup\mathcal{R}_{\ell-1}} \min(\ln|1 - x/r|,0) dx. 
\end{align}
where the last inequality holds by our bound on $w(x)$ and since $\min(\ln|1 - x/r|,0)$ is nonpositive.

The nearest eigenvalue to $\lambda_{\ell,h}$ is $\frac{1}{2^\ell z}$ away from it. Thus, the second closest eigenvalue to $r$ besides $\lambda_{\ell,h}$ is at least $\frac{1}{2^{\ell+1}z}$ away from $r$. By our assumption that $\eta \le \frac{1}{20\kappa^2}$, as discussed we have $\eta \leq \frac{1}{4\kappa z} \le \frac{1}{2^{\ell+2} z}$. Thus, the closest interval to $r$ besides $\mathcal{R}_{\ell,h}$ is at least $\frac{1}{2^{\ell+1}z} - \frac{1}{2^{\ell+2}z} \ge \frac{r}{8z}$ away.

Thus, using that again that $\ln|1 - x/r| $ is strictly increasing as $x$ moves away from $r$, that there are $3z-1$ eigenvalues in $\mathcal{R}_{\ell+1}\cup(\mathcal{R}_{\ell} \setminus \mathcal{R}_{\ell,h})\cup\mathcal{R}_{\ell-1}$, and \eqref{weightBound} we can lower bound the integral by:
\begin{align*}
\int_{\mathcal{R}_{\ell+1}\cup(\mathcal{R}_{\ell} \setminus \mathcal{R}_{\ell,h})\cup\mathcal{R}_{\ell-1}} w(x) &\ln|1 - x/r| dx \\
&\ge 2^{1/5}\cdot 2^{lc} \cdot 2\eta \sum_{i \in \{-\lfloor 1.5 z \rfloor,...,\lfloor 1.5z \rfloor\}\setminus 0} \min \left (\ln \left | 1 - \frac{r(1+\frac{i}{8z})}{r} \right |,0 \right)\\
&\geq 4\cdot 2^{1/5} \eta \cdot 2^{lc} \sum_{i=1}^{\lfloor 1.5z \rfloor} \min(\ln(i/8z),0) \\
&\geq 4\cdot 2^{1/5} \eta \cdot 2^{lc} \int_{x=0}^{1.5z} \ln(x/8z) dx \\
&\geq - 18.5 \cdot 2^{lc}\eta z.
\end{align*}
This bound combines with \eqref{bound_for_closest_bucket} to obtain a final lower bound on \eqref{int_mid} of:
\begin{align}
\label{int_mid_bound}
\sum_{i = \ell-1}^{\ell+1} \int_{\mathcal{R}_i} w(x) \ln|1 - x/r| dx &\geq -22.5\cdot2^{lc}\eta z.
\end{align}

\subsubsection{Lower region}
Next consider \eqref{int_low}, which involves values that are at least a factor of $2$ smaller than $r$. We have:
\begin{align*}
\text{For }  j \geq 2 \text{ and } x \in \mathcal{R}_{\ell+j} \text{,\hspace{1em}} & \ln \left|1 - \frac{x}{r}\right| \geq \ln\left(1 - \frac{1}{2^{j-1}}\right) \geq - \frac{1.39}{2^{j-1}}.
\end{align*}
For the last bound we use that $\frac{1}{2^{j-1}} \leq \frac{1}{2}$.
It follows that:
\begin{align}\label{556Bound}
\sum_{i = \ell+2}^{\lfloor \log_2(\kappa) \rfloor} &\int_{\mathcal{R}_i} w(x) \ln|1 - x/r| dx \geq \sum_{i = \ell+2}^{\lfloor \log_2(\kappa) \rfloor} - 2.78\eta z\cdot\frac{2^{ic}}{2^{i-\ell-1}} = - 5.56\cdot2^{\ell c}\eta z \sum_{j = 2}^{\lfloor \log_2(\kappa) \rfloor - \ell} \frac{1}{2^{j(1-c)}}.
\end{align}
Since we restrict $c \ge 1/5$, the sum above (which is positive) is at most:
\begin{align*}
 \sum_{j = 2}^{\lfloor \log_2(\kappa) \rfloor - \ell} \frac{1}{2^{j(1-c)}} \leq \frac{1}{2^{8/5}} \cdot \frac{1}{1-\frac{1}{2^{4/5}}} \le .8
\end{align*}
So we conclude using \eqref{556Bound} that:
\begin{align}
\label{int_low_bound}
\sum_{i = \ell+2}^{\lfloor \log_2(\kappa) \rfloor} &\int_{\mathcal{R}_i} w(x) \ln|1 - x/r| dx  > - 4.5 \cdot 2^{\ell c} \eta z.
\end{align}

\subsubsection{Upper region}

From \eqref{int_mid_bound} and \eqref{int_low_bound}, we see that \eqref{int_mid} and \eqref{int_low} sum to $- O(2^{\ell c} \eta z)$. Recall that we wanted the entirety of \eqref{int_high} + \eqref{int_mid} + \eqref{int_low} to sum to something greater than $- O(\eta z)$. For large values of $\ell$ (i.e., when $r$ is small), the $2^{\ell c}$ term is problematic. It could be on the order $- \kappa^c$. If this is the case, we need to rely on a positive value of \eqref{int_high} to cancel out the negative contribution of \eqref{int_mid} and \eqref{int_low}. Fortunately, from the intuition provided by Figure \ref{region_plot}, we expect \eqref{int_high} to increase as $r$ decreases.

We start by noting that:
\begin{align*}
\text{For }  j \geq 2 \text{ and } x \in \mathcal{R}_{\ell-j} \text{,\hspace{1em}} & \ln|1 - \frac{x}{r}| \geq \ln\left(2^{j-1} - 1\right) \geq \frac{j-2}{2}.
\end{align*}
It follows that 
\begin{align}\label{preHighBound}
\sum_{i = 1}^{\ell-2} \int_{\mathcal{R}_i} w(x) \ln|1 - x/r| dx&\geq \sum_{i = 1}^{\ell-2} 2\eta \cdot 2^{ic} \cdot \frac{\ell-i-2}{2} =  2^{\ell c} \eta z  \sum_{i = 1}^{\ell-2} \frac{\ell - i -2}{2^{c \left(\ell - i\right)}}.
\end{align}
By our requirement that $c \ge 1/5$, as long as $\ell \ge 20$ we can explicitly compute:
\begin{align}\label{tightnessSum}
\sum_{i = 1}^{\ell-2} \frac{\ell  - i - 2}{2^{c \left(\ell - i\right)}}  = \frac{1}{2^{3c}} + \frac{2}{2^{4c}} + \ldots +  \frac{\ell - 3}{2^{c(\ell - 1)}}
&\ge 27.4
\end{align}
which finally gives, using \eqref{preHighBound}:
\begin{align}\label{int_high_bound}
\sum_{i = 1}^{\ell-2} \int_{\mathcal{R}_i} w(x) \ln|1 - x/r| dx \ge 27.4 \cdot 2^{\ell c} \eta z.
\end{align}
We note for the interested reader that \eqref{preHighBound} is the reason that we cannot set $c$ too large (e.g. $c \ge 1/2$). If $c$ is too large, the sum in \eqref{tightnessSum} will be small, and will not be enough to cancel out the negative contributions from the center and lower regions.

\subsection{Putting it all together}
We can bound \eqref{cont_lower_bound_main_num_only} using our bounds on the upper region \eqref{int_high} (given in \eqref{int_high_bound}), the center region \eqref{int_mid} (given in \eqref{int_mid_bound}) and the lower region \eqref{int_low} (given in \eqref{int_low_bound}).
As long as $\ell \geq 20$ we have: 
\begin{align*}
\int_{\mathcal{R}} w(x) \ln|1 - x/r| dx \geq (-22.5 - 4.5 + 27.4) \cdot 2^{\ell c} \eta z \ge 0 \geq -\eta z.
\end{align*}

It remains to handle the case of $\ell < 20$. In this case, the concerning $2^{\ell c}$ term is not a problem. Specifically, when $\ell < 20$ we have $2^{\ell c} \le 2^{19/5}.$
%
Even ignoring the positive contribution of \eqref{int_high}, we can thus lower bound \eqref{cont_lower_bound_main_num_only}  using our center and lower region bounds by:
\begin{align*}
 \int_{\mathcal{R}} w(x) \log|1 - x/r| dx \geq (-22.5-4.5) \cdot 2^{19/5} \cdot \eta z \ge -377\eta z
\end{align*}
and it follows that \eqref{cont_lower_bound_main_spec} is lower bounded by
\begin{align*}
\min_{r \in [\frac{1}{\kappa}, 1+\eta]} \frac{\int_{\mathcal{R}} w(x) \log|1 - x/r| dx}{\int_{\mathcal{R}} w(x) dx} \geq -\frac{377}{\kappa^c}.
\end{align*}

Then, by the argument outlined in Section \ref{sec:lb_outline}, for any $k \leq \frac{\kappa^c}{377}$, there is no real rooted, degree $k$ polynomial $p$ with roots in $[\frac{1}{\kappa}, 1+\eta]$ such that:
\begin{align*}
p(0) &= 1 &&\text{and} & \log|p(x)| \leq -1 \text{ for all } x\in \mathcal{R}.
\end{align*}
Finally, applying Claim \ref{initial_lb_assume} proves Lemma \ref{lbIntermediate}, as desired.
\end{proof}

\section{Applications}\label{sec:applications}

In this section, we give example applications of Theorem \ref{mainLemmaFullRuntime} to matrix step function, matrix exponential, and top singular value approximation. We also show how Lanczos can be used to accelerate the computation of any function which is well approximated by a high degree polynomial with bounded coefficients. For each application, we show that Lanczos either improves upon or matches state-of-the-art runtimes, even when computations are performed with limited precision.

\subsection{Matrix step function approximation}
\label{sub:step}

In many applications it is necessary to compute the matrix step function $s_\lambda(\bv{A})$ where 
\begin{align*}
s_\lambda(x) \eqdef
\begin{cases}
0 \text{ for } x < \lambda\\
1 \text{ for } x \ge \lambda.
\end{cases}
\end{align*}
Computing $s_\lambda(\bv{A})\bv{x}$ is equivalent to projecting $\bv{x}$ onto the span of all eigenvectors of $\bv{A}$ with eigenvalue $\ge \lambda$. This projection is useful in data analysis algorithms that preprocess data points by projecting onto the top principal components of the data set -- here $\bv{A}$ would be the data covariance matrix, whose eigenvectors correspond to principal components of the data.
For example, as shown in \cite{frostig2016principal} and \cite{allen2016faster}, an algorithm for approximating $s_\lambda(\bv{A})\bv{x}$ can be used to efficiently solve the \emph{principal component regression} problem, a widely used form of regularized regression.
A projection algorithm can also be used to accelerate spectral clustering methods \cite{tremblay2016compressive}.

The matrix step function
$s_\lambda(\bv{A})$ is also useful because $\tr(s_\lambda(\bv{A}))$ is equal to the number of eigenvalues of $\bv{A}$ which are $ \ge \lambda$. This trace can be estimated up to $(1\pm \epsilon)$ relative error with probability $1-\delta$ by computing $\bv{x}^T s_\lambda(\bv{A}) \bv{x}$ for $O(\log(1/\delta)/\epsilon^2)$ random sign vectors \cite{hutchinson1990stochastic}. By composing step functions at different thresholds and using this trace estimation technique, it is possible to estimate the number of eigenvalues of $\bv{A}$ in any interval $[a,b]$, which is a useful primitive in estimating numerical rank \cite{ubaru2016fast}, tuning eigensolvers and other algorithms \cite{di2016efficient}, and estimating the value of matrix norms \cite{musco2017spectrum}.

\subsubsection{Soft step function application via Lanczos}
Due to its discontinuity at $\lambda$, $s_\lambda(x)$ cannot be uniformly approximated on the range of $\bv{A}$'s eigenvalues by any polynomial. Thus, we cannot apply Theorem \ref{mainLemmaFullRuntime} directly. However, it typically suffices to apply a softened step function that is allowed to deviate from the true step function in a small range around $\lambda$. For simplicity, we focus on applying such a function with $\lambda = 0$. Specifically, we wish to apply $h(\bv{A})$ where:
\begin{align}\label{softSign}
h(x) \in
\begin{cases}
 [0,\epsilon] \text{ for } x < -\gamma\\
   [0,1] \text{ for } x \in [-\gamma,\gamma] \\
 [1 - \epsilon,1] \text{ for } x \ge \gamma.
\end{cases}
\end{align}
For a positive semidefinite $\bv{A}$, by applying $h$ to $\bv{B} = \bv{A}(\bv{A}+\lambda \bv{I})^{-1} - \frac{1}{2}\bv{I}$, we can recover a soft step function at $\lambda$, which, for example, provably suffices to solve principal component regression \cite{frostig2016principal,allen2016faster} and to perform the norm estimation algorithms of \cite{musco2017spectrum}. We just need to apply $\bv{B}$ to the precision specified in Requirement \ref{req2}, which can be done, for example, using a fast iterative linear system solver. 

In  \cite{frostig2016principal}, Corollary 5.4, it is proven that for $q = O(\gamma^{-2} \log(1/\epsilon))$ the polynomial:
\begin{align}\label{slowPolynomial}
p_q(x) = \sum_{i=0}^q \left (x(1-x^2)^i \prod_{j=1}^i \frac{2j-1}{2j} \right )
\end{align}
is a valid softened sign function satisfying \eqref{softSign}. Additionally, it is shown in Lemma 5.5 that there is a lower degree polynomial $p^*(x)$ with degree $O(\gamma^{-1} \log(1/\epsilon\gamma))$ which uniformly approximates $p_q(x)$ to error $\epsilon$ on the range $[-1,1]$. Combining these two results we can apply Theorem \ref{mainLemmaFullRuntime} to obtain:
\begin{theorem}[Approximation of soft matrix sign function]\label{thm:step}
Given $\bv{B} \in \R^{n \times n}$ with $\norm{\bv{B}} \le 1/2$, $\bv{x} \in \R^n$, and $\epsilon < 1$, let $B = \log\left(\frac{n}{\epsilon\gamma}\right)$ and $q = O(\gamma^{-2} \log(1/\epsilon))$. Suppose Algorithm \ref{alg:lanczos} is run with $f(x) = p_q(x)$, which is a function satisfying \eqref{softSign}. After $k = O(\gamma^{-1} \log(1/\epsilon\gamma))$ iterations on a computer satisfying Requirement \ref{req1} and Requirement \ref{req2} for applying $\bv{B}$ to precision $\mach = 2^{-\Omega(B)}$ (e.g. a computer with $\Omega(B)$ bits of precision), the algorithm outputs $\bv{y}$ with $\norm{p_q(\bv{B}) \bv{x} - \bv{y}} \le \epsilon \norm{\bv{x}}$. The total runtime is $O(\mv(\bv{B})k + k^2 B + kB^2)$.
\end{theorem}

Note that the assumption $\norm{\bv{B}} \le \frac{1}{2}$ allows us to set $\eta = \Theta(\norm{\bv{B}})$ and still have $[\lmin(\bv{B}) - \eta ,\lmax(\bv{B}) + \eta] \subseteq [-1,1]$, so we can apply the uniform approximation bound of \cite{frostig2016principal}. If we apply Theorem \ref{thm:step} to $\bv{B} = \bv{A}(\bv{A}+\lambda \bv{I})^{-1} - \frac{1}{2}\bv{I}$ for PSD $\bv{A}$ to compute the step function at $\lambda$, the assumption $\norm{\bv{B}} \le 1/2$ holds.

\subsubsection{Comparision with prior work}
\cite{allen2016faster} shows how to directly apply a polynomial with degree $O(\gamma^{-1} \log(1/\epsilon \gamma))$ which approximates a softened sign function. Furthermore, this application can be made stable using the stable recurrence for Chebyshev polynomial computation, and thus matches Theorem \ref{thm:step}. Both \cite{frostig2016principal} and \cite{allen2016faster} acknowledge Lanczos as a standard method for applying matrix sign functions, but avoid the method due to the lack of a complete theory for its approximation quality. Theorem \ref{thm:step} demonstrates that end-to-end runtime bounds can in fact be achieved for the Lanczos method, matching the state-of-the-art given in \cite{allen2016faster}.

\subsection{Matrix exponential approximation}\label{sec:exponential}

We next consider the matrix exponential, which is applied widely in numerical computation, theoretical computer science, and machine learning. For example, computing $\exp(-\bv{A}) \bv{x}$ for a PSD $\bv{A}$ is an important step in the matrix multiplicative weights method for semidefinite programming \cite{arora2005fast,kale2007efficient} and in the balanced separator algorithm of \cite{matrixExp}. 
When $\bv{A}$ is a graph adjacency matrix, $\tr(\exp(\bv{A}))$ is known as the \emph{Estrada index}. As in the case of the sign function, its value can be estimated to $(1\pm \epsilon)$ multiplicative error with probability $1-\delta$ if $\exp(\bv{A})$ is applied to $O(\log(1/\delta) \epsilon^{-2})$ random vectors  \cite{han2016approximating}.

Approximating the matrix exponential, including via the Lanczos method \cite{saad1992analysis,druskin1998using}, has been widely studied -- see \cite{matExp} for a review. Here we use our results to give general end-to-end runtime bounds for this problem in finite precision, which as far as we know are state-of-the-art.

\subsubsection{Approximation of $\exp(\bv{A})$ for general $\bv{A}$}
We can apply Theorem \ref{mainLemmaFullRuntime} directly to the matrix exponential. $\exp(x)$ can be uniformly approximated up to error $\epsilon$ for $x\in[a,b]$ with a truncated Taylor series expansion at $(b+a)/2$ with degree $k =  O((b-a) + \log(e^{b+a}/\epsilon))$ (see e.g. Lemma 7.5 of \cite{matrixExp} with $\delta$ set to $\delta = \epsilon \cdot e^{-(b+a)/2}$). 
Applying Theorem \ref{mainLemmaFullRuntime} with $\eta = \norm{\bv{A}}$ we have:
\begin{theorem}[General matrix exponential approximation]\label{generalExp}
Given symmetric $\bv{A} \in \R^{n \times n}$, $\bv{x} \in \R^n$, and $\epsilon \le 1$, let $B = \log  \left ( \frac{\max(n,\norm{\bv{A}})}{\epsilon} \right )$. If Algorithm \ref{alg:lanczos} is run with $f(x) = \exp(x)$ for $k = O(\norm{\bv{A}}+ \log(1/\epsilon))$ iterations on a computer satisfying Requirements \ref{req1} and \ref{req2} for precision $\mach = 2^{-\Omega(B)}$ (e.g. a computer using $\Omega(B)$ bits of precision), it outputs $\bv{y}$ satisfying $\norm{\exp(\bv{A}) \bv{x} - \bv{y}} \le \epsilon C \norm{\bv{x}}$ where $C = e^{2 \norm{\bv{A}}}$. The total runtime is $O(\mv(\bv{A}) k + k^2 B + kB^2)$.
\end{theorem}
 

\subsubsection{Approximation of $\exp(-\bv{A})$ for positive semidefinite $\bv{A}$}

In applications such as to the matrix multiplicative weights update method and the balanced separator algorithm of \cite{matrixExp}, we are interesting in computing $\exp(-\bv{A})$ for positive semidefinite $\bv{A}$. In this case a better bound is achievable. Using Theorem 7.1 of \cite{matrixExp}, the linear dependence on $\norm{\bv{A}}$ in the iterations required for Theorem \ref{generalExp} can be improved to $\tilde O(\sqrt{\norm{\bv{A}}})$. Additionally, since $-\bv{A}$ has only non-positive eigenvalues, we can set $C = O(1)$.

However, the runtime of Lanczos still has a $\tilde O(k^2)$ term. We can significantly reduce $k$ and thus improve this cost via the rational approximation technique used in \cite{matrixExp}. Specifically, $\exp(\bv{A})$ can be approximated via a $k = O(\log(1/\epsilon))$ degree polynomial in $(\bv{I} + \frac{1}{k} \bv{A})^{-1}$. Further, our stability results immediately imply that it suffices to compute an approximation to this inverse, using e.g. the conjugate gradient method. Specifically we have:

\begin{theorem}[Improved matrix exponential approximation]\label{improvedExp}
Given PSD $\bv{A} \in \R^{n \times n}$, $\bv{x} \in \R^n$, and $\epsilon < 1$, let $k = O(\log(1/\epsilon))$, $B = \log \left ( \frac{n \max(\norm{\bv{A}},1)}{\epsilon} \right )$, and $\epsilon_1 = 2^{-\Omega(B)}$. Let $\alga(\bv{A},\bv{w},k,\epsilon_1)$ be an algorithm returning $\bv{z}$ with $\norm{ \left (\bv{I} + \frac{1}{k}\bv{A}\right)^{-1} \bv{w} - \bv{z}} \le \epsilon_1 \norm{\bv{w}}$ for any $\bv{w}$. There is an algorithm running on a computer with $\Omega(B)$ bits of precision that makes $k$ calls to $\alga(\bv{A},\bv{w},k,\epsilon_1)$ and uses $O(nk + kB^2)$ additional time to return $\bv{y}$ satisfying: $\norm{\exp(-\bv{A})\bv{x} - \bv{y}} \le \epsilon \norm{\bv{x}}$.
\end{theorem}
Theorem \ref{improvedExp} can be compared to Theorem 6.1 of \cite{matrixExp}. It has an improved dependence on $k$ since the modified Lanczos algorithm used in \cite{matrixExp} employs reorthogonalization at each iteration and thus incurs a cost of $O(nk^2)$. Additionally, \cite{matrixExp} focuses on handling error due to the approximate application of $(\bv{I} + \frac{1}{k}\bv{A})^{-1}$, but assumes exact arithmetic for all other operations.

\begin{proof}
We apply Theorem \ref{mainLemmaFullRuntime} with matrix $\bv{B} = (\bv{I} + \frac{1}{k} \bv{A})^{-1}$ and $f(x) = e^{-k/x+k}.$ We can write $\bv{B} = g(\bv{A})$ where $g(x) = \frac{1}{1+x/k}$ and thus have $f(\bv{B}) = \exp(-\bv{A})$ since $f(g(x)) = e^{-x}$. 

Additionally, $\bv{B}$'s eigenvalues all fall between $\frac{k}{\norm{\bv{A}}}$ and $1$. Set $\eta = \min(\frac{k}{2\norm{\bv{A}}},\frac{\epsilon}{ck^3}) \le \norm{\bv{B}}$ for sufficiently large constant $c$. Then for all $x \in [\lmin(\bv{B}) - \eta, \lmax(\bv{B}) + \eta]$, we can loosely bound:
\begin{align}\label{magnitudeBound}
|f(x)| \le e^{-k/(1+\eta) + k} \le e^{-k/(1+1/k) + k} \le e.
\end{align}
By Corollary 6.9 of \cite{matrixExp}, there is a degree $k$ polynomial $p^*(x)$ satisfying $p^*(0) = 0$ and:
\begin{align}\label{expPolyApprox}
\sup_{x \in (0,1]} \left | f(x) - p^*(x) \right | = O(k \cdot 2^{-k}).
\end{align}

We need to bound the error of approximation on the slightly extended range $[\lmin(\bv{B})-\eta, \lmax(\bv{B}) + \eta] \subset (0, 1+\eta]$. We do this simply by arguing that $f(x)$ and $p^*(x)$ cannot diverge substantially over the range $[1,1+\eta]$. In this range we can bound $f(x)$: 
\begin{align}\label{fChangeBound}
1 < e^{-k/x+k} \le e^{\frac{-k}{1+\epsilon/ck^3}+k} \le e^{\frac{-ck^4 + ck^4 + k\epsilon}{ck^3+\epsilon}} \le e^{\epsilon/(ck^2)} \le 1+ O \left (\frac{\epsilon}{k^2} \right ).
\end{align}

Additionally, by the Markov brother's inequality, any degree $k$ polynomial $p(x)$ with $|p(x)| \le 1$ for $x \in [-1,1]$ has derivative $p'(x) \le k^2$ on the same range.
By \eqref{expPolyApprox}, if we set $k = c\log(1/\epsilon)$ for large enough constant $c$, we have $\sup_{x \in (0,1]} \left | f(x) - p^*(x) \right | = O\left (\frac{\epsilon}{k} \right )$. Since $f(x) \le 1$ on this range, we thus loosely have $|p^*(x)| \le 2$ for $x \in [0,1]$. We can then claim that $p^*$ changes by at most $O \left ( \frac{\epsilon}{k} \right )$ on  $[1,1+\eta]$, which has width $O \left (\frac{\epsilon}{k^3}\right)$. Otherwise, $p^*$ would have derivative $\ge c k^2$ for some constant $c$ at some point in this range, contradicting Markov's inequality after appropriately shifting and scaling $p^*$ to have magnitude bounded by $1$ on $[-1,1]$.
Overall, combined with \eqref{fChangeBound} we have:
\begin{align*}
\delta_k \le \max_{x \in [\lmin(\bv{B})-\eta, \lmax(\bv{B}) + \eta]} |f(x) - p(x)| = O \left ( \frac{\epsilon}{k} \right ).
\end{align*}
Theorem \ref{mainLemmaFullRuntime} applies with $C = e$ from \eqref{magnitudeBound}, $k = O(\log(1/\epsilon))$ and $\eta = \min(\frac{k}{2\norm{\bv{A}}},\frac{\epsilon}{ck^3})$ as long as we use $\Omega \left (\log \left (\frac{ nk \norm{\bv{B}}}{\epsilon\eta} \right )\right) =  \Omega \left ( \log \left ( \frac{n \max(\norm{\bv{A}},1)}{\epsilon} \right ) \right )$ bits of precision (to satisfy Requirement \ref{req1}) and can compute $\bv{B}\bv{w}$ up to error $\epsilon_1\|\bv{w}\|$ for any $\bv{w}$ (to satisfy Requirement \ref{req2}). Accordingly,
\begin{align*}
\norm{f(\bv{B})\bv{x} - \bv{y}} = \norm{\exp(-\bv{A})\bv{x} - \bv{y}} \le \epsilon \norm{\bv{x}}.
\end{align*}
\end{proof}

For the balanced separator algorithm of \cite{matrixExp}, the linear system solver $\alga(\bv{A},\bv{w},k,\epsilon_1)$ can be implemented used a fast, near linear time Laplacian system solver. For general matrices, it can be implemented via the conjugate gradient method. Applying Theorem \ref{thm:greenbaum} to $\bv{B} = (\bv{I} + \frac{1}{k} \bv{A})$, setting $\eta = \lmin(\bv{B})/2 \ge 1/2$ and $k = O \left (\log(\kappa(\bv{B})/\epsilon_1) \cdot \sqrt{\kappa(\bv{B})} \right)$ ensures that CG computes $\bv{y}$ satisfying $\norm{ \left (\bv{I} + \frac{1}{k}\bv{A}\right)^{-1} \bv{w} - \bv{z}} \le \epsilon_1 \norm{\bv{w}}$ if $\Omega \left (\log (n \kappa(\bv{B})/\epsilon_1)\right )$ bits of precision are used. Additionally, we can multiply a vector by $\bv{B}$ in time $\mv(\bv{B}) = \mv(\bv{A}) + n$.
Plugging in $\kappa(\bv{B}) \le 1+ \norm{\bv{A}}/k$ gives:
\begin{corollary} Given PSD $\bv{A} \in \R^{n \times n}$, $\bv{x} \in \R^n$, and $\epsilon < 1$, there exists an algorithm running on a computer with $B = \Omega \left ( \log \left ( \frac{n \max(\norm{\bv{A}},1)}{\epsilon} \right ) \right )$ bits of precision which returns $\bv{y}$ satisfying $\norm{\exp(-\bv{A})\bv{x} - \bv{y}} \le \epsilon \norm{\bv{x}}$ in $O \left (\left[(\mv(\bv{A})+n) \log \left (\frac{n\max(\norm{\bv{A}},1)}{\epsilon} \right ) \sqrt{\frac{\norm{\bv{A}}}{\log(1/\epsilon)}+1} + \log^2 \left (\frac{n\max(\norm{\bv{A}},1)}{\epsilon} \right )\right]\cdot\log\frac{1}{\epsilon} \right )$ time.
\end{corollary}

\subsection{Top singular value approximation}
\label{sub:sing}

Beyond applications to matrix functions, the Lanczos method and related Krylov subspace methods are the most common iterative algorithms for computing approximate eigenvectors and eigenvalues of symmetric matrices. Once $\bv{Q}$ and $\bv{T}$ are obtained by Algorithm \ref{alg:lanczos} (or a variant) the \emph{Rayleigh-Ritz method} can be used to find approximate eigenpairs for $\bv{A}$. Specifically, $\bv{T}$'s eigenvalues are taken as approximate eigenvalues and $\bv{Q}\bv{v}_i$ is taken as an approximate eigenvector for each eigenvector $\bv{v}_i$ of $\bv{T}$. For a \emph{non-symmetric matrix} $\bv{B}$, the Lanczos method can be used to find approximate \emph{singular vectors and values} since these correspond to eigenpairs of $\bv{B}^T\bv{B}$ and $\bv{BB}^T$. 

Substantial literature studies the accuracy of these approximations, both under exact arithmetic and finite precision. 
While addressing the stability of the Rayleigh-Ritz method is beyond the scope of this work, it turns out that, unmodified, our Theorem \ref{mainLemmaFullRuntime} can prove the stability of a related algorithm for the common problem of approximating just the top singular value of a matrix. In particular, for error parameter $\Delta$, our goal is to find some vector $\bv{u}$ such that:
\begin{align}
\label{top_singular_guarantee}
\frac{\|\bv{B}\bv{u}\|}{\|\bv{u}\|} \geq (1-\Delta) \max_{\bv{x}} \frac{\|\bv{B}\bv{x}\|}{\|\bv{x}\|}.
\end{align}
Here $\max_{\bv{x}} \frac{\|\bv{B}\bv{x}\|}{\norm{\bv{x}}} = \norm{\bv{B}} = \smax(\bv{B})$ is $\bv{B}$'s top singular value. In addition to being a fundamental problem in its own right, via deflation techniques, an algorithm for approximating the top singular vector of a matrix can also be used for the important problem of finding a nearly optimal low-rank matrix approximation to $\bv{B}$ \cite{NIPS2016_6507}.

Suppose we have $\bv{B} \in \R^{m\times n}$ that we can multiply on the right by a vector in  $\mv(\bv{B})$ time. In \emph{exact arithmetic} a vector $\bv{u}$ satisfying \eqref{top_singular_guarantee} can be found in time (see e.g. \cite{sachdeva2014faster}):
\begin{align*}
O\left(\mv(\bv{B})\sqrt{1/\Delta}\log\frac{n}{\Delta} + \left(\sqrt{1/\Delta}\log\frac{n}{\Delta} \right)^2\right).
\end{align*}
Note that, unlike other commonly stated bounds for singular vector approximation with the Lanczos method, this runtime does not have a dependence on the gaps between $\bv{B}$'s singular values -- i.e. it does not require a sufficiently large gap to obtain high accuracy.
Since the second term is typically dominated by the first, it is an improvement over the $\tilde{O}(\nnz(\bv{B})/\Delta)$ gap-independent runtime required, for example, by the standard power method. 

We can use Theorem \ref{mainLemmaFullRuntime} to prove, to the best of our knowledge, the first rigorous gap-independent bound  for Lanczos that holds in finite precision. It essentially matches the algorithm's exact arithmetic runtime for singular vector approximation.

\begin{theorem}[Approximating the top singular vector and value]\label{thm:top_sing}
Suppose we are given $\bv{B} \in \R^{m\times n}$ and error parameter $\Delta \leq 1/2$. Let $q = \frac{4}{\Delta}\log\frac{n}{\Delta}$, $B = \log \frac{n}{\Delta}$, and let $\bv{z} \in \{-1,1\}^n$ be chosen randomly by selecting each entry to be $1$ with probability $1/2$ and $-1$ otherwise.  If Algorithm \ref{alg:lanczos} is run with $f(x) = x^q$ on $\bv{B}^T\bv{B}$ and input vector $\bv{z}$ for $O(\sqrt{1/\Delta}\log\frac{n}{\Delta})$ iterations on a computer satisfying Requirement \ref{req1} and Requirement \ref{req2} with precision $\mach = 2^{-\Omega(B)}$ (e.g. a computer with $\Omega(B)$ bits of precision), then with probability $\geq 1/2$, $\bv{y} = \bv{Q}\bv{T}^q\bv{e}_1$ satisfies
\begin{align*}
\frac{\|\bv{B}\bv{y}\|}{\|\bv{y}\|} \geq (1-\Delta) \max_{\bv{x}} \frac{\|\bv{B}\bv{x}\|}{\|\bv{x}\|} = (1-\Delta) \norm{\bv{B}}.
\end{align*}
$\bv{y}$ takes $O(\mv(\bv{B}) \sqrt{1/\Delta}\log\frac{n}{\Delta} + \left(\sqrt{1/\Delta}\log\frac{n}{\Delta} \right)^2 B + \left(\sqrt{1/\Delta}\log\frac{n}{\Delta} \right)B^2)$ time to compute. Note that if this randomized procedure is repeated $O(\log(1/\delta))$ times and the $\bv{y}$ maximizing $\|\bv{B}\bv{y}\|/\|\bv{y}\|$ is selected, then it will satisfy the guarantee with probability $(1-\delta)$. 
\end{theorem}

Before applying our results on function approximation under finite precision, to prove the theorem we first need to argue that, if computed \emph{exactly}, $\bv{\hat y} = \left(\bv{B}^T\bv{B}\right)^q\bv{z}$ provides a good approximate top eigenvector. Doing so amounts to a standard analysis of the power method with a random starting vector, which we include below:
\begin{lemma}[Power method]\label{power_method_exact}
For any $\bv{B} \in \R^{m\times n}$, $\bv{z} \in \{-1,1\}^n$ a random sign vector as described in Theorem \ref{thm:top_sing}, $\Delta \leq 1/2$, and $q = \frac{4}{\Delta}\log\frac{n}{\Delta}$,  with probability $\geq 1/2$, $\bv{\hat y} = \left(\bv{B}^T\bv{B}\right)^q\bv{z}$ satisfies $\norm{\bv{\hat y}} \ge \frac{\norm{\bv{B}^T\bv{B}}^q}{n}$ and:
\begin{align*}
\frac{\|\bv{B}\bv{\hat y}\|}{\|\bv{\hat y}\|} \geq (1-\Delta/2) \max_{\bv{x}} \frac{\|\bv{B}\bv{x}\|}{\|\bv{x}\|}.
\end{align*}
\end{lemma}
\begin{proof}
Let $\bv{B}^T\bv{B} = \bv{V}\bs{\Lambda}\bv{V}^T$ be an eigendecomposition of the PSD matrix $\bv{B}^T\bv{B}$.  $\bv{V}$ is orthonormal with columns $\bv{v}_1,\ldots,\bv{v}_n$ and $\bs{\Lambda}$ is a positive diagonal matrix with entries $\lambda_1 \geq \lambda_2 \geq \ldots \geq \lambda_n$.
\begin{align*}
\|\bv{\hat y}\|^2 &= \bv{z}^T\bv{V}\bs{\Lambda}^{q}\bv{V}^T\bv{z} = \sum_{i=1}^n \lambda_i^{q} \left(\bv{z}^T\bv{v}_i\right)^2 & &\text{and} & \|\bv{B}\bv{\hat y}\|^2 &= \bv{z}^T\bv{V}\bs{\Lambda}^{q+1}\bv{V}^T\bv{z} = \sum_{i=1}^n \lambda_i^{q+1} \left(\bv{z}^T\bv{v}_i\right)^2.
\end{align*}
Let $\lambda_r$ be the smallest eigenvalue with $\lambda_r \geq (1-\Delta/2)\lambda_1$.
Then note that $\|\bv{B}\bv{\hat y}\|^2 = \sum_{i=1}^n \lambda_i^{q+1} \left(\bv{z}^T\bv{v}_i\right)^2 \geq  \sum_{i=1}^r \lambda_i^{q+1} \left(\bv{z}^T\bv{v}_i\right)^2 \geq (1-\Delta/2)\lambda_1\sum_{i=1}^r \lambda_i^{q}\left(\bv{z}^T\bv{v}_i\right)^2$. It follows that:
\begin{align}
\label{big_exact_frac}
\frac{\|\bv{B}\bv{\hat y}\|^2}{\|\bv{\hat y}\|^2}
&\geq (1-\Delta/2)\lambda_1 \frac{\sum_{i=1}^r\lambda_i^{q}\left(\bv{z}^T\bv{v}_i\right)^2}{\sum_{i=1}^r \lambda_i^{q} \left(\bv{z}^T\bv{v}_i\right)^2 + \sum_{i=r+1}^n \lambda_i^{q} \left(\bv{z}^T\bv{v}_i\right)^2}.
\end{align}
We want to show that $\sum_{i=r+1}^n \lambda_i^{q} \left(\bv{z}^T\bv{v}_i\right)^2$ is small in comparison to $\sum_{i=1}^r \lambda_i^{q}\left(\bv{z}^T\bv{v}_i\right)^2 $ so that the entire fraction in \eqref{big_exact_frac} is not much smaller than 1. In fact, we will show that the first quantity is small in comparison to $\lambda_1^{q}\left(\bv{z}^T\bv{v}_i\right)^2$, which is sufficient.

Since $q = \frac{4}{\Delta}\log \frac{n}{\Delta}$ and $\frac{\lambda_i}{\lambda_1} < (1 - \Delta/2)$ for $i \geq r+1$, using the fact that $(1-x)^{1/x} \leq 1/e$ for $x \in [0,1]$, it is not hard to check that:
\begin{align}
\label{pushed_down}
\left(\frac{\lambda_i}{\lambda_1}\right)^{q}  \leq \frac{\Delta^2}{n^2}.
\end{align}
Additionally, since $\bv{z}$ is chosen randomly, with good probability, we don't expect that $\sum_{i=r+1}^n \left(\bv{z}^T\bv{v}_i\right)^2$ will be much larger than $\left(\bv{z}^T\bv{v}_1\right)^2$. 
Since $\|\bv{v}_1\|$ is a unit vector, it must have some entry $i$ with absolute value $\geq \frac{1}{\sqrt{n}}$. For any randomly drawn sign vector $\bv{z}$, let $\bv{\bar z}$ be the same vector, but with the sign of this $i^{\text{th}}$ entry flipped. Since $\bv{v}_1$'s $i^{\text{th}}$ entry has magnitude $\geq \frac{1}{\sqrt{n}}$, it holds that:
\begin{align*}
|\bv{z}^T\bv{v}_1 - \bv{\bar z}^T\bv{v}_1| \geq \frac{2}{\sqrt{n}}.
\end{align*}
Accordingly, for any $\bv{z}$, either $|\bv{z}^T\bv{v}_1|$ or $|\bv{\bar z}^T\bv{v}_1|$ is $\geq \frac{1}{\sqrt{n}}$. We conclude that for a \emph{randomly drawn} $\bv{z}$, with probability $\geq 1/2$,
\begin{align*}
\left(\bv{z}^T\bv{v}_1\right)^2 \geq \frac{1}{n}.
\end{align*}
This immediately gives by our formula for $\norm{\bv{\hat y}}$ our first claim that $\norm{\bv{\hat y}} \ge \lambda_1^q/n$.
Furthermore,
\begin{align*}
\sum_{i=r+1}^n \left(\bv{z}^T\bv{v}_i\right)^2 \leq \sum_{i=1}^n \left(\bv{z}^T\bv{v}_i\right)^2 = \|\bv{V}\bv{z}\|^2 = \|\bv{z}\|^2 = n
\end{align*}
so we can conclude that 
\begin{align}
\label{minimum_dot}
\left(\bv{z}^T\bv{v}_1\right)^2 \geq \frac{1}{n^2}\sum_{i=r+1}^n \left(\bv{z}^T\bv{v}_i\right)^2.
\end{align}

Combining \eqref{pushed_down} and \eqref{minimum_dot} and noting that $\frac{\Delta^2}{n^2} \leq  \frac{\Delta}{2n^2}$ for $\Delta \leq \frac{1}{2}$, we have that:
\begin{align*}
\sum_{i=r+1}^n \lambda_i^{q} \left(\bv{z}^T\bv{v}_i\right)^2 \leq n^2\cdot\frac{\Delta}{2n^2} \lambda_1^{q} \left(\bv{z}^T\bv{v}_1\right)^2 \leq \frac{\Delta}{2} \sum_{i=1}^r \lambda_i^{q} \left(\bv{z}^T\bv{v}_i\right)^2
\end{align*}
Plugging into \eqref{big_exact_frac}  , we have that 
\begin{align*}
\frac{\|\bv{B}\bv{\hat y}\|^2}{\|\bv{\hat y}\|^2} &\geq (1-\Delta/2) \lambda_1\frac{\sum_{i=1}^r\lambda_i^q\left(\bv{z}^T\bv{v}_i\right)^2}{\sum_{i=1}^r \lambda_i^q\left(\bv{z}^T\bv{v}_i\right)^2 + \frac{\Delta}{2}\sum_{i=1}^r  \lambda_i^{q} \left(\bv{z}^T\bv{v}_i\right)^2}  \geq \frac{1-\Delta/2}{1+\Delta/2} \lambda_1.
\end{align*}
Since $\sqrt{\frac{1-\Delta/2}{1+\Delta/2}} \geq (1-\Delta/2)$, we conclude that $\frac{\|\bv{B}\bv{\hat y}\|}{\|\bv{\hat y}\|} \geq  (1-\Delta/2)\sqrt{\lambda_1}$ and the lemma follows since $\sqrt{\lambda_1} = \max_{\bv{x}} \frac{\|\bv{B}\bv{x}\|}{\|\bv{x}\|}$.
\end{proof}
With Lemma \ref{power_method_exact} in place, we prove our main result on approximating the top singular vector.

\begin{proof}[Proof of Theorem \ref{thm:top_sing}]
%
We begin with Theorem 3.3 of \cite{sachdeva2014faster}, which says that for any $q$, there is a polynomial $p$ with degree $k = \left\lceil\sqrt{2q\log(2/\delta)}\right\rceil$ such that:
\begin{align}
\label{monomial approx}
\text{For all }& x\in [-1,1] & |x^q - p(x)| &\leq \delta.
\end{align}
Denote $\lmax \eqdef \lmax(\bv{B}^T\bv{B}) = \norm{\bv{B}^T\bv{B}}$. If we set $\eta = \frac{\lmax}{q}$, after scaling, \eqref{monomial approx} shows that there exists a degree $k$ polynomial $p'(x)$ satisfying: $$|x^q - p'(x)| \le [(1+1/q)\lmax]^q \cdot \delta \le e \delta \lambda_{\max}^q$$
 on the range $[\lmin - \eta, \lmax + \eta]$. 
 
 Set $q = \frac{4}{\Delta} \log \frac{n}{\Delta}$ as in Theorem \ref{power_method_exact} and $k = \Theta \left (\min \left (\sqrt{q\log(2qn/\Delta)},q \right)\right) = O(\sqrt{1/\Delta} \log \frac{n}{\Delta})$. Theorem \ref{mainLemmaFullRuntime} applied with $\bv{A} = \bv{B}^T \bv{B}$, $\eta = \frac{\lmax}{q}$, $\delta_k = O \left (\frac{\Delta}{n^{3/2}k}\lambda_{\max}^q \right)$, and $C = \lambda_{\max}^q$ shows that a computer with $B = \Omega \left (\log \left (\frac{nk \lmax}{\Delta \eta} \right ) \right ) = \Omega \left ( \log \frac{n}{\Delta} \right )$ bits of precision can compute $\bv{y}$ satisfying:
 \begin{align}\label{LanczosEigenApprox}
 \norm{(\bv{B}^T\bv{B})^q \bv{z} - \bv{y}} \le \frac{\Delta \lambda_{\max}^q}{4n^{3/2}} \norm{\bv{z}} \le \frac{\Delta \lambda_{\max}^q}{4n},
 \end{align}
 since $\norm{\bv{z}} = \sqrt{n}$.
 We note that we can multiply by $\bv{B}^T\bv{B}$ accurately, so Requirement \ref{req2} holds as required by Theorem \ref{mainLemmaFullRuntime}. Specifically, it is easy to check that if Requirement \ref{req2} holds for $\bv{B}$ with precision $2^{-cB}$ for some constant $c$, it holds with precision $O(2^{-cB})$ for $\bv{B}^T \bv{B}$.
 
Combining \eqref{LanczosEigenApprox} with the bound of Theorem \ref{power_method_exact} that $\norm{\bv{\hat y}} \ge \frac{\lambda_{\max}^q}{n}$ we first have:
\begin{align*}
\norm{\bv{\hat y} - \bv{y}} \le \frac{\Delta}{4} \norm{\bv{\hat y} }.
\end{align*}
It follows that
\begin{align*}
\frac{\norm{\bv{By}}}{\norm{\bv{y}}} &\geq \frac{\norm{\bv{B\hat y}} - \frac{\Delta}{4} \norm{\bv{B}}\norm{\bv{\hat y} } }{(1+\frac{\Delta}{4})\norm{\bv{\hat y}}} \\
&\geq \left(1-\frac{\Delta}{4}\right)\frac{\norm{\bv{B\hat y}}}{\norm{\bv{\hat y}}} - \frac{\Delta}{4}\norm{\bv{B}} \\
&\geq \left(1-\Delta\right)\norm{\bv{B}},
\end{align*}
where the last step follows from Theorem \ref{power_method_exact}'s claim that $\frac{\norm{\bv{B\hat y}}}{\norm{\bv{\hat y}}} \geq \left(1-\frac{\Delta}{2}\right)\norm{\bv{B}}$.
This proves the theorem, with runtime bounds following from Theorem \ref{mainLemmaFullRuntime}. 
\end{proof}

\medskip \noindent \textbf{Remark. \hspace{.25em}}
As discussed in Section \ref{sec:fpprelim}, computations in Algorithm \ref{alg:lanczos} won't overflow or lose accuracy due to underflow as long as the exponent in our floating point system has at $\Omega(\log\log(knC))$ bits. This is typically a very mild assumption. However, in Theorem \ref{thm:top_sing}, $C = \lmax(\bv{B}^T\bv{B})^q$ so we need $O(\log q + \log\log \norm{\bv{B}}) = \Omega(\log\frac{1}{\Delta})$ bits for our exponent. This may not be a reasonable assumption for some computers -- we might want $\Delta \approx \mach$ and in typical floating point systems fewer bits are allocated for the exponent than for the significand. This issue can be avoided in a number of ways. One simple approach is to instead apply $f(x) = \frac{1}{\lmax(\bv{T})^q}x^q$, which also satisfies the guarantees of Lemma \ref{power_method_exact}. Doing so avoids overflow or problematic underflow in Algorithm \ref{alg:lanczos} as long as we have $\Omega(\log\log(kn))$ exponent bits. It could lead to underflow when applying $f(x)$ to $\bv{T}$'s smaller eigenvalues, but this won't affect the outcome of the theorem -- as discussed in Section \ref{sec:fpprelim} the tiny additive error incurred from underflow when applying $f(\bv{T})$ is swamped by multiplicative error terms.

\subsection{Generic polynomial acceleration}
Our applications to the matrix step function and to approximating the top singular vector in Sections \ref{sub:step}  and \ref{sub:sing} share a common approach: in both cases, Lanczos is used to apply a function that is \emph{itself} a polynomial, one that is simple to describe and evaluate, but has high degree. We then claim that this polynomial can be approximated by a lower degree polynomial, and the number of iterations required by Lanczos depends on this lower degree. In both cases, it is possible to improve a degree $q$ polynomial to degree roughly $\sqrt{q}$ -- a significant gain for the applications. To use the common term from convex optimization, Lanczos provides a way of ``accelerating'' the computation of some high degree matrix polynomials. 

In fact, it turns out that \emph{any} degree $q$ polynomial with bounded coefficients in the monomial basis (or related simple bases) can be accelerated in a similar way to our two examples. 
To see this, we begin with the following result of \cite{frostig2016principal}:
\begin{lemma}[Polynomial Acceleration]
\label{accel_poly_lemma}
Let $p$ be an $O(k)$ degree polynomial that can be written as
\[
p(x) = \sum_{i = 0}^{k} f_i(x) \left( g_i(x)\right)^i
\]
where $f_i(x)$ and $g_i(x)$ are $O(1)$ degree polynomials satisfying $|f_i(x)| \leq a_i$ and $|g_i(x)| \leq 1$ for all $x \in [-1, 1]$. Then, there exists polynomial $q(x)$ of degree $O(\sqrt{k \log(A/\epsilon))}$ where $A = \sum_{i = 0}^{k} a_i$ such that $|p(x) - q(x)| \leq \epsilon$ for all $x \in [-1, 1]$.
\end{lemma}

Setting $g_i(x) = x$ and $f_i(x) = c_i$ for example, lets us accelerate any degree $k$ polynomial $p(x) = c_0 + c_1x + \ldots + c_k x^k$ with bounded coefficients. Lemma \ref{accel_poly_lemma} yields the following result, which generalizes Theorems \ref{thm:step} and \ref{thm:top_sing}:

\begin{theorem}[Application of Accelerated Polynomial]\label{thm:gen_accel}
Consider $\bv{A} \in \R^{n\times n}$ with $\norm{\bv{A}} \le 1$, $\bv{x}\in \R^n$, $\epsilon \leq 1$, and any degree $O(k)$ polynomial $p(x)$ which can be written as in Lemma \ref{accel_poly_lemma}. Let $B = \log \left (\frac{nk}{\epsilon} \right )$. If Algorithm \ref{alg:lanczos} is run with $f(x) = p(x)$ for $q = O(\sqrt{k \log(kA/\epsilon)})$ iterations on a computer satisfying Requirements \ref{req1} and \ref{req2} for precision $\mach = 2^{-\Omega(B)}$ (e.g. a computer using $\Omega(B)$ bits of precision), it outputs $\bv{y}$ satisfying $\norm{p(\bv{A}) \bv{x} - \bv{y}} \le \epsilon A \norm{\bv{x}}$. The total runtime is $O(\mv(\bv{A}) q + q^2 B + qB^2)$.
\end{theorem}
\begin{proof}
The proof follows from Theorem \ref{mainLemmaFullRuntime}. Since $p(x)$ can be written as in Lemma \ref{accel_poly_lemma}, it is not hard to see that $|p(x)| \le A$ for $x \in [-1,1]$. If we set $\eta = \Theta \left ( \min(\|\bv{A}\|,\frac{1}{k^2}) \right)$ we can also claim that $|p(x)| = O(A)$ on $[-1-\eta,1+\eta] \supseteq [\lmin(\bv{A}) - \eta,\lmax(\bv{A}) + \eta]$ (we bound $\eta \le \|\bv{A}\|$ to satisfy the requirement of Theorem \ref{mainLemmaFullRuntime}). This is a consequence of the Markov Brother's inequality. Let $p(x)$ have degree $ck$ and choose  $\eta < \frac{1}{2c^2k^2}$. Suppose, for the sake of contradiction that  $p(x) \ge 2A$ for some $x \in [1,1+\eta]$. Then $p(x)$ must have derivative $> 2Ac^2k^2$ somewhere in $[1,x]$. This would contradict the Markov Brother's inequality. An identical bound can be shown for the range $[-1-\eta,-1]$, overall allowing us to set $C = O(A)$ in applying Theorem \ref{mainLemmaFullRuntime}.

By Lemma \ref{accel_poly_lemma} there is an $O(\sqrt{k \log(kA/\epsilon)})$ polynomial $q(x)$ with $|p(x) - q(x)| = O(\epsilon/k)$ for all 
$x \in [-1,1] \supseteq [ \lmin(\bv{A}),\lmax(\bv{A})]$. We need to extend this approximation guarantee to all $x \in [-1-\eta,1+\eta]$. To do so, we first note that $p(x)-q(x)$ is an $O(k)$ degree polynomial -- we can assume that $q$ has degree at most that of $p$ or else we can just set $q(x) = p(x)$ achieving $\delta_k = 0$. Then, again by using the Markov brother's inequality, since $\eta = O(1/k^2)$ we have $|p(x) - q(x)| = O(\epsilon/k)$ for all $x \in [-1-\eta,1+\eta]$.
We can thus apply Theorem \ref{mainLemmaFullRuntime} with $\delta_k = O(\epsilon/k)$, giving the result.
\end{proof}

\section{Conclusions and future work}

In this work we study the stability of the Lanczos method for approximating matrix functions. We show that the method's finite arithmetic performance for many functions essentially matches the strongest known exact arithmetic bounds. At the same time, for the special case of linear systems, known techniques give finite precision bounds which are much weaker than what is known in exact arithmetic.

The most obvious question we leave open is understanding if our lower bound against Greenbaum's results for approximating $\bv{A}^{-1}\bv{x}$ in fact gives a lower bound on the number of iterations required by the Lanczos and CG algorithms. Alternatively, it is possible that an improved analysis could lead to stronger error bounds for finite precision Lanczos that actually match the guarantees available in exact arithmetic. It seems likely that such an analysis would have to go beyond the view of Lanczos as applying a single near optimal approximating polynomial, and thus could provide significant new insight into the behavior of the algorithm. 

Understanding whether finite precision Lanczos can match the performance of non-uniform approximating polynomials is also interesting beyond the case of positive definite linear systems. For a number of other functions, it is possible to prove stronger bounds than Theorem \ref{exact_lanczos_final_theorem} in exact arithmetic. In some of these cases, including for the matrix exponential, such results can be extended to finite precision in an analogous way to Greenbaum's work on linear systems \cite{golub1994estimates,druskin1998using}. It would be interesting to explore the strength of these bounds for functions besides $1/x$.

Finally, investigating the stability of Lanczos method for other tasks besides of the widely studied problem of eigenvector computation would be interesting. Block variants of Lanczos, or Lanczos with reorthogonalization, have recently been used to give state-of-the-art runtimes for low-rank matrix approximation \cite{rokhlin2009randomized,blanczos}. The analysis of these methods relies on the ability of Lanczos to apply optimal approximating polynomials and understanding the stability of this analysis is an interesting question. It has already been addressed for the closely related but slower block power method \cite{hardt,balcan16a}. 

\section*{Acknowledgements}
We would like to thank Roy Frostig for valuable guidance at all stages of this work and for his close collaboration on the projects that lead to this research. We would also like to thank Michael Cohen for the initial idea behind the lower bound proof and Jon Kelner and Richard Peng for a number of helpful conversations.

\bibliographystyle{alpha}
\bibliography{krylov_analysis}

\appendix

\section{Stability of post-processing for Lanczos}\label{sec:post}
In this section, we show that the final step in Algorithm \ref{alg:lanczos}, computing $\bv{Q}f(\bv{T})\bv{e}_1$, can be performed stably in $\tilde O(k^2)$ time since $\bv{T}$ is a $k \times k$ symmetric tridiagonal matrix. This claim relies on a $\tilde O(k^2)$ time backwards stable algorithm for computing a tridiagonal eigendecomposition, which was developed by Gu and Eisenstat \cite{gu1995divide}. Given an accurate eigendecomposition of $\bv{T}$, we can explicitly compute an approximation to $f(\bv{T})$ and thus to $\bv{Q} f(\bv{T}) \bv{e}_1$. Of course, since small error in computing the eigendecomposition can be amplified, this technique only gives an accurate result when $f$ is sufficiently smooth. In particular, we will show that as long as $f(x)$ is well approximated by a degree $k$ polynomial, then $f(\bv{T})$ can be applied stably. This characterization of smoothness is convenient because our accuracy bounds for Lanczos already depend on the degree to which $f(x)$ can be approximated by a polynomial.

\subsection{Stable symmetric tridiagonal eigendecomposition}
We first characterize the performance of Gu and Eisenstat's divide-and-conquer eigendecomposition algorithm for symmetric tridiagonal $\bv{T}$. We work through the error analysis carefully here, however we refer readers to \cite{gu1995divide} for a full discussion of the computations involved.
\begin{lemma}[Divide-and-Conquer Algorithm of \cite{gu1995divide}]\label{lem:guEisenstat} Given symmetric tridiagonal $\bv{T} \in \R^{k \times k}$ and error parameter $\epsilon$ with $ck^3 \log k \cdot \mach \le \epsilon \le 1/2$ for fixed constant $c$,
there is an algorithm running in $O(k^2 \log \frac{k}{\epsilon} + k \log^2 \frac{k}{\epsilon})$ time on a computer satisfying Requirements \ref{req1} and \ref{req2} with machine precision $\epsilon_{mach}$ which outputs $\bv{\tilde V} \in \mathbb{R}^{k \times k}$ and diagonal $\bs{\tilde \Lambda} \in \mathbb{R}^{k \times k}$ satisfying:
\begin{align*}
\norm{\bv{\tilde V} \bs{\tilde \Lambda} \bv{\tilde V}^T - \bv{T}} \le \epsilon \norm{\bv{T}} \text{ and } \norm{\bv{\tilde V}^T\bv{\tilde V} - \bv{I}} \le \epsilon.
\end{align*}  
\end{lemma}
\begin{proof}
The algorithm of \cite{gu1995divide} is recursive, partitioning $\bv{T}$ into two blocks $\bv{T}_1 \in m \times m$ and $\bv{T}_2 \in (N-m-1) \times (N-m-1)$ where $m = \lfloor k/2 \rfloor$ such that:
\begin{align*}
\bv{T} = \begin{bmatrix}
    \bv{T}_1  & \beta_{m+1} \bv{e}_m & \bv{0} \\
    \beta_{m+1}\bv{e}_m^T    & \alpha_{m+1} & \beta_{m+2} \bv{e}_1^T\\
    \bv{0}    & \beta_{m+2} \bv{e}_1 & \bv{T}_2
\end{bmatrix}.
\end{align*}
Note that $\alpha_{i},\beta_{j}$ are the corresponding entries of $\bv{T}$ in the notation of Algorithm \ref{alg:lanczos}.
Let $\bv{T}_i = \bv{V}_i \bs{\Lambda}_i \bv{V}_i$ be the eigendecomposition of $\bv{T}_i$ for $i=1,2$.
We can see that $\bv{T} = \bv{Z}\bv{H} \bv{Z}^T$
where:
\begin{align*}
\bv{H} &= \begin{bmatrix}
    \alpha_{m+1}   & \beta_{m+1} \bv{l}_1^T & \beta_{m+2} \bv{f}_2^T \\
    \beta_{m+1} \bv{l}_1    & \bs{\Lambda_1} & \bv{0}\\
    \beta_{m+2} \bv{f}_2    & \bv{0} & \bs{\Lambda}_2
\end{bmatrix} && \text{ and }  & \bv{Z} &= \begin{bmatrix}
    \bv{0}    & \bv{V}_1 & \bv{0} \\
    1    & \bv{0} & \bv{0}\\
    \bv{0}    & \bv{0} & \bv{V}_2
\end{bmatrix}.
\end{align*}
Here $\bv{l}_1^T$ is the last row of $\bv{V}_1$ and $\bv{f}_2^T$ is the first row of $\bv{V}_2$. $\bv{H}$ is a symmetric arrowhead matrix and a primary contribution of \cite{gu1995divide} is showing that it can be eigendecomposed stably in $\tilde O(k^2)$ time. Writing the eigendecomposition $\bv{H} = \bv{U}\bs{\Lambda} \bv{U}^T$, the eigendecomposition of $\bv{T}$ is then given by $\bv{T} = \bv{Z}\bv{U}\bs{\Lambda} \bv{U}^T \bv{Z}^T$.

We now proceed with the error analysis of this method.
Assume by induction that for $\bv{T}_i$ we compute an approximate eigendecomposition $\bv{\tilde V}_i \bs{\tilde \Lambda}_i \bv{\tilde V}_i^T$ satisfying:
\begin{align}\label{recursiveAssump}
\bv{\tilde V}_i \bs{\tilde \Lambda}_i \bv{\tilde V}_i^T = \bv{\tilde T}_i
\text{ where }\norm{\bv{T}_i-\bv{\tilde T}_i} \le \delta_T
\text{ and }\norm{\bv{\tilde V}_i^T \bv{\tilde V}_i - \bv{I}} \le \delta_I.
\end{align}
In the base case, $\bv{T}_i$ is just a single entry and so \eqref{recursiveAssump} holds trivially for $\delta_T = \delta_I = 0$.
We define:
\begin{align*}
\bv{\tilde T} = \begin{bmatrix}
    \bv{\tilde T}_1  & \beta_{m+1} \bv{e}_m & \bv{0} \\
    \beta_{m+1}\bv{e}_m^T    & \alpha_{m+1} & \beta_{m+2} \bv{e}_1^T\\
    \bv{0}    & \beta_{m+2} \bv{e}_1 & \bv{\tilde T}_2
\end{bmatrix}
\end{align*}
and have $\norm{\bv{T} - \bv{\tilde T}} \le \delta_T$ by \eqref{recursiveAssump} since $\bv{T}-\bv{\tilde T}$ is block diagonal with blocks $\bv{T}_i - \bv{\tilde T}_i$. Define:
\begin{align*}
\bv{\tilde H}_{\text{exact}} &= \begin{bmatrix}
    \alpha_{m+1}   & \beta_{m+1} \bv{\tilde l}_1^T & \beta_{m+2} \bv{\tilde f}_2^T \\
    \beta_{m+1} \bv{\tilde l}_1    & \bs{\tilde \Lambda_1} & \bv{0}\\
    \beta_{m+2} \bv{\tilde f}_2    & \bv{0} & \bs{\tilde \Lambda}_2
\end{bmatrix} && \text{ and }  & \bv{\tilde Z} &= \begin{bmatrix}
    \bv{0}    & \bv{\tilde V}_1 & \bv{0} \\
    1    & \bv{0} & \bv{0}\\
    \bv{0}    & \bv{0} & \bv{\tilde V}_2
\end{bmatrix}
\end{align*}
where $\bv{\tilde l}_1$ and $\bv{\tilde f}_2$ are the last and first rows of $\bv{\tilde V}_1$ and $\bv{\tilde V}_2$ respectively. Let
$
\bv{\tilde H} = \fl(\bv{\tilde H}_{\text{exact}})
$
be the result of computing $\bv{\tilde H}_{\text{exact}}$ in finite precision. 
By inspection we see that:
\begin{align*}
\bv{\tilde T} - \bv{\tilde Z} \bv{\tilde H}_{\text{exact}} \bv{\tilde Z}^T &= \begin{bmatrix}
    \bv{0}    & \beta_{m+1} (\bv{e}_m - \bv{\tilde V}_1 \bv{\tilde l}_1)   & \bv{0} \\
    \beta_{m+1} (\bv{e}_m - \bv{\tilde V}_1 \bv{\tilde l}_1)^T   & 0 & \beta_{m+2} (\bv{e}_1 - \bv{\tilde V}_2 \bv{\tilde f}_2)^T\\
    \bv{0} & \beta_{m+2} (\bv{e}_1 - \bv{\tilde V}_2 \bv{\tilde f}_2)  & \bv{0}
    \end{bmatrix}.
\end{align*}
By our inductive assumption that $\norm{\bv{\tilde V}_i^T \bv{\tilde V}_i - \bv{I}} \le \delta_I$ \eqref{recursiveAssump}, we have $\norm{\bv{\tilde Z}^T \bv{\tilde Z} - \bv{I}} \le \delta_I$ and further:
\begin{align*}
\norm{\bv{\tilde T} - \bv{\tilde Z}\bv{\tilde H}_{\text{exact}} \bv{\tilde Z}^T } &\le \beta_{m+1} \norm{\bv{e}_m - \bv{\tilde V}_1 \bv{\tilde l}_1} + \beta_{m+2} \norm{\bv{e}_1 - \bv{\tilde V}_2 \bv{\tilde f}_2}\\
&\le \beta_{m+1} \norm{\bv{\tilde V}_1 \bv{\tilde V}_1^T - \bv{I}} \norm{\bv{e}_k} +  \beta_{m+2} \norm{\bv{\tilde V}_2 \bv{\tilde V}_2^T - \bv{I}} \norm{\bv{e}_1}\\
& \le \delta_I (\beta_{m+1} + \beta_{m+2}).
\end{align*}
Then, by the triangle inequality we can loosely bound 
\begin{align}\label{tbound}
\norm{\bv{T} - \bv{\tilde Z} \bv{\tilde H}_{\text{exact}}\bv{\tilde Z}^T } &\le \norm{\bv{T} - \bv{\tilde T}} + \norm{\bv{\tilde T} -\bv{\tilde Z} \bv{\tilde H}_{\text{exact}}\bv{\tilde Z}^T }\nonumber \\ 
&\le \delta_T + \delta_I (\beta_{m+1} + \beta_{m+2}) \nonumber \\
&\le \delta_T + 2\delta_I\norm{\bv{T}}.
\end{align}
Finally, using Requirement \ref{req1}, $\bv{\tilde H}$'s entries are within relative error $\mach$ of the entries in $\bv{\tilde H}_{\text{exact}}$ so
\begin{align}
\label{h_norm_1}
\|\bv{\tilde H} - \bv{\tilde H}_{\text{exact}}\| \leq \norm{\bv{\tilde H} - \bv{\tilde H}_{\text{exact}}}_F \le \mach \norm{\bv{\tilde H}_\text{exact}}_F \le
\mach \sqrt{k}\|\bv{\tilde H}_\text{exact}\|.
\end{align}
Using \eqref{tbound} and submultiplicativity we can obtain a loose bound on $\|\bv{\tilde H}_{\text{exact}}\|$ of
\begin{align}
\label{h_norm_2}
 \norm{\bv{\tilde H}_{\text{exact}}} &\leq \norm{\bv{\tilde Z}^{-1}}^2\left(\norm{\bv{T}} + \delta_T + 2\delta_I\norm{\bv{T}}\right) \nonumber\\
 &\leq 8\norm{\bv{T}}  + 4\delta_T
\end{align}
as long as $\delta_I \leq  1/2$ and so $\norm{\bv{\tilde Z}^{-1}} \le 2$. We finally conclude that:
\begin{align*}
\norm{ \bv{\tilde Z} \bv{\tilde H}_{\text{exact}}\bv{\tilde Z}^T -  \bv{\tilde Z} \bv{\tilde H}\bv{\tilde Z}^T } \leq \norm{ \bv{\tilde Z}}^2  \norm{\bv{\tilde H}_{\text{exact}} -   \bv{\tilde H} }  \leq  18 \sqrt{k}\mach\norm{\bv{T}} + 9\sqrt{k} \mach \delta_T.
\end{align*}
Combined with \eqref{tbound}, we have that:
\begin{align}\label{tbound_final}
\norm{\bv{T} - \bv{\tilde Z} \bv{\tilde H}\bv{\tilde Z}^T } \le (1 + 9\sqrt{k}\mach) \delta_T + 2\delta_I\norm{\bv{T}} + 18\sqrt{k}\mach\norm{\bv{T}}.
\end{align}

We now discuss the eigendecomposition of $\bv{\tilde H}$. Gu and Eisenstat show that for some error parameter $\epsilon_1$ with $\mach \le \epsilon_1 \le \frac{1}{2k}$, it is possible to compute $\wh{\lambda_i}$  satisfying 
\begin{align}\label{eigBound}
\left |\wh{\lambda_i} -  \lambda_i(\bv{\tilde H}) \right |  \le \epsilon_1 k \norm{\bv{\tilde H}}\text{ for all } i
\end{align}
 in $O(k \log^2(1/\epsilon_1))$ time.
To do this they assume the existence of a root finder giving relative accuracy solutions to the roots of the characteristic polynomial of $\bv{\tilde H}$ (its eigenvalues). 
While significant work has studied efficient root finders for these polynomials (see e.g. \cite{bunch1978rank}), we can just use of a simple bisection method that converges to relative accuracy $(1\pm \epsilon_1)$ in $\log(1/\epsilon_1)$ iterations, as it will not significantly affect our final asymptotic runtime bounds. The second $\log(1/\epsilon_1)$ factor in the runtime bound above comes from the use of the Fast Multipole Method \cite{greengard1987fast} to evaluate the characteristic polynomial efficiently at each iteration of the bisection method.

Gu and Eisenstat next show that given $\wh{\lambda}_i$ satisfying \eqref{eigBound}, and letting $\bs{\hat \Lambda} = \diag(\wh{\lambda}_1,...,\wh{\lambda}_k)$, it is possible to compute $\bv{\hat{U}}$ such that $\bv{\hat{H}} \eqdef \bv{\hat{U}} \bs{\hat{\Lambda}} \bv{\hat{U}}^T$ approximates $\bv{\tilde H}$ up to additive error $O( \epsilon_1 k^2 \norm{\bv{\tilde H}})$ on each entry. This gives:
\begin{align}\label{hbound}
\norm {\bv{\tilde H} - \bv{\hat{H}}} = c_1\epsilon_1 k^3 \norm{\bv{\tilde H}}
\end{align}
for some constant $c_1$.
They further show  that $\norm{\bv{\hat{U}}^T \bv{\hat{U}} - \bv{I}} \le c_2\mach k^2$ for constant $c_2$.

Using the fast multipole method, they show how to compute $\bv{\tilde V}$ which approximates $\bv{\tilde Z} \bv{\hat{U}}$ to entrywise relative accuracy $\Theta(\epsilon_1)$ in $O(k^2 \log(1/\epsilon_1))$ time. 
Note that as long as $\delta_I \le 2$ and $\mach \le \frac{1}{ck^2}$ for sufficiently large $c$, then each entry $(\bv{\tilde Z} \bv{\hat U})^{(i,j)}$ is upper bounded by a fixed constant since these matrices are both near orthogonal. Thus $\bv{\tilde V}$ actually approximates $ \bv{\tilde Z}\bv{\hat{U}}$ up to $O(\epsilon_1)$ entrywise additive error and correspondingly
$\bv{\tilde V}^T\bv{\tilde V}$ approximates $ \bv{\hat{U}}^T\bv{\tilde Z} ^T\bv{\tilde Z} \bv{\hat{U}}$ up to entrywise additive error $O(\epsilon_1 k)$. Thus
$\norm{\bv{\tilde V}^T\bv{\tilde V} -  \bv{\hat{U}}^T\bv{\tilde Z}^T \bv{\tilde Z}\bv{\hat{U}}} = O(\epsilon_1k^2)$. 
This gives:
\begin{align}\label{orthCond}
\norm{\bv{\tilde V}^T\bv{\tilde V} - \bv{I}} &\le \norm{\bv{\tilde V}^T\bv{\tilde V} -  \bv{\hat{U}}^T\bv{\tilde Z}^T\bv{\tilde Z}\bv{\hat{U}} } + \norm{\bv{\hat{U}}^T\bv{\tilde Z}^T\bv{\tilde Z}\bv{\hat{U}} - \bv{I}}\nonumber\\
&\le \norm{\bv{\tilde V}^T\bv{\tilde V} - \bv{\hat{U}}^T\bv{\tilde Z}^T\bv{\tilde Z}\bv{\hat{U}}} + \norm{\bv{\hat U}^T \bv{\hat U}-\bv{I}} + \norm{\bv{\hat U}}^2 \norm{\bv{\tilde Z}^T \bv{\tilde Z}-\bv{I}}\nonumber\\ 
&\le c_3\epsilon_1 k^2 + (1+c_2 \mach k^2) \delta_I
\end{align}
for some constants $c_2,c_3$. In the last step we use that $\norm{\bv{\hat U}}^2 \le (1+c_2 \mach k^2)$, since as mentioned $\norm{\bv{\hat{U}}^T \bv{\hat{U}} - \bv{I}} \le c_2\mach k^2$.

$\bv{\tilde V} \bs{\hat{\Lambda}} \bv{\tilde V}^T$ approximates 
$\bv{\tilde Z} \bv{\hat{U}} \bs{\hat{\Lambda}} \bv{\hat{U}}^T \bv{\tilde Z}^T$ to entrywise additive error $O(\epsilon_1 k \norm{\bs{\hat{\Lambda}}}) = O(\epsilon_1 k \norm{\bv{\tilde H}})$ giving:
\begin{align}\label{qBound}
\norm{\bv{\tilde V} \bs{\hat{\Lambda}} \bv{\tilde V}^T - \bv{\tilde Z} \bv{\hat{U}} \bs{\hat{\Lambda}} \bv{\hat{U}}^T \bv{\tilde Z}^T} \le c_4\epsilon_1 k^2 \norm{\bv{\tilde H}}
\end{align}
for some constant $c_4$. 
Combining \eqref{hbound} and \eqref{qBound} with the recursive error bound \eqref{tbound_final}:
\begin{align}\label{normCond}
\norm{\bv{T} &- \bv{\tilde V} \bs{\hat{\Lambda}} \bv{\tilde V}^T} \le \norm{\bv{T} - \bv{\tilde Z} \bv{\tilde H} \bv{\tilde Z}^T } + \norm{\bv{\tilde Z} \bv{\tilde H} \bv{\tilde Z}^T - \bv{\tilde V} \bs{\hat{\Lambda}} \bv{\tilde V}^T}\nonumber\\
&\le (1 + 9\sqrt{k}\mach) \delta_T + 2\norm{\bv{T}}\delta_I + 18\sqrt{k}\mach\norm{\bv{T}} +  \norm{ \bv{\tilde Z} (\bv{\tilde H} - \bv{\hat{U}} \bs{\hat{\Lambda}} \bv{\hat{U}}^T) \bv{\tilde Z}^T} + \norm{ \bv{\tilde Z} \bv{\hat{U}} \bs{\hat{\Lambda}} \bv{\hat{U}}^T \bv{\tilde Z}^T-\bv{\tilde V} \bs{\hat{\Lambda}} \bv{\tilde V}^T}\nonumber\\
&\le (1 + 9\sqrt{k}\mach) \delta_T + 2 \norm{\bv{T}} \delta_I + 18\sqrt{k}\mach\norm{\bv{T}} + c_1 \epsilon_1 k^3 \norm{\bv{\tilde H}}\norm{\bv{\tilde Z}}^2 + c_4 \epsilon_1 k^2 \norm{\bv{\tilde H}}\nonumber\\
&\le (1 + 9\sqrt{k}\mach + c_5\epsilon_1k^3) \delta_T + c_6 (\delta_I + \epsilon_1 k^3)\norm{\bv{T}}   
\end{align}
for fixed constants $c_5,c_6$. In the last bound we use that $\norm{\bv{\tilde H}} = O(\norm{\bv{T}} + \delta_T)$, which follows from combining \eqref{h_norm_1} and \eqref{h_norm_2} and assuming $\mach \le \frac{1}{\sqrt{k}}$. We also use that $\norm{\bv{\tilde Z}} = O(1)$. Both of these bounds hold assuming $\delta_T,\delta_I \le 1/2$. 

Finally, over $\log k$ levels of recursion, as long as we set $\epsilon_1 \le \frac{\epsilon}{ck^3 \log k}$ for sufficiently large $c$, in the end, by \eqref{orthCond} and \eqref{normCond} we will have $\norm{\bv{\tilde V}^T\bv{\tilde V} - \bv{I}} \le \epsilon$ and $\norm{\bv{T}-\bv{\tilde V} \bs{\hat{\Lambda}} \bv{\tilde V}^T} \le \epsilon \norm{\bv{T}}$.

Specifically, by \eqref{orthCond}, $\delta_I$ just grows by a $c_3\epsilon_1 k^2$ additive factor and a $(1+c_2 \mach k^2)$ multiplicative factor at each level. Assuming $\mach < \frac{1}{c k^2 \log k}$, the multiplicative factor, even after compounding over $\log k$ levels, can be bounded by $(1 + c_2 \mach k^2)^{\log k} \le 2$ if $c$ is set large enough. Along with the 
the setting of $\epsilon_1$, this ensures that $\delta_I \le \frac{2c_3\epsilon}{c k}$ at each level. Thus $\delta_I \le \epsilon$ if $c$ is set large enough.

Similarly, 
$\delta_T$ grows according to \eqref{normCond}, increasing by an additive factor of $c_6(\delta_I + \epsilon_1 k^3)\norm{\bv{T}}$ and multiplicative factor of $(1 + 9\sqrt{k}\mach + c_5\epsilon_1k^3)$. By our setting of $\epsilon_1$ and assuming $\mach < \frac{1}{c \sqrt{k}\log k}$, the multiplicative factor can be bounded by $(1 + \frac{9 + c_5}{c\log k})^{\log k} \le 2$ if we set $c$ large enough. Additionally, the additive factor can be upper bounded by $c_6 \left ( \frac{2c_3\epsilon}{ck} + \frac{\epsilon}{c \log k} \right ) \norm{ \bv{T}}$ which is less than $\frac{\epsilon}{\log k} \norm{\bv{T}}$ by any constant factor if $c$ is set large enough. Thus, even when accounting for the multiplicative error and summing over $\log k$ levels, we have $\delta_T \le \epsilon\norm{\bv{T}}$.

Our final runtime bound follows from adding the $O(k^2 \log (1/\epsilon_1))$ cost of computing $\bv{\tilde V}$ to the $O(k \log^2(1/\epsilon_1))$ cost of computing $\bs{\hat{\Lambda}}$ and setting $\epsilon_1 = \frac{\epsilon}{ck^3 \log k}$
 so $\log(1/\epsilon_1) = O(\log(k/\epsilon))$.
\end{proof}

\subsection{Stable function application}
With Lemma \ref{lem:guEisenstat} in place we can complete the analysis of the Lanczos post processing step.
\begin{lemma}[Stable Post-Processing]\label{lem:postProcess}
Suppose we are given a symmetric tridiagonal $\bv{T} \in \R^{k \times k}$, $\bv{x} \in \R^k$ with $\norm{\bv{x}} = 1$\footnote{The theorem also holds with different constant factors when $\norm{\bv{x}} = O(1)$. We prove it for the case when the norm is exactly 1 because our broader analysis only applies it to truly unit norm vectors (i.e. to the basis vector $\bv{e}_1$).}, function $f$, $\eta \ge 0$, and error parameter $\epsilon \le 1/2$ with $ck^3 \log k \cdot \mach  \le \epsilon \le \frac{\eta}{4 \norm{\bv{T}}}$ for sufficiently large constant $c$.
Define $C = \max_{x \in [\lmin(\bv{T}) -\eta,\lmax(\bv{T}) + \eta]} |f(x)|$ and let
\begin{align*}
\delta_{q} =   \min_{\substack{\textnormal{polynomial $p$} \\\textnormal{ with degree $<q$}}} \left(\max_{x \in [\lmin(\bv{T}) - \eta,\lmax(\bv{T}) + \eta]}\left|f(x) - p(x)\right|	\right).
\end{align*}
There is an algorithm running in $O(k^2 \log \frac{k}{\epsilon} + k\log^2  \frac{k}{\epsilon})$ time on a computer satisfying Requirements \ref{req1} and \ref{req2} with relative precision $\epsilon_{mach}$ which returns $\bv{y}$ satisfying:
\begin{align*}
\norm{f(\bv{T})\bv{x} - \bv{y}} &\le 2\delta_{q} + \epsilon \cdot \left ( \frac{16 q^3 C  \norm{\bv{T}}}{\lmax(\bv{T}) - \lmin(\bv{T}) + 2\eta} + 16 C \right ).
\end{align*}
\end{lemma}
\begin{proof}
Assume that we have $\bv{\tilde V}, \bs{\tilde \Lambda}$ satisfying the guarantees of Lemma \ref{lem:guEisenstat} for error $\epsilon$. Let $\bv{\tilde T} = \bv{\tilde V} \bs{\tilde \Lambda} \bv{\tilde V}^T$ and let $\bv{y} = \fl(\bv{\tilde V} f(\bs{\tilde \Lambda}) \bv{\tilde V}^T\bv{x})$
be the result of computing $\bv{\tilde V} f(\bs{\tilde \Lambda}) \bv{\tilde V}^T\bv{x}$ in finite precision. 


We introduce an orthonormal matrix $\bv{\bar V}$ such that $\norm{\bv{\tilde V}-\bv{\bar V}} \le \epsilon$. To see that such a $\bv{\bar V}$ exists, by the condition of Lemma \ref{lem:guEisenstat} we can write $\bv{\tilde V}^T \bv{\tilde V} = \bv{I} + \bs{\Delta}$ for some symmetric $\bs{\Delta}$ with $\norm{\bs{\Delta}} \le \epsilon$. Writing the eigendecomposition $\bs{\Delta} = \bv{Z} \bv{S} \bv{Z}^T$, we have $\bv{\tilde V}^T \bv{\tilde V} = \bv{Z} (\bv{I} + \bv{S}) \bv{Z}^T$. So for some orthonormal matrix $\bv{M}$, $\bv{\tilde V} = \bv{M} (\bv{I}+\bv{S})^{1/2} \bv{Z}^T$ and thus:
\begin{align*}
\bv{\tilde V} = \bv{M}\bv{Z}^T + \bv{M} \bv{\hat S} \bv{Z}^T
\end{align*}
where $\|\bv{\hat S}\| \leq \epsilon$. We can then define the orthonormal matrix $\bv{\bar V} \eqdef \bv{M} \bv{Z}^T$ and have 
\begin{align}
\label{orth_distance}
\norm{\bv{\tilde V} - \bv{\bar V}} \le \epsilon.
\end{align}
 Let $\bv{\bar T} = \bv{\bar V} \bs{\tilde \Lambda}\bv{\bar V}^T$. Using Lemma \ref{lem:guEisenstat}, we have that:
\begin{align*}
\norm{\bv{T}-\bv{\bar T}} &\le \norm{\bv{T}-\bv{\tilde T}} + \norm{\bv{\tilde T}-\bv{\bar T}}\\
&\le \epsilon \norm{\bv{T}} + \norm{\bv{\tilde V} \bs{\tilde \Lambda} \bv{\tilde V}^T - \bv{\bar V} \bs{\tilde \Lambda} \bv{\tilde V}^T} + \norm{\bv{\bar V} \bs{\tilde \Lambda} \bv{\tilde V}^T - \bv{\bar V} \bs{\tilde \Lambda} \bv{\bar V}^T}\\
&\le \epsilon \norm{\bv{T}}  + (\norm{\bv{\tilde V}} + 1)\norm{\bv{\tilde V} - \bv{\bar V}} \norm{\bs{\tilde \Lambda}}\\
&\le 4\epsilon \norm{\bv{T}}.
\end{align*}
 We now have by triangle inequality:
\begin{align}\label{fullFBound}
\norm{f(\bv{T})\bv{x} - \bv{y}} \le \norm{f(\bv{T})\bv{x} - f(\bv{\bar T})\bv{x}} + \norm{f( \bv{\bar T})\bv{x} -  \bv{y}} .
\end{align}
The first norm is small since $\norm{\bv{T}-\bv{\bar T}} \le 4\epsilon \norm{\bv{T}} \le \eta$. Specifically, we can apply Lemma \ref{polyErrorGeneral} in Appendix \ref{sec:perturbation}, which uses  an argument similar to that in Lemma \ref{poly_with_error} to prove that for any $p$ with degree $< q$ and $|p(x)| \le C$ on $[\lmin(\bv{T})-\eta,\lmax(\bv{T})+\eta]$:
\begin{align*}
\norm{p(\bv{T})\bv{x} - p(\bv{\bar T})\bv{x}} \le \frac{4 q^3 C \norm{\bv{T}-\bv{\bar T}}}{\lmax(\bv{T}) - \lmin(\bv{T}) + 2\eta} \le \frac{16 q^3 C \epsilon \norm{\bv{T}}}{\lmax(\bv{T}) - \lmin(\bv{T}) + 2\eta}.
\end{align*}
Using triangle inequality and again that all eigenvalues of $\bv{\bar T}$ lie in $[\lmin(\bv{T})-4\epsilon\norm{\bv{T}},\lmax(\bv{T})+4\epsilon\norm{\bv{T}}]$ and thus in $[\lmin(\bv{T}) -\eta,\lmax(\bv{T})+\eta]$ by our assumption that $\epsilon \le \frac{\eta}{4\norm{\bv{T}}}$ gives:
\begin{align}\label{firstFBound}
\norm{f(\bv{T})\bv{x} - f(\bv{\bar T})\bv{x}} &\le \norm{f(\bv{T})\bv{x} - p(\bv{T})\bv{x}} + \norm{f(\bv{\bar T})\bv{x} - p(\bv{\bar T})\bv{x}} + \norm{p(\bv{T})\bv{x} - p(\bv{\bar T})\bv{x}}\nonumber\\
&\le 2 \delta_{q} + \frac{16 q^3 C \epsilon \norm{\bv{T}}}{\lmax(\bv{T}) - \lmin(\bv{T}) + 2\eta}.
\end{align}

We now bound the second term of \eqref{fullFBound}: $\norm{f( \bv{\bar T})\bv{x} -  \bv{y}}$. First note that, but our assumption in Section \ref{sec:fpprelim} that $f$ can be computed to relative error $\mach$, 
\begin{align}
\label{f_bound}
\|\fl(f(\bs{\tilde \Lambda})) - f(\bs{\tilde \Lambda})\| \leq \mach \norm{f(\bs{\tilde \Lambda})} \leq \mach C
\end{align}
 since $\bs{\tilde \Lambda}$'s eigenvalues lie in $[\lmin(\bv{T}) -\eta,\lmax(\bv{T})+\eta]$. Additionally, by 
Requirement \ref{req2}, for any square matrix $\bv{B}$ and vector $\bv{w}$, there is some matrix $\bv{E}$ such that:
\begin{align}
\label{back_mult}
\fl(\bv{B}\bv{w}) &= \left(\bv{B} + \bv{E}\right)\bv{w} & &\text{and} & \|\bv{E}\| &\leq 2k^{3/2}\mach\norm{\bv{B}}.
\end{align}

Accordingly, we can simply write $\bv{y} = \fl(\bv{\tilde V} f(\bs{\tilde \Lambda}) \bv{\tilde V}^T\bv{x})$ as:
\begin{align*}
\bv{y} =  \left(\bv{\bar V} + \bv{E}_{1}\right)\left(f(\bs{\tilde \Lambda})+ \bv{E}_{2}\right)\left(\bv{\bar V}^T + \bv{E}_{3}\right)\bv{x}.
\end{align*}
We can show that $\|\bv{E}_{1}\|$ and $\norm{\bv{E}_{3}}$ are upper bounded by $\epsilon + 4k^{3/2}\mach$ using \eqref{orth_distance}, \eqref{back_mult}, and the loose bound $\norm{\bv{\tilde V}} \leq 2$. We can show that $\norm{\bv{E}_{2}}$ is upper bounded by $4 \mach Ck^{3/2}$ using \eqref{f_bound} and \eqref{back_mult}. Since $4 \mach k^{3/2} \leq \epsilon$, $\|\bv{E}_{1}\|, \|\bv{E}_{3}\| \leq 2\epsilon$ and $\|\bv{E}_{2}\| \leq \epsilon C$.
Using that $\epsilon \leq 1/2$, it follows that:
\begin{align*}
\norm{f( \bv{\bar T})\bv{x} - \bv{y}} &= \norm{\bv{\bar V}f( \bs{\tilde \Lambda})\bv{\bar V}^T\bv{x} - \bv{y}} \\
&\leq \norm{\bv{E}_1}(C + \epsilon C)(1 + 2\epsilon) + \norm{\bv{E}_3}(C + \epsilon C)(1 + 2\epsilon) + \norm{\bv{E}_2}(1 + 2\epsilon)^2 \\
&\leq 16\epsilon C.
\end{align*}
Plugging this bound and \eqref{firstFBound} into \eqref{fullFBound} gives the lemma.
\end{proof}

\section{Tighter results for linear systems}\label{sec:linsystems_app}

In this section we discuss how bounds on function approximation via Lanczos can be improved for the special case of $f(\bv{A}) = \bv{A}^{-1}$ when $\bv{A}$ is positive definite, both in exact arithmetic and finite precision. 
In particular, we provide a short proof of the exact arithmetic bound presented in  \eqref{introBound2} and discuss Greenbaum's analogous result for finite precision conjugate gradient (Theorem \ref{thm:greenbaum}) in full detail. Ultimately, our lower bound in Section \ref{sec:lower} shows that, while Theorem \ref{thm:greenbaum} is a natural extension of \eqref{introBound2} to finite precision, it actually gives much weaker iteration bounds.

One topic which we do not discuss in depth is that, besides tighter approximation bounds, the computational cost of the Lanczos method can be somewhat improved when solving linear systems. Specifically, it is possible to compute the approximation $\bv{y} = \|\bv{x}\| \bv{Q}\bv{T}^{-1}\bv{e}_1$ ``on-the-fly'', without explicitly storing $\bv{Q}$ or $\bv{T}$. While this does not improve on the asymptotic runtime complexity of Algorithm \ref{alg:lanczos}, it improves the space complexity from $O(kn)$ to simply order $O(n)$ .

Such space-optimized methods yield the popular conjugate gradient algorithm (CG) and its relatives. In fact, Greenbaum's analysis applies to a variant of CG (Algorithm \ref{alg:green_cg}). Like all variants, it computes an approximation to $\bv{A}^{-1}\bv{x}$ that, at least in exact arithmetic, is equivalent to $\|\bv{x}\|\bv{Q}\bv{T}^{-1}\bv{e}_1$, the approximation obtained from the Lanczos method (Algorithm \ref{alg:lanczos}). The finite precision behavior of Greenbaum's conjugate gradient implementation is also very similar to the finite precision behavior of the Lanczos method we study. In fact, her work is based on the same basic results of Paige that we depend on in Section \ref{lanczos_in_finite_precision}.

\subsection{Linear systems in exact arithmetic}
We begin by proving \eqref{introBound2}, showing that the approximation quality of Lanczos in exact arithmetic (Theorem \ref{exact_lanczos_final_theorem}) can be improved when our goal is to approximate $\bv{A}^{-1}\bv{x}$ for positive definite $\bv{A}$. It is not hard to see that an identical bound holds when $\bv{A}$ is positive semidefinite (i.e. may be singular) and $f(\bv{A}) = \bv{A}^+$ is the pseudoinverse. That is, $f(x) = 1/x$ for $x > 0$ and $0$ for $x = 0$. However, we restrict our attention to full rank matrices for simplicity, and since it is for these matrices which Greenbaum's finite precision bounds hold.

\begingroup
\makeatletter
\apptocmd{\thetheorem}{\unless\ifx\protect\@unexpandable@protect\protect \,\,-- Exact Arithmetic\fi}{}{}
\makeatother
\begin{theorem}[Approximate Application of $\bv{A}^{-1}$]\label{exact_lanczos_linear_systems}
Suppose $\bv{Q} \in \R^{n\times k}$, $\bv{T} \in \R^{k\times k}$, $\beta_{k+1}$, and $\bv{q}_{k+1}$ are computed by the Lanczos algorithm (Algorithm \ref{alg:lanczos}), run with exact arithmetic on positive definite $\bv{A} \in \R^{n \times n}$ and $\bv{x} \in \R^n$ for $k \le n$ iterations. Let
\begin{align*}
\bar\delta_{k} =  \min_{\substack{\textnormal{polynomial $p$} \\ \textnormal{w/ degree $<k$}}}\left(\max_{x \in \{ \lambda_1(\bv{A}),\lambda_2(\bv{A}),...,\lambda_n(\bv{A})\}}\left| 1/x - p(x)\right|	\right).
\end{align*}
Then if we approximate $\bv{A}^{-1}\bv{x}$ by $\bv{y}_k = \|\bv{x}\| \bv{Q} \bv{T}^{-1}\bv{e}_1$, we are guaranteed that:
	\begin{align} \label{eq:approx_linear_system_app}
		\|\bv{A}^{-1}\bv{x} - \bv{y}_k\| \leq \sqrt{\kappa(\bv{A})} \bar\delta_{k} \|\bv{x}\|, 
		\end{align}
where $\kappa(\bv{A})$ is the condition number $\lmax(\bv{A})/\lmin(\bv{A})$.
\end{theorem}
\endgroup
\begin{proof}
Let $\bv{A} = \bv{V}\bs{\Lambda}\bv{V}^T$ be $\bv{A}$'s eigendecomposition. Let $\bv{A}^{1/2} = \bv{V}\bs{\Lambda}^{1/2}\bv{V}^T$ and $\bv{A}^{-1/2} = \bv{V}\bs{\Lambda}^{-1/2}\bv{V}^T$. Since $\bv{A}$ is positive semidefinite, $\bs{\Lambda}$ has no negative entries, so both of these matrices are real.
Recall that $\bv{q}_1 = \bv{x}/\|\bv{x}\|$ and consider the minimization problem:
\begin{align*}
\bv{y}^* = \argmin_\bv{y} \|\bv{A}^{-1/2}\bv{q}_1 - \bv{A}^{1/2}\bv{Q}\bv{y}\|.
\end{align*}
This is a standard linear regression problem and is solved by
\begin{align*}
\bv{y}^* =  \left(\bv{Q}^T\bv{A}\bv{Q}\right)^{-1}\left(\bv{Q}^T\bv{A}^{1/2}\right)\bv{A}^{-1/2}\bv{x} =  \left(\bv{Q}^T\bv{A}\bv{Q}\right)^{-1}\bv{Q}^T\bv{q}_1.
\end{align*}
From Claim \ref{claim:exact_lanc_guarantee} we have that $\bv{Q}^T\bv{A}\bv{Q} = \bv{T}$ and that $\bv{q}_1$ is orthogonal to all other columns in $\bv{Q}$. Thus,
\begin{align*}
\bv{y}^* =  \bv{T}^{-1}\bv{e}_1.
\end{align*}
Since $p(\bv{A})\bv{q}_1$ can be written as $\bv{Q}\bv{y}$ for any polynomial $p$ with degree $< k$, it follows that 
\begin{align}
\label{opt_cg_first_bound}
\|\bv{A}^{-1/2}\bv{q}_1 - \bv{A}^{1/2}\bv{Q}\bv{T}^{-1}\bv{e}_1\| \leq \min_{\substack{\text{polynomial $p$} \\ \text{w/ degree $<k$}}} \|\bv{A}^{-1/2}\bv{q}_1 - \bv{A}^{1/2}p(\bv{A})\bv{q}_1\|.
\end{align}
As an aside, if we scale by $\norm{\bv{x}}$ and define the $\bv{A}$-norm $\norm{\bv{v}}_\bv{A} \eqdef \bv{v}^T\bv{A}\bv{v}$, then this can be rewritten:
\begin{align*}
\|\bv{A}^{-1}\bv{x} - \left(\norm{\bv{x}}\bv{Q}\bv{T}^{-1}\bv{e}_1\right)\|_{\bv{A}} \leq \min_{\substack{\text{polynomial $p$} \\ \text{w/ degree $<k$}}} \|\bv{A}^{-1}\bv{x} - p(\bv{A})\bv{x}\|_\bv{A}.
\end{align*}
So, \eqref{opt_cg_first_bound} is equivalent to the perhaps more familiar statement that, ``the Lanczos approximation to $\bv{A}^{-1}\bv{x}$ is optimal with respect to the $\bv{A}$-norm amongst all degree $<k$ matrix polynomials $p(\bv{A})\bv{x}$."
Returning to \eqref{opt_cg_first_bound},
\begin{align}
\|\bv{A}^{-1/2}\bv{q}_1 - \bv{A}^{1/2}\bv{Q}\bv{T}^{-1}\bv{e}_1\| &= \|\bv{A}^{1/2}\left(\bv{A}^{-1}\bv{q}_1 - \bv{Q}\bv{T}^{-1}\bv{e}_1\right)\| \nonumber \\
&\geq \sqrt{\lmin(\bv{A})}\|\bv{A}^{-1}\bv{q}_1 - \bv{Q}\bv{T}^{-1}\bv{e}_1\|. \label{anorm_lower}
\end{align}
Additionally,
\begin{align}
\|\bv{A}^{-1/2}\bv{q}_1 - \bv{A}^{1/2}p(\bv{A})\bv{q}_1\| &= \|\bv{A}^{1/2}\left(\bv{A}^{-1}\bv{q}_1 - p(\bv{A})\bv{q}_1\right)\| \nonumber\\
&\leq \sqrt{\lmax(\bv{A})}\|\bv{A}^{-1}\bv{q}_1 - p(\bv{A})\bv{q}_1\|. \label{anorm_upper}
\end{align}
Plugging \eqref{anorm_lower} and \eqref{anorm_upper} into \eqref{opt_cg_first_bound}, we see that 
\begin{align*}
\|\bv{A}^{-1}\bv{q}_1 - \bv{Q}\bv{T}^{-1}\bv{e}_1\|\ &\leq \sqrt{\kappa(\bv{A})} \min_{\substack{\text{polynomial $p$} \\ \text{w/ degree $<k$}}}  \|\bv{A}^{-1}\bv{q}_1 - p(\bv{A})\bv{q}_1\| \\
&\leq \sqrt{\kappa(\bv{A})} \min_{\substack{\text{polynomial $p$} \\ \text{w/ degree $< k$}}}\left(\max_{x \in \{ \lambda_1(\bv{A}),\ldots,\lambda_n(\bv{A})\}}\left| 1/x - p(x)\right|	\right).
\end{align*}
Theorem \ref{exact_lanczos_linear_systems} follows from scaling both sides by $\|\bv{x}\|$.

\end{proof}

\subsection{Linear systems in finite precision: Greenbaum's Analysis}
As discussed in the Section \ref{intro_linear_systems}, Greenbaum proves a natural extension of Theorem \ref{exact_lanczos_linear_systems} for finite precision computations in \cite{greenbaum1989behavior}. She studies a standard implementation of the conjugate gradient method, included here as Algorithm \ref{alg:green_cg}. This method only requires $O(n)$ space, in contrast to the $O(nk)$ space required by the more general Lanczos method.

\begin{algorithm}
\caption{Conjugate Gradient Method}
{\bf input}: positive semidefinite $\bv{A} \in \mathbb{R}^{n \times n}$, $\#$ of iterations $k$, vector $\bv{x} \in \R^n$\\
{\bf output}: vector $\bv{y}\in \R^{n}$ that approximates $\bv{A}^{-1}\bv{x}$
\begin{algorithmic}[1]
\STATE{$\bv{y} = \bv{0}$, $\bv{r} = \bv{b}$, $\bv{p} = \bv{b}$}
\FOR{$i \in {1,\ldots,k}$} 
\STATE {$\alpha \gets \|\bv{r}\|/\langle \bv{r}, \bv{A}\bv{p}\rangle$}
\STATE {$\bv{y} \gets \bv{y} + \alpha\bv{p}$} 
\STATE {$\bv{r}_{new} \gets \bv{r} - \alpha\bv{A}\bv{p}$} 
\STATE {$\beta \gets -\norm{\bv{r}_{new}}/\norm{\bv{r}}$}
 \IF{$\beta == 0$} \STATE {break loop}\ENDIF
 \STATE {$\bv{p} \gets \bv{r}_{new} - \beta\bv{p}$} 
\STATE {$\bv{r} \gets \bv{r}_{new}$}
 \ENDFOR
\RETURN{$\bv{y}$}
\end{algorithmic}
\label{alg:green_cg}
\end{algorithm}

Although it's not computed explicitly, just as in the Lanczos algorithm, the changing coefficients $\alpha$ and $\beta$ generated by  Algorithm \ref{alg:green_cg} can be used to form a tridiagonal matrix $\bv{T}$. Furthermore, since each $\beta$ shows how the norm of the residual $\bv{r} = \bv{b} - \bv{A}\bv{y}$ decreases over time, $\bv{T}$'s entries uniquely determine the error of the conjugate gradient iteration at any step $k$.

At a high level, Greenbaum shows that the $\bv{T}$ produced by a finite precision CG implementation is \emph{equivalent} to the $\bv{T}$ that would be formed by a running CG on a larger matrix, $\bar{\bv{A}}$, who's eigenvalues all lie in small intervals around the eigenvalues of $\bv{A}$.
She can thus characterize the performance of CG in finite precision on $\bv{A}$ by the performance of CG in exact arithmetic on $\bar{\bv{A}}$. In particular, Theorem 3 in \cite{greenbaum1989behavior} gives:

\begin{theorem}[Theorem 3 in \cite{greenbaum1989behavior}, simplified]\label{greenBaumSimplified}
Let $\bv{y}$ be the output of Algorithm \ref{alg:green_cg} run for $k$ iterations on positive definite $\bv{A} \in \R^{n \times n}$ and $\bv{x} \in \R^n$, with computations performed with $\Omega \left (\log \frac{nk(\|\bv{A}\|+1)}{\min(\eta,\lmin(\bv{A}) )} \right)$ bits of precision. Let $\Delta = \min(\eta,\lmin(\bv{A})/5)$ There exists a matrix $\bar{\bv{A}}$ who's eigenvalues all lie in $\bigcup_{i=1}^n [\lambda_i(\bv{A})-\Delta, \lambda_i(\bv{A})+\Delta]$ and a vector $\bar{\bv{x}}$ with $\norm{\bar{\bv{x}}}_{\bar{\bv{A}}} = \norm{\bv{x}}_{\bv{A}}$ such that, if Algorithm \ref{alg:green_cg} is run for $k$ iterations on $\bar{\bv{A}}$ and $\bar{\bv{x}}$ in \emph{exact arithmetic} to produce $\bar{\bv{y}}$, then:
\begin{align}
\label{coarse_gbaum}
\norm{\bv{A}^{-1}\bv{x} - \bv{y}}_{\bv{A}} \leq 1.2\norm{\bar{\bv{A}}^{-1}\bar{\bv{x}} - \bar{\bv{y}}}_{\bar{\bv{A}}}.
\end{align}
\end{theorem}
Note that for any positive definite $\bv{M}$, and $\bv{z}$ we define $\norm{\bv{z}}_\bv{M} \eqdef \bv{z}^T\bv{M}\bv{z}$. $\bar{\bv{A}}$ is positive definite since $\Delta \le \lmin(\bv{A})/5$.
 
Theorem \ref{greenBaumSimplified} implies the version of Greenbaum's results stated in Theorem \ref{thm:greenbaum}. 
\begin{proof}[Proof of Theorem \ref{thm:greenbaum}]
From our proof of Theorem \ref{exact_lanczos_linear_systems} we have that:
\begin{align*}
 \norm{\bar{\bv{A}}^{-1}\bar{\bv{x}} - \bar{\bv{y}}}_{\bar{\bv{A}}} \leq \sqrt{\lmax(\bar{\bv{A}})} \min_{\substack{\text{polynomial $p$}\\ \text{ with degree $<k$} }} \left ( \max_{x \in \bigcup_{i=1}^n [\lambda_i(\bv{A})-\Delta, \lambda_i(\bv{A})+\Delta]} |p(x)-1/x| \right )\|\bar{\bv{x}}\|.
\end{align*}
Additionally, 
\begin{align*}
\norm{\bv{A}^{-1}\bv{x} - \bv{y}}_{\bv{A}} \geq \sqrt{\lmin(\bv{A})}\norm{\bv{A}^{-1}\bv{x} - \bv{y}}.
\end{align*}
Since $\Delta \leq \lmin(\bv{A})/5$, $\lmax(\bar{\bv{A}}) \leq 1.2\lmax(\bv{A})$. Accordingly, \eqref{coarse_gbaum} simplifies to
\begin{align}
\label{gbaum_second_to_last}
\norm{\bv{A}^{-1}\bv{x} - \bv{y}_{\bv{A}}}  \leq 1.44\sqrt{\kappa(\bv{A})} \min_{\substack{\text{polynomial $p$}\\ \text{ with degree $<k$} }} \left ( \max_{x \in \bigcup_{i=1}^n [\lambda_i(\bv{A})-\Delta, \lambda_i(\bv{A})+\Delta]} |p(x)-1/x| \right )\|\bar{\bv{x}}\|.
\end{align}
Finally,
\begin{align*}
\|\bar{\bv{x}}\| \leq \sqrt{\frac{1}{\lmin(\bar{\bv{A}})}} \|\bar{\bv{x}}\|_{\bar{\bv{A}}} = \sqrt{\frac{1}{\lmin(\bar{\bv{A}})}} \|\bv{x}\|_{\bv{A}} \leq \sqrt{\frac{\lmax{\bv{A}}}{\lmin(\bar{\bv{A}})}} \|\bv{x}\| \leq 1.25 \sqrt{\kappa(\bv{A})} \|\bv{x}\|.
\end{align*}
Plugging in \eqref{gbaum_second_to_last} yields Theorem \ref{thm:greenbaum}.
\end{proof}

\section{General polynomial perturbation bounds}\label{sec:perturbation}
Here we prove that bounded polynomials are generally stable under small perturbations of the input matrix $\bv{A}$. This result is used in proving Lemma \ref{lem:postProcess}, which guarantees that the final step in Algorithm \ref{alg:lanczos} can be performed stably in finite precision. We focus on symmetric perturbations (of symmetric matrices) although the analysis can be extended to general asymmetric perturbations.



\begin{lemma}\label{polyErrorGeneral}  Given symmetric $\bv{A}\in \R^{n \times n}$, symmetric $\bv{E} \in \R^{n \times n}$, $\bv{x}$ with $\norm{\bv{x}} = 1$, and $\eta \ge \norm{\bv{E}}$
if $p$ is a polynomial with degree $<k$ and $|p(x)| \leq C$ for all $x\in [\lmin(\bv{A}) - \eta,\lmax(\bv{A})+\eta]$ then:
\begin{align}
\label{smoothBoundGeneral}
\|p(\bv{A}-\bv{E})\bv{x}- p(\bv{A})\bv{x} \| \leq \frac{4 k^3 C}{\lmax(\bv{A}) - \lmin(\bv{A}) + 2\eta} \norm{\bv{E}}.
\end{align}
\end{lemma}

As in Lemma \ref{poly_with_error}, the result follows by writing $p$ in the Chebyshev basis. Letting $T_i$ be the $i^{th}$
Chebyshev polynomial of the first kind (see definition in \eqref{recurrence}), define $\rmax = \lmax(\bv{A}) + \eta$, $\rmin = \lmin(\bv{A}) -\eta$, and:
%
\begin{align}\label{tprimeGeneral}
	\delta &= \frac{2}{\rmax - \rmin} &&\text{and} 	& \ol{T}_i(x) &= T_i\left(\delta(x-\rmin) - 1\right).
\end{align}
In this way, we have $\ol{T}(\rmin) = T(-1) $ and $\ol{T}(\rmax) = T(1)$. 

\begin{lemma}\label{chebyErrorGeneral} Given symmetric $\bv{A}\in \R^{n \times n}$, symmetric $\bv{E} \in \R^{n \times n}$, $\bv{x}$ with $\norm{\bv{x}} = 1$, and $\eta \ge \norm{\bv{E}}$, for all $i < k$,
\begin{align}
\label{chebyErrorGeneralEq}
\|\ol{T}_i(\bv{A} - \bv{E})\bv{x} - \ol{T}_i(\bv{A})\bv{x}\| \leq \frac{4i^2}{\rmax-\rmin}  \|\bv{E}\|, 
\end{align}
where $\rmax = \lmax(\bv{A}) + \eta$ and $\rmin = \lmin(\bv{A}) -\eta$, as above.
\end{lemma}
\begin{proof}

Define:
$
\barA \eqdef \delta(\bv{A} - \rmin\bv{I}) - \bv{I}
$
so \eqref{chebyErrorGeneralEq} is equivalent to:
\begin{align}
\label{shifted_cheby_poly_error_bound_general}
\|T_i(\barA - \delta\bv{E})\bv{x} - T_i(\barA)\bv{x}\| \leq \frac{2i^2}{\rmax-\rmin} \|\bv{E}\|.
\end{align}
We use the following notation, mirroring that of Lemma \ref{cheby_poly_with_error}:
\begin{align*}
	\bv{t}_i 		&\eqdef T_i(\barA)\bv{x}, \hspace{3em}
	\bv{\tilde t}_i	\eqdef T_i(\barA - \delta\bv{E})\bv{x}, \\
	\bv{d}_i 		&\eqdef \bv{t}_i - \bv{\tilde t}_i, \hspace{3em}
	\bs{\xi}_i 		\eqdef  \delta\bv{E}\bv{\tilde t}_{i-1}.
\end{align*}
Obtaining \eqref{shifted_cheby_poly_error_bound_general} requires showing $\|\bv{d}_i\| \le \frac{2i^2 \|\bv{E}\|}{\rmax-\rmin}$. From the Chebyshev recurrence \eqref{recurrence} for all $i \ge 2$:
\begin{align}
	\bv{d}_i 	&= \left( 2\barA\bv{t}_{i-1} - \bv{t}_{i-2}\right) - \left( 2\left( \barA - \delta\bv{E} \right)\bv{\tilde t}_{i-1} - \bv{\tilde t}_{i-2} \right) \nonumber\\
				&=2\bs{\xi}_i + \left(2\barA\bv{d}_{i-1} - \bv{d}_{i-2}\right) \label{eq:error_recur_general}.
\end{align}
Let $U_i$ be the $i^{th}$ Chebyshevy polynomial of the second kind (see definition in \eqref{secondRecurrence}). Using the same argument as used to show \eqref{eq:sum_of_type_2s} in the proof of Lemma \ref{poly_with_error} we have, for any $i \ge 0$,
\begin{align} \label{eq:sum_of_type_2s_general}
	\bv{d}_i = U_{i-1}(\barA)\bs{\xi}_1 + 2\sum_{j=2}^{i} U_{i-j}(\barA)\bs{\xi}_j.
\end{align}
where for convenience we define $U_{k}(x) = 0$ for any $k < 0$.
It follows that:
\begin{align*}
	\|\bv{d}_i \| \leq 2\sum_{j=1}^{i} \|U_{i-j}(\barA)\|\|\bs{\xi}_j\|.
\end{align*}
Since all of $\barA$'s eigenvalues lie in $[-1,1]$ and for values in this range $U_k(x) \leq k+1$ \cite{chebyBook}:
\begin{align}\label{dSumBound_general}
	\| \bv{d}_i \| \leq 2\sum_{j=1}^{i} (i-j+1)\|\bs{\xi}_j\| \leq 2\sum_{j=1}^{i} i\|\bs{\xi}_j\|.
\end{align}

We finally bound $\|\bs{\xi}_j\| = \|\delta\bv{E}\bv{\tilde t}_{j-1}\|$ by:
\begin{align}\label{eq:final_norm_bound_general}
\|\bs{\xi}_j\|  &\le \delta \norm{\bv{E}} \norm{T_{j-1}(\barA-\delta \bv{E})}.
\end{align}
By our requirement that $\eta \ge \norm{\bv{E}}$, $\barA-\delta \bv{E} = \delta\left(\bv{A} -\bv{E} - \rmin\bv{I}\right) - \bv{I}$ is a symmetric matrix with all eigenvalues in $[-1,1]$. 
Therefore,
\begin{align}\label{eq:cheby_upper_bound_general}
\left\|T_{j-1}(\barA-\delta \bv{E})\right\| \le \max_{x\in \left[-1, 1\right]} \left |T_{j-1}(x) \right | \le 1.
\end{align}
%
Plugging this back into \eqref{eq:final_norm_bound_general}, we have
$
\norm{\bs{\xi}_j} \le \delta \norm{\bv{E}}
$
and plugging into \eqref{dSumBound_general}, $\norm{\bv{d}_i} \le \frac{4 i^2}{\rmax-\rmin} \norm{\bv{E}}$.
\end{proof}
Using Lemma \ref{chebyErrorGeneral}, we can prove Lemma \ref{polyErrorGeneral}. The argument is omitted, as it is identical to the proof of Lemma \ref{poly_with_error}.

\section{Potential function proof of Chebyshev polynomial optimality}\label{tighterLowerBound}

Our lower bound in Section \ref{sec:lower} (specifically Lemma \ref{lbIntermediate}) shows that there is a matrix $\bv{A}$ with condition number $ \kappa$, such that any polynomial which has $p(0) = 1$ and $|p(x)| < 1/3$ for $x$ in a small range around each of $\bv{A}$'s eigenvalues must have degree $\Omega(\kappa^c)$ for some constant $1/5 \le c \le 1/2$. We do not carefully optimize the value of $c$, as any $\poly(\kappa)$ iteration bound demonstrates that Greenbaum's finite precision bound (Theorem \ref{thm:greenbaum}) is significantly weaker than the exact arithmetic bound of \eqref{introBound2} (proven in Theorem \ref{exact_lanczos_linear_systems}).

Here we demonstrate that a continuous version of our potential function argument can prove that 
any polynomial which is small on the \emph{entire} interval $[1/\kappa, 1]$ but has $p(0) = 1$ must have degree $ \Omega(\sqrt{\kappa}/\log \kappa)$. This matches the optimal bound achievable via the Chebyshev polynomials up to an $O(\log \kappa)$. While there are many alternative proofs of this fact, we include the argument because we believe that a more careful discretization could lead to a tighter version of Lemma \ref{lbIntermediate}.

\begin{theorem}
There exists a fixed constant $c > 0$ such that for any $\kappa > 1$, any polynomial $p$ with $p(0) = 1$ and $|p(x)| \le 1/3$ for all $x \in [1/\kappa, 1]$ must have degree $k \ge c \sqrt{\kappa}/\log \kappa$.
\end{theorem}
\begin{proof}
As in Section \ref{sec:lower}, we can again assume that all roots of $p(x)$, $r_1, \ldots, r_k$, are real and lie in $[1/\kappa, 1]$. Moving a root that is outside this range to the nearest boundary of the range or taking the real part of an imaginary root can only decrease the magnitude of $p$ at $x \in [1/\kappa, 1]$.
We write, using that $p(0) = 1$:
\begin{align*}
g(x) = \log(|p(x)|) = \sum_{i=1}^k \log|1 - x/r_i|.
\end{align*}

We want to lower bound $\max_{x\in[1/\kappa,1]} g(x)$. To do so, we first note that for any positive weight function $w(x)$:
\begin{align}\label{maxBound}
\max_{x\in[1/\kappa,1]} g(x) \geq \frac{\int_{1/\kappa}^1 w(x) g(x) dx}{\int_{1/\kappa}^1 w(x) dx}.
\end{align}
I.e. any weighted average lower bounds the maximum of a function.
We also note that:
\begin{align}
\label{cont_lower_bound_main}
\frac{1}{k} \frac{\int_{1/\kappa}^1 w(x) g(x) dx}{\int_{1/\kappa}^1 w(x) dx} \geq  \min_{r\in [1/\kappa,1]} \frac{\int_{1/\kappa}^1 w(x) \log|1 - x/r| dx}{\int_{1/\kappa}^1 w(x) dx},
\end{align}
So we focus on bounding this second quantity.
We set:
\begin{align*}
w(x) = \frac{1}{x^{1.5}}.
\end{align*}
Under this weighting, the denominator in \eqref{cont_lower_bound_main} evaluates to:
\begin{align}
\label{eq_denom}
\int_{1/\kappa}^1 \frac{1}{x^{1.5}} dx = \left.\frac{-2}{\sqrt{x}} \right\rvert_{1/\kappa}^1 = 2\left(\sqrt{\kappa} - 1\right).
\end{align}
We now consider the numerator, which we denote as $N$.
\begin{align}
\label{numerator_main}
N = \int_{1/\kappa}^1\frac{1}{x^{1.5}} \log|1 - x/r|dx = \int_{1/\kappa}^r \frac{1}{x^{1.5}} \log(1 - x/r) dx + \int_{r}^1 \frac{1}{x^{1.5}} \log(x/r - 1) dx.
\end{align}
We will ultimately show that $\min_{r\in [1/\kappa,1]}\left[ N \right] \geq -c \log \kappa$ for some fixed constant $c $. Then when we divide by the denominator of $2\left(\sqrt{\kappa} - 1\right)$ computed in \eqref{eq_denom}, we see that any root makes very little progress -- just $O(\frac{\log \kappa}{\sqrt{\kappa}})$ -- towards decreasing the weighted average of $g(x)$. Ultimately, we will thus require $k = O(\sqrt{\kappa}/\log \kappa)$ roots for $|p(x)|$ to be bounded for all $x \in [1/\kappa,1]$ (formally shown by applying \eqref{cont_lower_bound_main}).
The first term of the split integral in \eqref{numerator_main} can be evaluated as:
\begin{align}
\label{cont_part_1}
 \int_{1/\kappa}^r \frac{1}{x^{1.5}} \log(1 - x/r) dx &\geq \int_{0}^r \frac{1}{x^{1.5}} \log(1 - x/r) dx \nonumber\\
 &= - 2\left. \log\left(\frac{\left(1 + \sqrt{\frac{x}{r}}\right)^{1/\sqrt{r} + 1/\sqrt{x}}}{\left(1 - \sqrt{\frac{x}{r}}\right)^{1/\sqrt{r} - 1/\sqrt{x}}} \right) \right \rvert_{0}^r \nonumber\\
 &= -2\log\left(2^{2/\sqrt{r}}\right).
\end{align}
The first inequality follows because $r \geq \frac{1}{\kappa}$.
The second term can be evaluated as:
\begin{align}
\label{cont_part_2}
\int_{r}^{1} \frac{1}{x^{1.5}} \log(x/r - 1) dx &= -2 \left. \log\left(\frac{\left(1 + \sqrt{\frac{x}{r}}\right)^{1/\sqrt{r} + 1/\sqrt{x}}}{\left(\sqrt{\frac{x}{r}} - 1\right)^{1/\sqrt{r} - 1/\sqrt{x}}} \right) \right \rvert_{r}^{1} \nonumber\\
 &= -2\log\left(\frac{\left(1 + \sqrt{\frac{1}{r}}\right)^{1/\sqrt{r} +1}}{\left(\sqrt{\frac{1}{r}} - 1\right)^{1/\sqrt{r} - 1}} \right)  + 2\log\left(2^{2/\sqrt{r}}\right) \nonumber \\
 & \geq -2\log\left(8/r\right)  + 2\log\left(2^{2/\sqrt{r}}\right)
 \end{align}
The  inequality follows from noting that $\frac{1}{\sqrt{r}} \geq 1$ and
\begin{align*}
\text{For all } x &\geq 1, &
\frac{(1+x)^{x + 1}}{(x-1)^{x-1}} \leq 8x^2.
\end{align*}
Plugging \eqref{cont_part_1} and \eqref{cont_part_2} in \eqref{numerator_main} we have:
\begin{align*}
N \geq -2\log(8/r) = - O(\log \kappa).
\end{align*}
Returning to \eqref{maxBound}  and \eqref{cont_lower_bound_main}, we can combine this with \eqref{eq_denom} (and an assumption that $\kappa \geq 2$) to conclude that:
\begin{align*}
\max_{x\in[1/\kappa,1]} g(x)  \geq - ck\log(\kappa)/\sqrt{\kappa}
\end{align*}
 for some fixed constant $c > 0$.
Thus, since $g(x) \eqdef \log(|p(x)|)$,
\begin{align*}
\max_{x\in[1/\kappa,1]} |p(x)|  \geq e^{-ck\log(\kappa)/\sqrt{\kappa}}
\end{align*}
and it follows that for $|p(x)|$ to be small (e.g. $\leq 1/3$) for all $x \in [1/\kappa,1]$ we need $k \geq \delta \sqrt{\kappa}/\log(\kappa)$ for some fixed constant $\delta > 0$.
\end{proof}

\section{Other omitted proofs}\label{omitted_proofs}
\begingroup
\makeatletter
\apptocmd{\thetheorem}{\unless\ifx\protect\@unexpandable@protect\protect \,\,-- Exact Arithmetic\fi}{}{}
\makeatother
\begin{repclaim}{claim:exact_lanc_guarantee}[Lanczos Output Guarantee]
Run for $k\leq n$ iterations using exact arithmetic operations, the Lanczos algorithm (Algorithm \ref{alg:lanczos}) computes $\bv{Q} \in \R^{n\times k}$, an additional column vector $\bv{q}_{k+1} \in \R^n$, 
a scalar $\beta_{k+1}$, and a symmetric tridiagonal matrix $\bv{T} \in \R^{k\times k}$ such that:
	\begin{align}
		\bv{A}\bv{Q} = \bv{Q}\bv{T} + \beta_{k+1}\bv{q}_{k+1}\bv{e}_k^T, \tag{\ref{eq:exact_lanc_guarantee}}
	\end{align}
and
	\begin{align}
		\left[\bv{Q},\bv{q}_{k+1}\right]^T\left[\bv{Q},\bv{q}_{k+1}\right] = \bv{I}. \tag{\ref{spectral_norm_bound_on_q}}
	\end{align}
Together \eqref{eq:exact_lanc_guarantee} and \eqref{spectral_norm_bound_on_q} also imply that:
\begin{align}
\lmin(\bv{T}) &\geq \lmin(\bv{A}) &&\text{and} &\lmax(\bv{T}) &\leq  \lmax(\bv{A}). \tag{\ref{exact_arith_eig_bound}}
\end{align}
When run for $k \ge n$ iterations, the algorithm terminates at the $n^{th}$ iteration with $\beta_{n+1} = 0$.
\end{repclaim}
\endgroup

\begin{proof}
\eqref{eq:exact_lanc_guarantee} is not hard to check directly by examining Algorithm \ref{alg:lanczos} (see \cite{parlett1979lanczos} for a full exposition). For \eqref{spectral_norm_bound_on_q}, note that by Step 10 in Algorithm \ref{alg:lanczos}, each $\|\bv{q}_i\| = 1$. So we just need to show that $\left[\bv{q}_1,\ldots,\bv{q}_k\right]$ are mutually orthogonal.

 Assume by induction that $[\bv{q}_1, \ldots, \bv{q}_{k-1}]$ are mutually orthogonal. Now consider the value $\beta_{k-1}$. $\beta_{k-1} = \|\bv{q}_{k-1}\|$ before $\bv{q}_{k-1}$ is normalized in Step 10. Thus by the computation of $\bv{q}_{k-1}$ in Steps 3-5 we have: $\beta_{k-1} = \bv{q}_{k-1}^T\left(\bv{A}\bv{q}_{k-2} - \beta_{k-2}\bv{q}_{k-3}  - \alpha_{k-2}\bv{q}_{k-2}\right)$. By the induction hypothesis, this reduces to $\beta_{k-1} = \left(\bv{q}_{k-1}^T\bv{A}\right)\bv{q}_{k-2}$. 
 
 The above relation should make Steps 3-5 of Algorithm \ref{alg:lanczos} more clear. We set $\bv{q}_k$ to equal $\bv{A}\bv{q}_{k-1}$, and explicitly orthogonalize against $\bv{q}_{k-2}$ (Step 3) and then $\bv{q}_{k-1}$ (Step 4-5). So we have $\bv{q}_k^T\bv{q}_{k-1} = 0$ and $\bv{q}_k^T\bv{q}_{k-2} = 0$. It remains to consider $[\bv{q}_1, \ldots, \bv{q}_{k-3}]$. For $j \geq 3$, $\bv{q}_{k-j}^T\bv{A}$ lies in the span of $[\bv{q}_1, \ldots, \bv{q}_{k-2}]$ and then applying the inductive hypothesis we see that $\bv{q}_{k-j}^T\bv{A}\bv{q}_{k-1} = 0$. So there is no need to explicitly orthogonalize $\bv{A}\bv{q}_{k-1}$ against these vectors when we generate $\bv{q}_k$.
 It follows that $\bv{q}_k$ is orthogonal to \emph{all} vectors in $[\bv{q}_1, \ldots, \bv{q}_{k-1}]$, which proves \eqref{spectral_norm_bound_on_q}. 
 
 If $k = n$, the same argument shows that $\bv{q}_{n+1}$ is orthogonal to $[\bv{q}_1,...,\bv{q}_n]$, which implies that $\bv{q}_{n+1} = \bv{0}$ and so the loop terminates at step (7-8). We thus have $\beta_{n+1} \bv{q}_{n+1}\bv{e}_n^T = \bv{0}$ and $\bv{AQ} = \bv{QT}$. 
 
 Finally, \eqref{exact_arith_eig_bound} follows from the Courant-Fischer min-max principle. In particular, if we multiply \eqref{eq:exact_lanc_guarantee} on the left by $\bv{Q}^T$ and note from \eqref{spectral_norm_bound_on_q} that $\bv{Q}^T\bv{Q}= \bv{I}$ and $\bv{Q}^T\bv{q}_{k+1} = \bv{0}$, then as desired,
 \begin{align*}
 \bv{T} = \bv{Q}^T\bv{A}\bv{Q}.
 \end{align*}
Then since $\|\bv{Q}\bv{y}\| = 1$ for any unit norm $\bv{y}$,
 \begin{align*}
\lmin(\bv{A})  = \min_{\bv{x}\in \R^n :\|\bv{x}\| = 1} \bv{x}^T\bv{A}\bv{x} \hspace{1em}&\leq\hspace{1em} \min_{\bv{y} \in \R^k:\|\bv{y}\| = 1} \bv{y}^T \bv{Q}^T\bv{A}\bv{Q}\bv{y} = \lmin(\bv{T})\text{  and} \\
\lmax(\bv{A})  = \max_{\bv{x}\in \R^n:\|\bv{x}\| = 1} \bv{x}^T\bv{A}\bv{x} \hspace{1em}&\geq\hspace{1em} \max_{\bv{y} \in \R^k:\|\bv{y}\| = 1} \bv{y}^T \bv{Q}^T\bv{A}\bv{Q}\bv{y} = \lmax(\bv{T}).
 \end{align*}
\end{proof}

\end{document}